\documentclass[12pt]{article}

\usepackage{amsfonts}
\usepackage{amstext}
\usepackage{comment}
\usepackage{mathtools}
\usepackage{amsmath}
\allowdisplaybreaks
\newcommand\numberthis{\addtocounter{equation}{1}\tag{\theequation}}
\usepackage{amsthm}
\usepackage{amssymb}
\usepackage{enumitem}
\usepackage[round]{natbib}
\usepackage{hyperref}
\hypersetup{
	colorlinks=true,
	linkcolor=blue,
	citecolor=blue,
	filecolor=magenta,      
	urlcolor=cyan,
}
\usepackage{setspace}
\usepackage{fancyhdr}
\usepackage{float}
\usepackage[dvips]{color}
\usepackage[pdftex]{graphicx}
\numberwithin{equation}{section}
\numberwithin{table}{section}
\usepackage[dvipsnames]{xcolor}
\usepackage{multirow}
\usepackage{subcaption}
\usepackage{subfloat}
\usepackage{commath}
\usepackage{bbm}
\usepackage{longtable}
\usepackage{tcolorbox}
\usepackage{algorithm}
\usepackage[noend]{algpseudocode}
\usepackage{fullpage}
\usepackage{authblk}
\usepackage{marginnote}
\usepackage[
    protrusion=true,
    activate={true,nocompatibility},
    final,
    tracking=true,
    kerning=true,
    spacing=true,
    factor=1100]{microtype}
\SetTracking{encoding={*}, shape=sc}{40}
\allowdisplaybreaks
\usepackage[colorinlistoftodos]{todonotes}\usepackage{tagging}
\usepackage{thmtools}

\graphicspath{{figures/}, {other_figures/}}

\definecolor{tolblue}{HTML}{4477AA}

\DeclareMathOperator*{\argmin}{arg\,min}

\renewcommand{\P}{\mathbb{P}}
\newcommand{\E}{\mathbb{E}}
\newcommand{\V}{\mathbb{V}}
\newcommand{\g}{}

\newcommand{\R}{\mathbb{R}}

\newcommand{\NN}{\mathcal N}

\newcommand{\DD}{\mathcal D}
\newcommand{\I}{\mathbbm{1}}
\newcommand{\rmise}{\sqrt{\text{MISE}}}
\newcommand{\rise}{\sqrt{\text{ISE}}}
\newcommand{\1}{\mathbbm{1}}
\renewcommand{\d}{\mathrm{d}}
\newcommand{\ind}{\perp\!\!\!\!\perp} 
\newcommand{\diam}{\operatorname{diam}}

\newtheorem{cor}{Corollary}

\newtheorem{lemma}{Lemma}
\newtheorem{theorem}{Theorem}
\newtheorem{proposition}{Proposition}
\theoremstyle{remark}
\newtheorem{example}{Example}
\theoremstyle{definition}
\newtheorem{remark}{Remark}
\newtheorem{assumption}{Assumption}

\newcommand\CoAuthorMark{\footnotemark[\arabic{footnote}]}

\title{
\bf
Extremal Random Forests
}

\author[1, 2]{Nicola Gnecco\thanks{Authors contributed equally.}}
\author[2, 3]{Edossa Merga Terefe\protect\CoAuthorMark}
\author[2]{Sebastian Engelke}
\affil[1]{Department of Mathematical Sciences, University of Copenhagen, Denmark}
\affil[2]{Research Center for Statistics, University of Geneva, Switzerland}
\affil[3]{Statistics Department, Hawassa University, Ethiopia}

\date{\today}

\usetag{}

\begin{document}

\maketitle

\section*{Abstract}
Classical methods for quantile regression fail in cases where the quantile of interest is extreme and only few or no training data points exceed it. Asymptotic results from extreme value theory can be used to extrapolate beyond the range of the data, and several approaches exist that use linear regression, kernel methods or generalized additive models.
Most of these methods break down if the predictor space has more than a few dimensions or if the regression function of extreme quantiles is complex. We propose a method for extreme quantile regression that combines the flexibility of random forests with the theory of extrapolation. Our extremal random forest (ERF) estimates the parameters of a generalized Pareto distribution, conditional on the predictor vector, by maximizing a local likelihood with weights extracted from a quantile random forest.  We penalize the shape parameter in this likelihood to regularize its variability in the predictor space. Under general domain of attraction conditions, we show consistency of the estimated parameters in both the unpenalized and penalized case.
Simulation studies show that our ERF outperforms both classical quantile regression methods and existing regression approaches from extreme value theory.
We apply our methodology to extreme quantile prediction for U.S.~wage data.

{\it Keywords:}   extreme quantiles; local likelihood estimation; quantile regression; random forests; threshold exceedances.

\section{Introduction} \label{sec:intro}

Quantile regression is a well-established technique to model statistical quantities that go beyond the conditional expectation that is used for standard regression analysis \citep{koen1978}.
This is particularly valuable in applications such as economics, survival analysis, medicine, and finance~\citep{angrist2009b, yang1999censored, heagerty1999semiparametric, taylor1999quantile,keming2003},
where one needs to model the heteroscedasticity of the response or conditional quantiles such as the median.

In this paper, we consider the problem of estimating high conditional quantiles of a response variable $Y \in \R$ given a set of predictors $X \in \R^p$ in large dimensions, an important task in risk
assessment for rare events \citep{chernozhukov2005}.
For a fixed predictor value $x$, define $Q_x(\tau)$ as the quantile at level $\tau\in(0, 1)$ of the conditional distribution of $Y \mid X = x$.
We are interested in estimating extreme quantiles where $\tau \approx 1$ is close to one. This estimation problem exhibits two fundamental challenges that are illustrated in Figure~\ref{fig:intro}, which shows a simulation similar to \citet[][Figure~2]{athe2019}. The predictor space has $p=40$ dimensions, and only the first variable $X_1$ has a signal corresponding to a scale shift in $Y$; see Example~\ref{ex:gen-mod} in Section~\ref{sec:algo} for details.

The first challenge in estimating $Q_x(\tau)$ relates to the fact that for an extreme probability level, say $\tau = 0.9995$ as in Figure~\ref{fig:intro}, there are typically only a few or no observations in the sample that exceed the corresponding conditional $\tau$-quantiles. Indeed, for a sample of size $n$, the expected number of exceedances above the conditional $\tau$-quantile is $n(1-\tau)$, which becomes smaller than one if $\tau > 1 - 1/n$. Therefore, using an empirical estimator based on quantile loss leads to a large bias and variance.
A second challenge stems from the possibly large dimension of the
predictor space $\mathbb R^p$, where there might be no training observations close to $\g x$; note that the Figure~\ref{fig:intro} only shows the first of the 40 dimensions of $\g X$. Too simple regression models may then introduce additional bias.

The first challenge can be addressed by relying on tail approximations motivated by extreme value theory \citep[e.g.,][]{deh2006a}, which allow the extrapolation to quantile levels beyond the range of the data. Existing methods that use extrapolation in the presence of predictors rely on (transformations of) linear \citep{chernozhukov2005, WangTsai2009, huixia2012, huixia2013} functions, additive models \citep{CDD05, benjamin2019}, non-parametric regression \citep{beirlant2004, martins2015} and local smoothing methods \citep{abdelaati2011, girard2012, gar2014, goe2014, goe2015, GardesStupfler2019, Velthoenetal2019, girard2022}.
However, these approaches are either not flexible enough to model complex response surfaces or do not scale well in larger dimensions $p$ of the predictor space.

Regarding the second challenge, several quantile regression methods have been proposed in the statistical and machine learning literature that can cope with predictor spaces in large dimensions and complex regression surfaces \citep{taylor2000quantile, frie2001}.
In particular, here exist several forest-based approaches for quantile regression \citep{mein2006, athe2019}.
These methods are based on the random forest originally developed by~\cite{brei2001} and can estimate flexible quantile regression functions. Compared to methods such as gradient boosting and neural networks, the main advantage of forest-based approaches is that they require little tuning and that their statistical properties are relatively well understood \citep{athe2019}. They scale well with the dimension of the predictor space as opposed to approaches based on generalized additive models \citep{koen2011} and kernel-based methods \citep{yu1998}.
While these methods work well for the estimation of quantiles inside the data range, such as $\tau_n = 0.8$ 
in Figure~\ref{fig:intro}, their performance deteriorates for quantile estimation at extreme levels $\tau \approx 1$ close to the upper endpoint of the response distribution.

In this paper, we bring together ideas from extreme value theory and forest-based methods to tackle the challenges of extreme quantile regression in large predictor dimensions $p$.
To extrapolate beyond the data range, we rely on the approximation by the generalized Pareto distribution (GPD) of the exceedances over an intermediate threshold~$u$; see the triangles in Figure~\ref{fig:intro}. Under mild assumptions, the conditional distribution of $Y \mid X = x$, given that $Y > u$ can be approximated by~\citep{balk1974,pick1975}
         \begin{align}\label{eq:gpd-extrap}
            \P\left( Y - u \leq   z \mid Y > u, X = x\right) \approx  1 - \left(1 + \frac{\xi(x)z}{\sigma_u(x)}\right)_+^{-1/\xi(x)}, \qquad z \geq 0,
          \end{align}
          where $\sigma_u(x) > 0$ and $\xi(x)\in\R$ are the conditional scale and shape parameters of the GPD, respectively. This includes responses with heavy tails ($\xi(x) > 0$), light tails ($\xi(x) = 0$) and with finite upper end points ($\xi(x)<0$).
In practice, the threshold $u$ is typically an estimate of the intermediate quantile $Q_x(\tau_n)$, where $\tau_n$ is chosen small enough such that this conditional quantile can be estimated by classical regression methods, that is, the expected number of exceedances $n(1-\tau_n) \to \infty$. At the same time, it should be large enough so that the approximation in~\eqref{eq:gpd-extrap} by the GPD is accurate, that is, $\tau_n \to 1$. By inverting the distribution function of the GPD, we readily obtain an approximation that allows us to extrapolate to extreme quantiles at levels~$\tau > \tau_n$.

To cope with complex response surfaces and large predictor spaces dimensions, we rely on ideas from the random forest literature \citep{mein2006, athe2019}. Our new extremal random forest (ERF) localizes the estimation of the GPD parameter vector $\theta(x) = (\sigma_u(x), \xi(x))$ around the predictor value $x$ using forest-based weights. Since only a few extreme observations are typically available for training, the simple tuning of random forests is a great advantage.
We further propose a penalized version of the local GPD estimation that regularizes the variability of the shape parameter in the predictor space. 

While our approach can be applied for arbitrary shape parameters $\xi(x)\in\mathbb R$, for the theoretical study we concentrate on the heavy-tailed case with positive shapes.
Under general domain of attraction conditions on the conditional response $Y \mid X=x$, we show the consistency of the ERF estimator $\hat\theta(x)$ and its penalized version $\hat\theta_{\mathrm{pen}}(x)$ for the true parameter vector $\theta(x)$.
Since our loss function, namely the GPD log-likelihood, is non-convex and misspecified, i.e., the sample follows a GPD distribution only approximately, 
the proof strategy of \cite{athe2019} cannot be used. Instead, we rely on a careful analysis of the first order conditions of the GPD likelihood; see \cite{zhou2009} for the unconditional case.
As a side result, we establish the consistency of a random forest Hill estimator, a localized, predictor-dependent version of the classical estimator by \cite{hill1975}.

Our ERF algorithm combines the advantages of accurate tail extrapolation at levels $\tau \approx 1$ with a flexible regression method that scales well with predictor dimension. In simulations, we show that ERF outperforms extreme value theory and quantile regression techniques to estimate extreme quantiles. Moreover, it is competitive with the recent gradient boosting by \cite{velthoen2021} and has the advantage of significantly easier tuning and the theoretical guarantee of our consistency result. Finally, we apply~our methodology to extreme quantile prediction for U.S.~wage data \citep{angrist2009}.
The ERF algorithm is available as an \texttt{R} package at \url{https://github.com/nicolagnecco/erf}.

\begin{figure}[!tb]
  \centering
  \includegraphics[scale=1]{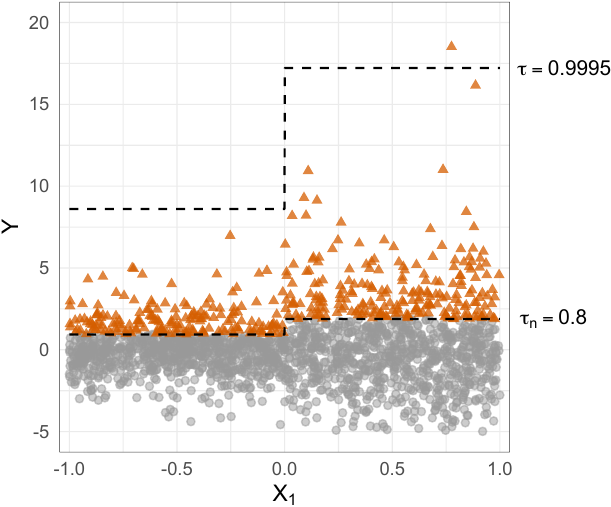}
  \caption{Realization of $n = 2000$ samples from the generative model in Example~\ref{ex:gen-mod} in Section~\ref{sec:algo}. Response $Y$ is plotted against the first predictor $X_1$.
    Dashed lines represent the quantile functions associated to the intermediate $\tau_n = 0.8$ and high $\tau = 1 - 1/n = 0.9995$ quantile levels. Triangles are observations above the intermediate threshold.}
  \label{fig:intro}
\end{figure}

\section{Background}\label{sec:background}
\subsection{Extreme Value Theory} \label{sec:evt_background}

The first challenge of extreme quantile regression is that only a few or even no data points exceed the quantiles of interest. This section considers the classical case of unconditional extremes without predictors. Let $Y_1, \dots, Y_n$ be $n$ independent copies of a real-valued random variable $Y$. The notion of an extreme quantile $\tau = \tau_n$ is typically expressed relative to the sample size $n$. The expected number of observations in the sample that exceed the $\tau_n$-quantile is then $n(1-\tau_n)$. A quantile with level $\tau_n\to 1$ such that $n(1-\tau_n) \to \infty$ is called an intermediate quantile. Empirical estimation in this case still works well since the effective sample size, that is, the number of exceedances, grows to infinity \citep{deh2006a}.
For risk assessment, the most critical case is if the quantile of interest is eventually beyond the range of the data, that is, $(1 - \tau_n) n \to 0$ as $n\to \infty$. Then, we can no longer rely on empirical estimators but must resort to asymptotically motivated approximations from extreme value theory.

Let $u^* \in (-\infty,\infty]$ be the upper endpoint of the distribution of $Y$. Under mild regularity assumptions on the tail of $Y$, the Pickands--Balkema--De Haan theorem~\citep{balk1974,pick1975} states that there exists a normalizing function $\sigma_u > 0$ with
\begin{align}\label{gpd_def}
  \lim_{u \to u^*} \mathbb P \left( \frac{Y-u}{\sigma_u} \leq z \mid Y > u\right)  = G(z; (1,\xi)),
\end{align}
where the limit on the right-hand side is the distribution function of the generalized Pareto distribution (GPD)~\citep{pick1975} given by
\begin{equation}\label{eqpot2}
  G(z; \theta) = 1-\left(1 + \frac{\xi}{\sigma} z \right)_+^{-1/\xi},\quad z > 0,
\end{equation}
and $\theta =(\sigma, \xi) \in (0, \infty)\times\R$ is the parameter vector consisting of scale and shape, respectively.
The shape parameter $\xi \in \mathbb R$, also known as the extreme value index \citep{beir2005}, characterizes the decay of the tail of $Y$. If $\xi > 0$, then $Y$ is heavy-tailed; if $\xi = 0$, then $Y$ is light-tailed; if $\xi < 0$ then $Y$ has a finite upper endpoint. Moreover, the GPD is a natural model for the distribution tails since it is the only possible limit of threshold exceedances as in~\eqref{gpd_def}. 
Note that the convergence of exceedances is equivalent to the classical result of extreme value theory that states the convergence of the suitably normalized maximum of $n$ i.i.d.~copies of $Y$ to the generalized extreme value distributions \citep{fish1928, gned1943}. 

The GPD approximation can be directly translated into an approximation for the small probability of $Y$ exceeding a high threshold $y$. By Bayes' theorem and~\eqref{gpd_def} we obtain
\begin{equation}\label{eq:bayes}
  \begin{split}
    \P(Y > y)
    & = \P(Y > u)\ \P(Y > y \mid Y > u)
    \approx \P(Y > u) \left\{1 - G(y-u; \sigma_u, \xi)\right\},
  \end{split}
\end{equation}
where $u < y$ denotes an intermediate threshold. In applications, the scale and shape parameters of the GPD have to be estimated from independent observations $Y_1,\dots, Y_n$ of $Y$. We fix an intermediate quantile level $\tau_n$ and define the exceedances $Z_i = (Y_i - \hat Q(\tau_n))_+$, $i=1,\dots, n$, where $\hat Q(\tau_n)$ denotes the empirical $\tau_n$ quantile.
We obtain estimates $\hat \theta = (\hat\sigma, \hat \xi)$ of the GPD parameter vector $\theta$ by maximum-likelihood, where the
negative log-likelihood (or deviance) contribution of the $i$th exceedance $Z_i$ is
\begin{equation}\label{eq:loglik}
  \ell_\theta(Z_i) =   \log \sigma + \left(1 + \frac{1}{\xi}\right)\log \left(1 + \frac{\xi}{\sigma}Z_i\right),\quad \theta \in (0, \infty)\times\R,
\end{equation}
if $Z_i > 0$, and zero otherwise. Combining approximation~\eqref{eq:bayes} with~\eqref{eqpot2} and letting $\P(Y > y) = 1 - \tau$ and $\P(Y > u) = 1 - \tau_n$, we obtain an approximation for the quantile of $Y$ at level $\tau > \tau_n$ as
\begin{equation}\label{eq:unc-quant}
  \hat Q(\tau) \approx \hat Q(\tau_n) + \frac{\hat \sigma}{\hat \xi} \left[\left(\frac{1 - \tau}{1 - \tau_n}\right) ^{-\hat \xi} - 1\right].
\end{equation}

\subsection{Quantile Regression and Generalized Random Forests}\label{sec:quant-reg}

Given a pair $(\g X, Y)$ of predictor vector $X \in \R^p$ and response variable $Y \in \R$, quantile regression deals with modeling the conditional $\tau$-quantile $Q_x(\tau)$ of the conditional distribution of $Y$ given that $X = \g x$ for a particular predictor value $x \in \mathbb R^p$. The main challenge is that the dimension $p$ of the predictor space may be large and that the quantile surface $Q_x(\tau)$ as a function $x$ may be a complex, highly non-linear function.

Let $(\g X_1, Y_1),\dots, (\g X_n, Y_n) $ be $n$ independent copies of the random vector $(\g X, Y)$. In contrast to the setting in Section~\ref{sec:evt_background}, classical methods for quantile regression consider a fixed quantile level $\tau_n \equiv \tau$ that does not change with the sample size. On a population level, these methods exploit the fact that the conditional quantile function is the minimizer of the expectation of the quantile loss $\rho_{\tau}(c) = c(\tau-\I{\{ c<0\}})$, $c \in \mathbb R$,~\citep[][]{koen1978}, that is $Q_{\g x}(\tau) = \argmin_{q \in \R} \E[\rho_{\tau}(Y - q) \mid \g X = \g x]$.
The previous expectation cannot be estimated directly on the sample level since the observed predictor values do not typically include the value $x$. A natural estimator is
\begin{equation}\label{eq:weight1}
  \hat{Q}_{\g x}(\tau) = \argmin_{q \in \R}\sum_{i=1}^{n}w_n(\g x, X_i) \rho_{\tau}(Y_i - q),
\end{equation}
where $x' \mapsto w_n(x, x')$ is a set of localizing similarity weights around the predictor value of interest. The weights can for instance be obtained by a kernel approach~\citep{yu1998}, but this is limited to moderately large dimensions \citep{stone1980optimal, stone1982optimal}.

In order to model more complex quantile surfaces in larger dimensions, \citet{mein2006} and~\cite{athe2019} propose to use estimator~\eqref{eq:weight1} with similarity weights $w_n(\cdot, \cdot)$ obtained from a random forest.
Random forests~\citep{brei2001} are an ensemble method used for both regression and classification tasks and consist of fitting $B$ decision trees to the training data.
In regression settings, each decision tree predicts a test point $x\in\R^p$  by $\mu_b(x) := \sum_{i = 1}^{n} \I\{X_i \in L_b(x)\} Y_i /|L_b(x)|$, for all $b = 1, \dots, B$,
where $L_b(x) \subset \R^p$ denotes the rectangular region containing $x$ in the tree $b$ and $|L_b(x)|$ the number of observations in $L_b(x)$. With similarity weights $w_{n, b}(x, X_i) := \I\{X_i \in L_b(x)\} / |L_b(x)|$, the random forest predictions are $\mu(x) := \frac{1}{B} \sum_{b = 1}^B \mu_b(x) = \sum_{i = 1}^{n} w_n(x, X_i) Y_i$,
where $w_{n}(x, X_i) = \sum_{b=1}^B w_{n, b}(x, X_i) / B$  is the average weight across $B$ trees.

The original idea of~\cite{mein2006} is to use the weights estimated by this standard regression random forest for quantile regression in~\eqref{eq:weight1}. Since trees are grown by minimizing the mean squared error loss, this leads to the fact that $w_n(x, X_i)$ takes large values for those observations $i$ such that $\E[Y \mid X = X_i] \approx \E[Y \mid X = x]$. In many situations the conditional expectation is not representative of the whole conditional distribution of $Y\mid X = x$, and it may happen that $w_n(x, X_i)$ is large but $Q_{X_i}(\tau) \not \approx Q_{x}(\tau)$; see \citet[][Figure~2]{athe2019} or our Figure~\ref{fig:intro} where the conditional expectation is constant over the predictor space. In these cases, the similarity weights estimated with standard random forest do not capture the heterogeneity of the quantile function and are thus not well-suited for quantile regression tasks.
\citet{athe2019} introduced generalized random forests (GRF), a method designed to fit random forests with custom loss functions and retaining the appealing features of classical random forests.
An important application of GRF is quantile regression, where the trees of the forest are grown to minimize the quantile loss.
In this work, we rely on GRF with quantile loss to estimate similarity weights $w_n(\cdot, \cdot)$ that capture the variation of the entire conditional distribution of $Y \mid X = x$ in the predictor space.
In practice, the GRF algorithm estimates simultaneously conditional quantiles at levels $\tau = 0.1, 0.5, 0.9$ as a proxy for the conditional distribution of $Y \mid X = x$.
For simplicity, in the sequel, we refer to GRF with quantile loss as GRF.

\section{Extremal Random Forest}\label{sec:erf}

\subsection{The Algorithm}
\label{sec:algo}

In this work we study a method for estimation of the conditional GPD parameters in~\eqref{eq:gpd-extrap} and flexible extreme quantile regression where both challenges described in Sections~\ref{sec:evt_background} and~\ref{sec:quant-reg} occur simultaneously. Consider the random vector $(X, Y)$ of predictors $X \in \mathcal X \subset \R^p$ and response $Y \in \R$, with $\mathcal X$ compact.
 Let $(\g X_1, Y_1), \dots, (\g X_n, Y_n)$ be independent copies of $(\g X, Y)$.
In many applications in risk assessment, the goal is to estimate the quantile function
$\g x \mapsto Q_{\g x}(\tau)$,
at an extreme level $\tau$, where the expected number of observations in the sample that exceed their conditional quantiles is small and possibly tends to $0$ as $n \to \infty$; see Section~\ref{sec:evt_background}.
To illustrate the challenges of this estimation problem, we consider an example where the scale of the response variable $Y$ is modeled as a step function of the covariates $\g X$.
This corresponds to \citet[][Figure~2]{athe2019}, except that we assume that the noise of the response variable is heavy-tailed instead of Gaussian.
\begin{example}\label{ex:gen-mod}
  Let $\g X  \sim U_p$ be a uniform distribution on the cube $[-1, 1]^p$ in dimension $p$ and
  $Y \mid \g X = \g x \sim s(\g x)\ T_{4}$, where $T_{\nu}$ denotes a Student's $t$-distribution with $\nu > 0$ degrees of freedom.
  The shape parameter of the conditional distribution $Y\mid X = x$ is then constant $\xi(\g x) = 1/\nu(x) \equiv 0.25$ and we choose the $s(\g x) = 1 + \I\{x_1 > 0\}$ for $\g x \in \R^p$. The GPD scale parameter $\sigma_u(x)$ of $Y \mid \g X = \g x$ and therefore also the quantile function $Q_{\g x}(\tau)$ only depend on $X_1$. The other predictors are noise variables.
\end{example}

As discussed in the introduction, the estimation of tail probabilities and quantiles exhibits the two difficulties of localization of predictors and extrapolation in the direction of the response variable. 
Our methodology accurately addresses both of these challenges. For effective localizing in the predictor space, even when the dimension is large, we use the weights emerging from GRF \citep{athe2019}. For correct extrapolation in the tail of the conditional response variable, we rely on the asymptotic theory of extremes and fit a localized generalized Pareto distribution; see Section~\ref{sec:evt_background}. More precisely, we assume that the distribution function of $Y  - u$, conditional on the exceedance $Y > u$ over a high threshold $u$, is approximately generalized Pareto \citep{balk1974} with scale and shape parameters depending on the predictor value $\g x$.
\begin{assumption}[Domain of attraction]\label{ass:doa}
  For every $x \in \mathcal{X}$, we let $u^*(x) \in (-\infty,\infty]$ be the upper endpoint of the conditional distribution function $F_x$ of $Y \mid X=x$, and assume that it is continuous and strictly monotonically increasing.
  We further assume that $F_x$ is in the domain of attraction of an extreme value distribution with shape parameter $\xi(x) \in \mathbb R$, that is, there exists a function $(x, u) \mapsto \sigma_u(x) > 0$ such that for all $y>0$
          \begin{align}\label{eq:gpd-approx}
            \lim_{u \to u^*(x)} \P\left(\frac{Y - u}{\sigma_u(x)} \leq  z \mid Y > u, X = x\right) =  1- (1 + \xi(x)z)_+^{-1/\xi(x)},
          \end{align}
          where we call $\theta(x) = (\sigma_u(x),\xi(x))$ the conditional GPD parameters.
\end{assumption}

\begin{remark}
  In the conditional framework, the scale and shape parameters are functions $\sigma_u(\cdot): \mathcal X \to (0, \infty)$ and $\xi: \mathcal X \to \R$ on the predictor space, respectively. As in the unconditional case, the scale function depends on the threshold $u$, but we often drop the subscript for notational simplicity.
  The convergence~\eqref{eq:gpd-approx} is equivalent to several other conditions, such as the convergence of the normalized maxima of independent copies of $Y \mid X =x$ to a generalized extreme value distribution.
\end{remark}

Assumption~\ref{ass:doa} is a conditional version of~\eqref{gpd_def} and means that the GPD approximation~\eqref{eq:bayes} and the quantile approximation~\eqref{eq:unc-quant} hold for the distribution of $Y \mid X=x$ for any $x \in \mathcal X$. It is satisfied by most data-generating processes as, for instance, in Example~\ref{ex:gen-mod}.

To use this approximation in practice, we have to choose a threshold $u$ that depends on the $n$ training observations.
To show the pointwise consistency of the estimators of the conditional GPD parameters in Section~\ref{sec:consistency}, it will be crucial to guarantee that at each point $x\in\mathcal X$ in the predictor space, there are approximately the same amount of expected exceedances. The threshold $u(x) = \hat Q_{x}(\tau_n)$ is therefore usually taken to be a predictor-dependent estimator of the intermediate quantile function. 
Here, $\tau_n \in (0,1)$ is an intermediate probability level that is chosen such that $\hat Q_{x}(\tau_n)$ can be obtained by classical quantile regression techniques; see Section~\ref{sec:quant-reg}.
In principle, any quantile regression method can be used to fit $\hat Q_x(\tau_n)$. We choose GRF with quantile loss \citep{athe2019} since it is a method suitable for flexible quantile regression problems and it requires little tuning.

In order to formulate our estimators of the conditional GPD parameters $\theta(x)$ and the extreme quantile $ Q_{x}(\tau)$, we define the exceedances in the training data as
\begin{equation}\label{def_exc}
  Z_i := (Y_i - \hat Q_{\g X_i}(\tau_n))_+, \quad i = 1, \dots, n;
\end{equation}
see the triangles in Figure~\ref{fig:intro}. The limit relation~\eqref{eq:gpd-approx} implies that the distribution of $Z_i$ can be well approximated by a GPD with parameter vector $\theta(X_i)$. 
For estimation of the GPD parameter vector $\theta(x) = (\sigma(\g x), \xi(\g x))$ we rely on those exceedances that carry most information on the tail of $Y \mid \g X = \g x$. 
  Such a localization can be achieved by assigning to each exceedance $Z_i$ a suitable weight $w_n(\g x, X_i)$ that reflects the importance for estimating $\theta(x)$; see Section~\ref{sec:quant-reg} for a similar rationale in the context of quantile regression.
To do so, we use the localizing weight functions $w_n(\g x, X_i)$ estimated from a GRF \citep{athe2019} whose tuning parameters are optimized for the purpose of estimating the conditional GPD parameters; this GRF can therefore be \emph{different} from the GRF used for the intermediate quantile $\hat Q_{x}(\tau_n)$.
We would like to define the estimator of the conditional GPD parameter $\hat\theta(x)$ as the minimizer of the weighted (negative) log-likelihood
\begin{equation}\label{eq:weighted-loglik}
  L_{n}(\theta; x) =  \sum_{i =1}^{n} w_n(\g x, X_i) \ell_\theta(Z_i) 1\{Z_i > 0\}, \quad x \in \mathcal {X},
\end{equation}
where $\ell_\theta$ is defined in~\eqref{eq:loglik}.
In practice, the parameter space $\theta(\mathcal X) = \{\vartheta \in (0, \infty) \times \R: \vartheta = \theta(x) \text{ for some } x \in \mathcal X \}$ is unknown.
As explained by~\citet{dombry2015}, it is not guaranteed that the log-likelihood of the generalized extreme value distribution has a global optimum over the parameter space $(0, \infty) \times \R$. In fact, \citet{smith1985} shows no maximum likelihood estimator exists when $\xi \leq -1$. Analogous results apply to the GPD log-likelihood $L_n(\theta; x)$ \citep{dre2004}. We therefore define $\hat\theta(x)$ as the optimizer of $L_n(\theta; x)$ over an arbitrarily large compact set $\Theta \subset (0, \infty) \times (-1, \infty)$ such that $\theta(\mathcal X) \subset \mathrm{Int}\ \Theta$, that is,
\begin{equation}\label{eq:weighted-mle}
  \hat\theta(x) \in \argmin_{\theta \in \Theta}L_{n}(\theta; x).
\end{equation}
In practice, the minimizer is obtained by solving the first order conditions $\nabla L_n(\theta; x) = 0$, which are given in~(A.10) %
in the Appendix. 
The estimated pair $(\hat Q_{\g x}(\tau_n), \hat \theta(x))$ of  intermediate quantile and conditional GPD parameters can be plugged into extrapolation formula~\eqref{eq:unc-quant}
to obtain an estimate  $\hat Q_{\g x}(\tau)$ of the extreme conditional quantile at level $\tau > \tau_n$.

In Algorithm~\ref{alg:erf}, we describe our prediction method, which we call the extremal random forest (ERF).
The algorithm consists of two subroutines, namely \textsc{ERF-Fit} and \textsc{ERF-Predict}. The \textsc{ERF-Fit} subroutine estimates a similarity weight function $(x, y) \mapsto w_n(x, y)$ and an intermediate quantile function $x \mapsto \hat Q_{x}(\tau_n)$ from the training data, for $x, y \in \mathcal X$.
The similarity weight function $w_n(\cdot, \cdot)$ is estimated with a generalized quantile random forest (GRF) from \citep{athe2019}, whereas the intermediate quantile function $\hat Q_{\cdot}(\tau_n)$ can be estimated with any quantile regression technique of choice. 
The \textsc{ERF-Predict} subroutine predicts the extreme $\tau$-quantile $\hat Q_{x}(\tau)$, with $\tau > \tau_n$, at point $x\in\mathcal X$ by estimating the GPD parameter vector $\theta(x)$ as in~\eqref{eq:weighted-mle}. 
  We note that the localized likelihood in~\eqref{eq:weighted-loglik} can be seen as a nearest-neighbor or kernel approach \citep[e.g.,][]{abdelaati2011, GardesStupfler2019}, where the weight for each observation is estimated adaptively by the tree splitting of the random forest.
\begin{algorithm}
  \caption{Extremal random forest (ERF)}
  \label{alg:erf}
  Denote by $\DD = \{(\g X_i, Y_i)\}_{i = 1}^{n}$ the training data.
  Let $\g x \in \mathbb R^p$ be a test predictor value.
  Specify the intermediate quantile level $\tau_n$ and the extreme quantile level $\tau$, with $\tau_n < \tau$.
  Let $\alpha$ be a vector of hyperparameters supplied to \textsc{GRF}.

  \begin{algorithmic}[1]
    \Procedure{ERF-Fit}{$\DD, \tau_n, \alpha$}

    \State $w_n(\cdot, \cdot) \gets$ \Call{GRF}{$\DD, \alpha$}

    \State $\hat Q_{\cdot}(\tau_n) \gets$ \Call{QuantileRegression}{$\DD$}

    \State \textbf{output} \texttt{erf}  $\gets [\DD, w_n(\cdot, \cdot),  \hat Q_{\cdot}(\tau_n)] $

    \EndProcedure
  \end{algorithmic}

  \begin{algorithmic}[1]
    \Procedure{ERF-Predict}{\texttt{erf}, $\g x, \tau$}

    \State $Z_i \gets (Y_i - \hat Q_{\g X_i}(\tau_n))_+$, with $i = 1, \dots, n$

    \State $\hat \theta(x)\gets \argmin_{\theta} L_n(\theta; x)$ as in~\eqref{eq:weighted-loglik}

      \State $\hat Q_{\g x}(\tau)\gets \Call{GPD}{\hat Q_{x}(\tau_n),\hat \theta(x)}$ 

    \State \textbf{output} $\hat \theta(x)$ and $\hat Q_{\g x}(\tau)$
    \EndProcedure
  \end{algorithmic}

  The subroutine GRF estimates the similarity weight function $w_n(\cdot, \cdot)$  using the generalized random forest of~\citet{athe2019}.
  The subroutine \textsc{QuantileRegression} fits the intermediate conditional quantile function $\hat Q_{\cdot}(\tau_n)$ using a quantile regression technique of choice.
  The object \texttt{erf} returned by \textsc{ERF-Fit} 
  is a list containing the training data $\DD$,
  the fitted intermediate quantile $\hat Q_{\cdot}(\tau_n)$,
  and the estimated similarity weight function $w_n (\cdot, \cdot)$.
\end{algorithm}

Appendix~B %
shows the estimated GRF weights $w_n(\g x, X_i)$ used in the likelihood in~\eqref{eq:weighted-loglik} for Example~\ref{ex:gen-mod} and specific values of $x$. It can be seen that the weights are large for training observations $X_i$ where the distribution of $Y \mid X = X_i$ is equal to the one of $Y \mid X =x$.

Generalized random forests have several tuning parameters, such as the number of predictors selected at each split and the minimum node size.
Appendix~C %
presents a cross-validation scheme to tune such hyperparameters within our algorithm. For large values of $\tau \approx 1$, the quantile loss is not a reliable evaluation metric since there might be few or no test observations above this level.
In our case, we instead rely on the tail approximation in~\eqref{eq:gpd-approx} and use the deviance of the GPD as a reasonable metric for cross-validation.

\subsection{Consistency}\label{sec:consistency}

For sample size $n$ and intermediate quantile level $\tau_n$ with $\tau_n \to 1$ and $n(1-\tau_n) \to \infty$, ERF provides an estimate $\hat\theta(x) = (\hat\sigma(x), \hat\xi(x))$ of the conditional GPD parameter $\theta(x)$ that describes the distribution of $(Y \mid Y > \hat Q_x(\tau_n), X =x)$. This estimate is obtained in~\eqref{eq:weighted-mle} as the maximizer of the localized GPD likelihood, which takes as input the exceedances defined in~\eqref{def_exc}. The latter requires an estimator of the intermediate quantile function, and as already noted in Section~\ref{sec:algo}, any existing method can be used. We assume in the sequel that this method is uniformly consistent.
In the asymptotic theory of extreme values it is common to denote by  $k = n(1-\tau_n)$ the expected number of exceedances and thus the effective sample size for GPD estimation. The requirement for $\tau_n$ to be an intermediate quantile level is equivalent to $k/n\to 0$ and $k \to\infty$.

\begin{assumption}[Uniform consistency of intermediate quantile estimator]
  \label{ass:hatq_x-uniform-consistent}
  The estimated intermediate quantile function is uniformly consistent
  at level $\tau_n = 1-k/n$ with $k/n\to 0$ and $k \to\infty$, in the sense that
    $\sup_{x \in \mathcal{X}} |\hat Q_{x}(\tau_n)/ Q_{x}(\tau_n)| \stackrel{\P}{\to} 1$ as $n\to \infty.$
\end{assumption}
This assumption is weaker than requiring that the estimated quantiles converge to the true counterparts since only the ratio needs to be close to one.
For instance, a possible choice for such a uniformly consistent method is given in~\citet{huixia2013}.

The ERF method is at the interface of random forests and extreme value theory, and both fields have their challenges related to the analysis of asymptotic properties.
Consistency and asymptotic normality of classical \citep{mein2006, biau12a, scornet2015, wager2018estimation}
and generalized random forests \citep{athe2019} have only recently been established.
The results by~\citet{athe2019} require regularity conditions (see Assumptions 1--6 of their paper) that are not satisfied in our setting. In particular, the negative GPD log-likelihood $\theta \mapsto \ell_\theta(z)$ that we consider is not a convex function and, therefore, it does not satisfy Assumption~6 in \citet{athe2019}.
  An additional challenge arises from the fact that
the theory in \citet{athe2019} is developed for data that come from a fixed distribution. Since we work under the domain of attraction condition in Assumption~\ref{ass:doa} our model is misspecified, in the sense that the sample follows a GPD distribution only approximately. Moreover, with changing thresholds, the distribution of the exceedances changes. This pre-limit approximation is the reason why the asymptotic analysis of extreme value estimators is notoriously difficult even in the i.i.d.~case \citep{dre2004,zhou2009}.

We thus require assumptions from both fields, namely on how the forest is grown and the tail behavior of the response as a function of the predictors.
Similarly to \cite{WangTsai2009}, \citet{gar2014} and \cite{goe2015}, 
we focus on the heavy-tailed case where $\xi(x) > 0$ for all $x \in \mathcal X$, where the tail and the quantile functions of the conditional distribution of $Y \mid X=x$ can be written as 
$1 - F_x(y) = y^{-1/\xi(x)}\tilde{\ell}_x(y)$, $Q_x(\tau) = (1-\tau)^{-\xi(x)}\ell_x\left( (1-\tau)^{-1} \right)$,
respectively, where $\tilde{\ell}_x, \ell_x: \R \to \R$ are slowly varying functions \citep[e.g.,][]{bingham1989}.
Any such slowly varying function $\ell$ has a normalized representation 
\begin{align}\label{eq:slowly-karamata}
  \ell(y) = c \exp \int_1^{y}\frac{\alpha(t)}{t}\mathrm{d}t, \qquad y\geq 1, 
\end{align}
characterized by a constant $c>0$ and a function $\alpha: [1,\infty) \to \R$ with $\lim_{t\to\infty} \alpha(t)=0$.
We denote the characterizing tuples for the functions $\tilde \ell_x$ and $\ell_x$ by $(\tilde c(x), \tilde \alpha_x)$ and $(c(x), \alpha_x)$, respectively, for any $x\in \mathcal X$.
In order to localize information in the predictor space, we need to assume a certain regularity of the conditional quantile function at extreme levels.   

\begin{assumption}[Lipschitz conditions]\label{ass:lipschitz}
  Assume that the predictor space $\mathcal{X} \subseteq \R^p$ is compact and that the predictor distribution possesses a density on $\mathcal X$ that is bounded away from zero and infinity.
 Moreover, assume that the shape parameter function $\xi: \mathcal{X} \to \R$ is Lipschitz continuous with $\xi(x) > 0$ for all $x \in \mathcal{X}$ with Lipschitz constant $L_\xi$ such that $|\xi(x) - \xi(y)| \leq L_\xi \norm{x - y}_2$,
  for all $x,y\in \mathcal X$.
  Moreover, the functions $\log c, \alpha_\cdot(t) : \mathcal{X} \to \R$ are Lipschitz and uniformly (in $t$) Lipschitz continuous with constants $L_{c}$ and $L_\alpha$, respectively, that is,
  $ |\log c(x) - \log c(y)| \leq  L_{c} \norm{x - y}_2$ and $\sup_{t \geq 1}|\alpha_x(t) - \alpha_y(t)| \leq L_{\alpha} \norm{x - y}_2$, for all $x,y\in \mathcal X$.
  Finally, we assume that $\lim_{t\to\infty} \tilde\alpha_x(t)=0$ uniformly in $x\in\mathcal X$.
\end{assumption}

  These Lipschitz conditions are fairly natural and also appear in similar form in previous extreme quantile regression techniques \citep[e.g.,][]{goe2015, gar2014, GardesStupfler2019}. 
   The next example illustrates that they are satisfied for a large class of models.
     
\begin{example}\label{ex:lipschitz}
  Suppose that $Y_0$ has a heavy-tailed distribution with shape index $\xi_0$, and parameters $c_0$ and $\alpha_0$ in~\eqref{eq:slowly-karamata} of the slowly varying function of its quantile function $Q(\cdot)$. Consider the predictor-dependent model $(Y \mid X = x)  \sim s(x) Y_0^{\xi(x)}$, $x\in \mathcal X.$
  It can be readily verified that the quantile function of this model is 
  \[Q_x(\tau) = s(x) Q(\tau)^{\xi(x)} = (1-\tau)^{-\xi_0 \xi(x)}  s(x) c_0^{\xi(x)} \exp \left\{ \xi(x) \int_1^y \alpha_0(t)/t \mathrm d t \right\}. \]
  Suppose that the function $s(x)$ and $\xi(x)$ are Lipschitz and strictly positive on $\mathcal X$. Then all conditions of Assumption~\ref{ass:lipschitz} are satisfied.
\end{example}

Concerning the specification of the random forest and the corresponding similarity weights, we follow \citet{athe2019}.
  In particular, we put an assumption on the rates of convergence of the leaf's diameter of each tree in the forest.

\begin{assumption}[Leaf's diameter rate of convergence]\label{ass:sim-weights}
    Let $b = 1, \dots, B$ denote a tree in the forest and let $x \in \mathcal{X}$ be a fixed predictor point. Define the diameter of the leaf $L_b(x)$ by $\diam(L_b(x)) \coloneqq \sup\{\norm{y - x}_2 : y \in L_b(x)\}$. 
    Let $s < n$ denote the number of observations used to grow the tree.
    We assume that the diameter of the leaf $L_b(x)$ converges in probability to zero, that is, for every $\varepsilon > 0$, $\P[ \diam( L_b(x) ) > \varepsilon] \to 0$ as $s \to \infty$.
    Furthermore, we assume that for $s$ large enough, the expected value of the leaf's diameter satisfies $\E[ \diam( L_b(x) ) ] = \mathcal{O}\left( s^{-C} \right)$, 
    for some positive constant $C > 0$.
  
\end{assumption}
  The leaf's diameter can be seen as a data-driven bandwidth parameter in a kernel. Unlike in kernel-based methods, where it is common to assume a deterministic bandwidth converging to zero, here, we put an assumption on the rate of convergence of a stochastic `bandwidth'.
  As we show in Appendix~A.1, %
  the GRF from ~\cite{athe2019} satisfies Assumption~\ref{ass:sim-weights}. 
  The similarity weights $w_n(x, X_i)$ for the exceedances $Z_i$ in the localized likelihood~\eqref{eq:weighted-loglik}
are the main ingredient for flexible estimation of the conditional GPD parameters $\theta(x)$.
For consistency of the estimator, the weights must localize around the point of interest $x$ as $n\to\infty$; that is, only observations with $X_i$ close to $x$ get positive weights. 
Similarity weights from a GRF depend on the leaf's diameter of each tree, which satisfies Assumption~\ref{ass:sim-weights}, and therefore, they localize around the point of interest $x$ as $n\to\infty$.

The following theorem shows the existence and consistency of a solution of the first order conditions~(A.10) %
in Appendix~A.2 %
corresponding to the localized optimization problem \eqref{eq:weighted-mle}. Define the event 
$A_n = \{\text{there}$ exists a solution of the first order conditions~(A.10) for sample size $n \}.$

\begin{theorem}[Consistency of $\hat{\theta}(x)$]\label{thm:consistency-new}
  Let $x \in \mathcal{X}$ denote a fixed predictor value and let
  $\tau_n = 1-k/n$ be an intermediate quantile level. Suppose that Assumptions~\ref{ass:doa}--\ref{ass:sim-weights} hold. We choose constants $0 < \beta_s < \beta_k < 1$ and let the number of exceedances and the subsample size of the random forest be respectively 
  \begin{align}\label{eq:n-k-s}
    k = n^{\beta_k}\ \text{and}\ s = n^{\beta_s}.
  \end{align}
  Then, with probability tending to one, there exists a solution $\hat\theta(x) \coloneqq (\hat\sigma(x), \hat\xi(x))$ to the localized first order conditions in~(A.10), that is, $\mathbb P(A_n) \to 1$ as $n\to\infty$, and on this set the solution is consistent
  \begin{align}\label{eq:consistency}
    \hat\xi(x) \stackrel{\P}{\to} \xi(x)\ \text{and}\ \frac{\hat\sigma(x)}{\xi(x)Q_x(\tau_n)} \stackrel{\P}{\to} 1, \qquad n\to\infty.
  \end{align}
\end{theorem}
\begin{remark}\label{rem:consistency}
  Several remarks concerning the above theorem are in place.
  \begin{itemize}
    \item[(i)]
    Generalized random forests \citep{wager2018estimation} require only $s\to\infty$ and $s/n \to 0$, as $n\to\infty$. For ERF, we have the stronger condition that also $s/k \to 0$. This is natural since the effective sample size for GPD estimation is of order $k$ rather than~$n$.   
    \item[(ii)]
    The population version of the scale parameter depends on $n$ and is only asymptotically defined. In the heavy-tailed case $\xi(x) > 0$, a possible choice for $\sigma_u(x)$ in~\eqref{ass:doa} is $\xi(x) u$. Since we use $u = Q_x(\tau_n)$ as (population) threshold, this explains the normalizing sequence for $\hat\sigma(x)$ in~\eqref{eq:consistency}.  
    \item[(iii)]
    While we only consider the heavy-tailed case $\xi(x)>0$ here, the proof strategy for the case $\xi(x) < 0$
      would follow a similar structure, which we discuss in Appendix~A.7. %
    The case $\xi(x) = 0$, however,  would require a different proof strategy; we refer to \cite{zhou2009} for the unconditional case. 
    \item[(iv)]
    The proof of Theorem~\ref{thm:consistency-new} reveals that under Assumption~\ref{sec:consistency}, the same data can be used to first fit the intermediate threshold model $\hat Q_x(\tau_n)$ and then to compute the exceedances as input for our localized optimization~\eqref{eq:weighted-mle}.    
  \end{itemize}
\end{remark}

To the best of our knowledge, Theorem~\ref{thm:consistency-new} is the first consistency proof of a forest-based maximum likelihood estimator of the GPD parameters that works for 
large (fixed) dimension of the predictor space and
complex parameter response surfaces. \cite{WangTsai2009} show asymptotic normality for the model parameters for the heavy-tailed case, but only in the situation where the covariate dependence is linear (after a log transformation). There are no asymptotic results for models for generalized Pareto distributions with parameters depending in a more complex way on the covariates such as through generalized additive models \citep{CDD05, benjamin2019}, trees \citep{FLT20}, gradient boosting \citep{velthoen2021} or neural networks \citep{Olivier2022}.

The proof of Theorem~\ref{thm:consistency-new} relies on the structure of the consistency proof in the unconditional case of \cite{zhou2009}. 
Since in our case we have predictor dependent data and need to localize the first order conditions, we encounter significant additional difficulties.
A main step in our proof is to establish the consistency of
a local Hill estimator for the extreme value index. While in the unconditional case, this is a classical result, we state it for the random forest Hill estimator
as a corollary of Theorem~\ref{thm:consistency-new}, which is of independent interest. 

\begin{cor}\label{cor:hill}
   Define the random forest Hill estimator as 
  \begin{align}\label{erf_hill}
     \hat \xi_H(x) = \frac{n}{k} \sum_{i = 1}^{n}  w_n(x, X_i) \1\{Z_i > 0\}\log\left(1 + Z_i / \hat{Q}_x(\tau_n)\right).
  \end{align}
  Suppose the assumptions of Theorem~\ref{thm:consistency-new} hold. Then $\hat\xi_H(x) \stackrel{\P}{\to} \xi(x)$ as $n\to\infty.$
\end{cor}

\begin{remark}
  The classical Hill estimator \citep{hill1975} for i.i.d.~data $Y_1,\dots, Y_n$ is $\hat \xi_H = \frac{1}{k}\sum_{i=1}^n \1\{ Y_i > \hat Q(\tau_n) \}\left[\log Y_i  - \log \hat Q(\tau_n) \right] = \frac{n}{k}\sum_{i=1}^n \frac{1}{n}\1\{ Z_i > 0 \}\log\left(1 + Z_i/\hat Q(\tau_n) \right)$,
where $\hat Q(\tau_n)$ is the empirical quantile of the sample at level $\tau_n = 1 - k/n$, and the exceedances are defined as $Z_i = (Y_i - \hat Q(\tau_n))_+$.
This illustrates the similarity to the random forest Hill estimator in~\eqref{erf_hill}. The main difference is that the classical estimator uses the same weights $1/n$ for all samples, and the unconditional intermediate quantile $\hat Q(\tau_n)$ simply equals the $(n-k)$th order statistic $Y_{n-k,n}$ of the sample. On the other hand, in the predictor-dependent case, the localizing weights play a crucial role, and the exceedances rely on an estimate of the intermediate conditional quantile at $x\in\mathcal X$.   
\end{remark}

As suggested by a referee, it is worthwhile to note that in the  heavy-tailed case a simpler approximation than~\eqref{eq:unc-quant} for the extreme quantiles is possible. Indeed, if we choose $\sigma_u(x) = \xi(x) Q_x(\tau_n)$ as in Remark~\ref{rem:consistency}, then for $\tau > \tau_n$ we have $Q_x(\tau) \approx Q_{x}(\tau_n)\left(\frac{1 - \tau}{1 - \tau_n}\right) ^{-\xi(x)}$.
Using this approximation is an alternative approach for extreme quantile estimation due to \cite{wei1978}. It is a common strategy
for unconditional data \citep[e.g.,][]{girard2012,girard2022}, as well as in the predictor dependent case where $\xi(x)$ is estimated with linear or kernel-based methods \citep[e.g.,][]{WangTsai2009,abdelaati2011, huixia2012, GardesStupfler2019}.
We may consider the Weissman extrapolation in conjunction with our random forest Hill estimator~\eqref{erf_hill} as an alternative to ERF.
Yet another method in the heavy-tailed case is to use the fact that the log-transformed exceedances are approximately exponential with mean $\xi(x)$ that can be fitted by a classical random forest. Appendix~D.2 %
provides details on these alternative methods and compares them to ERF, together with a sensitivity analysis with respect to the intermediate quantile level $\tau_n$. In summary, ERF outperforms the other two methods significantly when pre-asymptotic bias is present, that is, when the data are not exactly GPD distributed but are only in the domain of attraction. In this more realistic scenario, ERF is also more stable with respect to the choice of $\tau_n$. In the remainder of the paper we therefore focus on the GPD-based ERF, but the Weissman-type estimators may be of independent interest.

\subsection{Penalized Log-Likelihood}

The shape $\xi$ of the GPD is the most crucial parameter since it determines the tail behavior of $Y$ at extreme quantile levels; the extrapolation formula~\eqref{eq:unc-quant} shows the highly nonlinear influence of the shape parameter on large quantiles.
Estimation of the shape parameter is notoriously challenging, and the maximization of the GPD likelihood may exhibit convergence problems for small sample sizes \citep{coles1999likelihood}.
Penalization can help reduce the variance of an estimator at the cost of higher bias \citep{hast2009}. 
Several schemes have been proposed for unconditional GPD estimation using penalty functions \citep{coles1999likelihood} and priors \citep{de2003bayesian} on the shape parameter in the frequentist and Bayesian frameworks, respectively.

While the above regularization methods are tailored to i.i.d.~data,
in our setting, we want to penalize the variation of the shape function $x \mapsto \xi(\g x)$ across the predictor space~$\mathcal X$. In spatial applications, for instance, it is common to assume a constant shape parameter at different locations \citep[e.g.,][]{fer2012,eng2017a}. Similarly, in ERF, we shrink the estimates $\hat \xi(\g x)$ to a shape parameter estimate $\hat \xi$ that is constant in the predictor space $\mathcal X$. In general, $\hat \xi$ could be fixed and given by expert knowledge, but often a good choice is the unconditional fit  obtained by minimizing the GPD deviance in~\eqref{eq:weighted-loglik} with constant weights $w_n(x, y) = 1$ for all $x, y \in \mathcal X$.

We propose to penalize the weighted GPD deviance~\eqref{eq:weighted-loglik} with the squared distance between the estimates of $\xi(\g x)$ and the estimated constant shape parameter $\hat \xi$, that is,
\begin{equation}\label{eq:pen-mle}
  \hat\theta_{\mathrm{pen}}(x) = \argmin_{(\sigma, \xi) = \theta \in \Theta}\frac{n}{k} L_n(\theta; x) + \lambda_n (\xi - \hat\xi)^2,
\end{equation}
where $\lambda_n \geq 0$ is a tuning parameter.
The parameter $\lambda_n$ allows interpolating between a simpler model with a smooth or constant shape function when $\lambda_n$ is large, and a more complex model with a varying shape over the predictor space when $\lambda_n$ is small.
This penalized negative log-likelihood can be interpreted in a Bayesian sense: it is equivalent to the maximum \emph{a posteriori} GPD estimator when putting Gaussian prior $N(\hat \xi, 1/(2\lambda_n))$
on the shape parameter $\xi$.
\citet{bucher2020penalized} propose the same penalization as in~\eqref{eq:pen-mle} to estimate the generalized extreme value distribution parameters, where the prior distribution is centered around an expert belief $\hat\xi\equiv \xi_0$ and $\lambda_n \geq 0$ reflects the confidence in such belief.

Similarly to the unpenalized optimization problem in~\eqref{eq:weighted-mle}, in practice an optimizer of~\eqref{eq:pen-mle} is found by solving the corresponding first order conditions~(A.27) %
in Appendix~A.3. %
Under a mild assumption on the constant shape parameter estimate $\hat \xi$, we show existence and consistency of the penalized estimator $\hat\theta_{\mathrm{pen}}(x)$ if the sequence $\lambda_n$ tends to $0$ as $n\to\infty$. This is the same condition on the penalization parameter as in the classical regression case with lasso or ridge penalties \citep{fu2000}. Define the set
$B_n = \{ \text{there exists a solution of the first order conditions~(A.27) for sample size } n \}.$

\begin{theorem}
  [Consistency of $\hat\theta_{\mathrm{pen}}$]\label{thm:consistency-pen}
  Let $x \in \mathcal{X}$ denote a fixed predictor value and let
  $\tau_n = 1-k/n$ be an intermediate quantile level. Suppose that the assumptions of Theorem~\ref{thm:consistency-new} hold.
  Furthermore, let $\lambda_n$ be a sequence satisfying $\lambda_n \to 0$ as $n\to\infty$ and assume that $\hat \xi$ is bounded in probability as $n\to\infty$.
  Then, with probability tending to one, there exist a solution $\hat\theta_{\mathrm{pen}}(x) = (\hat\sigma_{\mathrm{pen}}(x), \hat\xi_{\mathrm{pen}}(x))$ to the penalized first order conditions in~(A.27), that is, $\mathbb P(B_n) \to 1$ as $n\to\infty$, and on this set the solution is consistent, that is $ \hat\xi_{\mathrm{pen}}(x) \stackrel{\P}{\to} \xi(x)$ and $\hat\sigma_{\mathrm{pen}}(x)/(\xi(x)Q_x(\tau_n)) \stackrel{\P}{\to}$ as $n \to \infty$.
\end{theorem}

\begin{remark}
  The assumption that the constant shape parameter estimate $\hat \xi$ is bounded in probability as $n\to\infty$ is very weak. It is trivially satisfied if it is chosen as a constant $\xi_0$ by expert knowledge, or implied by the classical consistency if the unconditional estimator for the shape parameter is used \citep{dre2004,zhou2009}. 
\end{remark}

In practice, when we penalize the shape parameter we modify Algorithm~\ref{alg:erf} by replacing Line~3 of the \textsc{ERF-Predict} subroutine with~\eqref{eq:pen-mle}. Similarly, we cross-validate $\lambda$ using the scheme presented in Appendix~C %
on the modified Algorithm~\ref{alg:erf}.
Figure~\ref{fig:cv-lambda} shows the square root MISE over 50 simulations for different values of $\lambda$ and different quantile levels. 
\begin{figure}[!tb]
  \centering
  \includegraphics[scale=1]{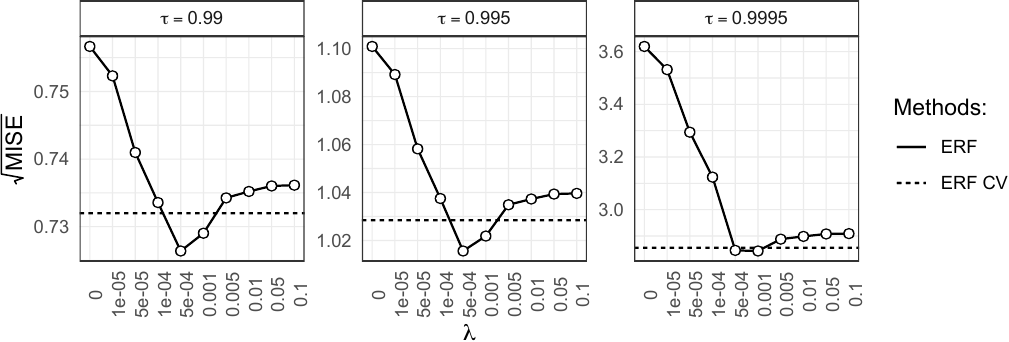}
  \caption{Square root MISE of ERF for different penalty values $\lambda$ and quantile levels $\tau$ over 50 simulations. The data is generated according to Example~\ref{ex:gen-mod}.
  }
  \label{fig:cv-lambda}
\end{figure}

\section{Simulation Study}\label{sec:sim-study}
\subsection{Setup}\label{sec:setup}
We compare ERF to other quantile regression methods on simulated data sets and assess the properties of the different approaches.
We simulate $n$ independent training observations from the random vector $(X,Y)$. The predictor $\g X \in \R^p$ follows a uniform distribution on the cube $[-1, 1]^p$ for different dimensions~$p$, and the conditional response variable $Y \mid X = x$ follows distributions with tail heaviness depending on the simulation study. The goal is to predict the conditional quantiles $Q_{\g x}(\tau)$ for moderately to very extreme quantile levels $\tau>0$.
We evaluate the methods on test data $\{\g x_i\}_{i = 1}^{n'}$ of  $n'=1000$ observations generated with a Halton sequence \citep{halton1964} on the cube $[-1, 1]^p$.
For a fitted quantile regression function $x \mapsto \hat Q_{x}(\tau)$, $\tau \in (0, 1)$, we compute the test integrated squared error (ISE) as $\text{ISE} = \sum_{i = 1}^{n'} \left(\hat Q_{\g x_i}(\tau) - Q_{\g x_i}(\tau)\right)^2/n'$,
where $x \mapsto Q_{x}(\tau)$ is the true quantile function of the model.
We obtain mean integrated squared error (MISE) by averaging $m = 50$ repetitions of the fitting and evaluation process.

The first experiment studies how ERF performs on the two challenges of high quantile levels and large-dimensional predictor spaces illustrated in Figure~\ref{fig:intro}.
The data sets follow the model of Example~\ref{ex:gen-mod} where the response has a Student's $t$-distribution with scale shift. We consider the methods' performances for different dimensions $p$ of the predictor space and different quantile levels~$\tau$.
The second experiment studies the robustness of the methods
to different tail heaviness, ranging from exponential tail ($\xi = 0$) to heavy tails ($\xi = 0.33$).

In the third experiment (see Appendix~D.1), %
we consider more complex regression functions for the conditional response variables to assess the performance of the quantile regression methods on complex data. The underlying models depend on more than one predictor value, and both the scale and the shape parameters vary simultaneously. According to Example~\ref{ex:lipschitz}, they all satisfy Assumption~\ref{ass:lipschitz} of our consistency Theorem~\ref{thm:consistency-new}.

\subsection{Competing Methods and Tuning Parameters}\label{sec:competing-methods}
Among the forest-based algorithms, we consider quantile regression forests \citep{mein2006}, denoted by QRF, and generalized random forests~\citep{athe2019}, denoted by GRF.
Since these methods do not rely on the GPD likelihood, it is not possible to cross-validate their tuning parameters as in Appendix~C for prediction error of extreme quantiles. However, we notice that their tuning parameters do not significantly influence the results and thus use the default values; see Section~\ref{sec:quant-reg} for details on forest-based approaches. As a hybrid method that uses forest-based weights, we consider the method EGP Tail \citep{tail2017} who assume that the entire conditional distribution $Y \mid X = x$ follows a parametric family called extended generalized Pareto (EGP) distribution.

The proposed ERF method is part of the class of extrapolation approaches that model the exceedances $Z_i$ in~\eqref{def_exc} by conditional GPD distributions. Among the numerous methods that follow this strategy we present only those from~\cite{benjamin2019} and~\cite{velthoen2021} as they turn out to be most competitive. Other existing extrapolation based methods are not flexible enough in our setting \citep{WangTsai2009, huixia2012} or do not perform well with larger noise dimensions \citep{abdelaati2011, GardesStupfler2019}.
The method from \cite{benjamin2019}, denoted by EGAM, uses generalized additive models for the parameters of a GPD distribution. Here, we model the scale and shape parameters as smooth additive functions of the covariates without interaction effects.
\cite{velthoen2021} propose the GBEX method to estimate the GPD parameters using gradient boosting \citep{frie2001, fried2002}. 
For the fitting of all competing methods, we follow the authors' recommendations.
We also consider the unconditional model as a baseline, where we fit constant GPD parameters $(\sigma, \xi)$ to the conditional exceedances~$Z_i$.

For the sake of comparability, for all extrapolation methods, i.e., ERF, GBEX, EGAM, and unconditional, we use the same exceedances $Z_i = (Y_i - \hat Q_x^{GRF}(\tau_n))_+$, which are computed from a GRF with intermediate quantile level $\tau_n = 0.8 \leq \tau$.
From Figure~S11 in Appendix~D.2 we observe that ERF is rather robust to the choice of the intermediate quantile level $\tau_n$.  
In general, the optimal choice of $\tau_n$ depends on the properties of the data \citep[][Section 3.2]{deh2006a}, and there are numerous data-driven methods for choosing the threshold, typically based on stable regions of some statistic as a function of $\tau_n$ \citep[e.g.,][Section 6.2.2]{embr2012}.
In the predictor-dependent case, approaches using discrepancy metrics have been proposed \citep{WangTsai2009,huixia2013}.

Concerning ERF, we cross-validate the minimum node size $\kappa \in \{10, 40, 100\}$ of the GRF and the penalty term $\lambda \in \{0, 0.01, 0.001\}$ of the penalized log-likelihood in~\eqref{eq:pen-mle} using the repeated cross-validation scheme described in Appendix~C. %
We leave the other tuning parameters of the random forests at their default values; see the documentation for \texttt{quantile\_forest} in~\citet{grf}.
All simulation results can be reproduced following the description and code at \url{https://github.com/nicolagnecco/erf-numerical-results}.

\subsection{Experiment 1}\label{sec:exp-1}

In this simulation study, the data follows the model of Example~\ref{ex:gen-mod} where the response variable $Y \mid \g X = \g x$ follows a Student's $t$-distribution with $\nu(\g x) \coloneqq 1/ \xi(x) =  4$ degrees of freedom and scale $s(\g x) = 1 + \I\{x_1 > 0\}$. This is the same setup as in the simulation in \citet[][Section 5]{athe2019}, except that here we use Student's $t$-distribution instead of Gaussian for the noise. There is only one signal variable $X_1$ and $p-1$ noise variables. We generate $n = 2000$ training data and consider different dimensions $p$ and quantile levels $\tau$.

We first fix the dimension $p=10$ and investigate the effect of different target quantile levels $\tau$. The left panel of Figure~\ref{fig:wrt-q} shows the $\rmise$, the square root of the MISE defined in Section~\ref{sec:setup}, for varying values of $\tau$ close to $1$. At the intermediate quantile level~$\tau_n = 0.8$ all methods show a similar performance; in fact, the extrapolation methods coincide at this level since they use the same GRF-based estimator for the intermediate quantile. As the quantile level~$\tau$ increases we observe that the performance curves diverge. The forest-based quantile regression methods, which do not explicitly use extreme value theory for tail approximations, cannot extrapolate well to extreme quantile levels. This includes the EGP Tail method that does not focus on modeling the tail. Among the extrapolation methods, the unconditional baseline does not perform well since it cannot capture the shift in the scale function. While the EGAM does better, it already suffers from the relatively large dimension of the noise variables, a fact that we discuss in detail below. By far, the best methods are ERF and GBEX. Both combine the flexibility in the predictor space with correct extrapolation originating from the GPD approximation.

\begin{figure}[!tb]
  \centering
  \includegraphics[scale=1]{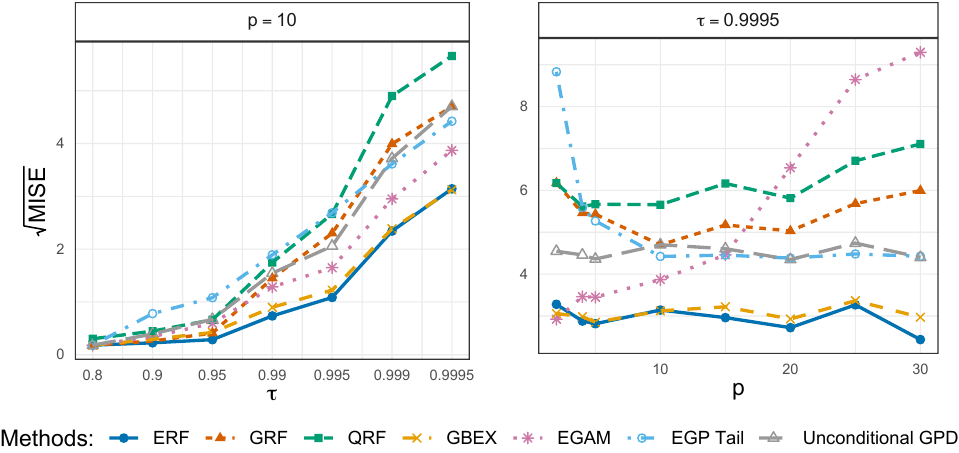}
  \caption{Square root MISE for different methods against the quantile level $\tau$ in dimension $p = 10$ (left), and against the model dimension $p$ for quantile levels $\tau=0.9995$ (right).
  }
  \label{fig:wrt-q}
\end{figure}

We next compare the performances for varying dimensions $p$ of the predictor space. The right panel of Figure~\ref{fig:wrt-q} shows the $\rmise$ as a function of $p$ for fixed quantile level $\tau = 0.9995$.
QRF and GRF are robust against growing dimensions and additional noise variables,
but the performance is not competitive for this high quantile level.
For smaller dimensions, the methods deteriorate because trees can only place splits on the signal variable $X_1$, increasing the variance.
The performance of EGAM clearly illustrates the problem of this method in large dimensions. The method cannot filter the signal from the many noise variables even though. Moreover, as mentioned by~\cite{benjamin2019}, the method becomes computationally demanding as $p$ grows. The unconditional model is unaffected by the noise dimension since it does not use the predictor values.
Both ERF and GBEX combine the advantages of the two types of approaches. They are both robust against additional noise variables and perform well even for large dimensional predictor spaces.

\subsection{Experiment 2}
The second experiment investigates the robustness of the quantile regression methods against
noise distributions with different tail heaviness in a large dimension.
The simulation setup is similar to the previous section and the data follows the model of Example~\ref{ex:gen-mod}, where we set $p=40$. We simulate data for noise distributions with shape parameters $\xi = 0, 1/4, 1/3$, where for the light-tailed case $\xi = 0$ we choose a Gaussian distribution and otherwise a Student's $t$ distribution with $\xi = 1/4, 1/3$ corresponding $v = 4, 3$ degrees of freedom, respectively. We exclude EGAM in this experiment since its performance decreases for large $p$ and it becomes computationally prohibitive (see Figure~\ref{fig:wrt-q}).

Figure~\ref{fig:boxplots-step} shows boxplots of the $\rise$ for the extreme quantile level $\tau = 0.9995$ for the different methods and different shape parameters. The triangles correspond to the average values. To make the plot easier to visualize, we remove large outliers of GRF and QRF.
The picture is similar for the three noise distributions. We observe that ERF performs very well also in the Gaussian case. Since our method relies on the GPD, estimation is not restricted to positive shape parameters, as opposed to approaches based on the Hill estimator \citep[e.g.,][]{huixia2012, huixia2013}.
Unsurprisingly, as the noise becomes very heavy-tailed (right-hand side of Figure~\ref{fig:boxplots-step}) the performances of all methods become closer since the problem becomes increasingly difficult.
Note that the performance of both QRF and GRF degrades for large values of $\xi$ and they exhibit increasingly large outliers resulting in an average exceeding the upper quartile. This underlines that classical methods without proper extrapolation are insufficient for extreme quantile regression.

\begin{figure}[!tb]
  \centering
  \hspace*{-.5in}
  \includegraphics[scale=1]{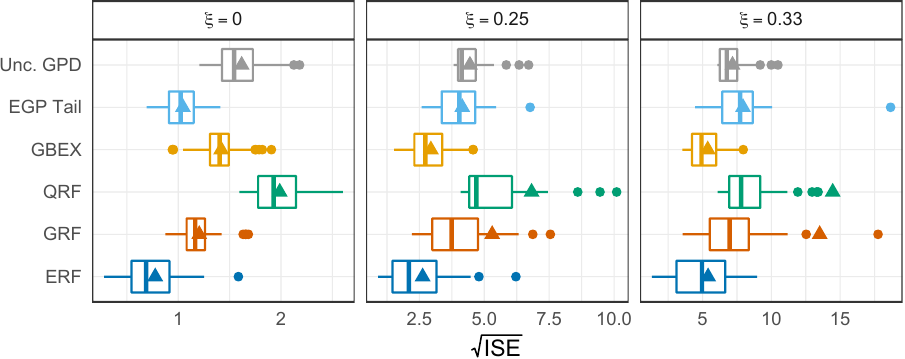}
  \caption{Boxplots of $\rise$ over $m = 50$ simulations, for different tail indices in the noise distribution at the quantile level $\tau = 0.9995$.
    The predictor space dimension is $p = 40$.
    Triangles represent the average values.
  }
  \label{fig:boxplots-step}
\end{figure}

\section{Analysis of the U.S.~Wage Structure}
\label{sec:application}

We compare the performance of ERF, GBEX, GRF, and the unconditional GPD on the U.S. census microdata for the year 1980~\citep{angrist2009}.
As described therein, the data set consists of 65,023 U.S.-born black and white men of age between 40--49, with five to twenty years of education, and with positive annual earnings and hours worked in the year before the census.
The large number of observations makes this dataset suitable to assess the performance of  the different methods at very high quantile levels.
The response~$Y$ describes the
weekly wage, expressed in 1989 U.S. dollars computed as the annual income divided by the number of weeks worked. The predictor vector consists of the numerical variables age and years of education and the categorical predictor whether the person is black or white.
To have a predictor space with larger dimension, we add ten random predictors sampled independently and uniformly on $[-1, 1]$, resulting in an overall dimension $p = 13$.

We fit ERF
repeating three times 5-fold cross-validation to tune the minimum node size $\kappa \in \{5, 40, 100\}$. To stabilize the variance of the shape parameter, we set the penalty $\lambda = 0.01$. We use the same tuning parameter setup as in~\ref{sec:competing-methods} for the other methods. In particular, we use GRF to predict the intermediate conditional quantiles at level $\tau_n = 0.8$ for all extrapolation-based methods.
We split the original data into two halves of 32,511 and 32,512 samples, and we use the first portion to perform exploratory data analysis and the second one to fit and evaluate the different methods.

For the exploratory data analysis, we fit ERF on a random subset made of 10\% of the data (i.e., 3,251 observations), and predict the GPD parameters $\hat\theta(x) = (\hat\sigma(x), \hat\xi(x))$ on the left-out observations (i.e., 29,260 observations).
Figure~\ref{fig:erf-params} shows the estimated GPD parameters $\hat\theta(x)$ as a function of years of education.
The scale parameter depends positively on years of education, whereas it is quite homogeneous between the black and white groups. In particular, it has a clear jump around 15-16 years of education, which corresponds to the end of undergraduate studies.
The shape parameter is relatively homogeneous for the black and white groups and looks stable for education. It ranges between 0.22 and 0.24, indicating heavy tails throughout the predictor space.
Moreover, Figure~S13 %
in Appendix~E.1 %
shows that the scale and shape parameters do not seem to depend on the predictor age.
\begin{figure}[!tb]
  \centering
  \includegraphics[scale=.8]{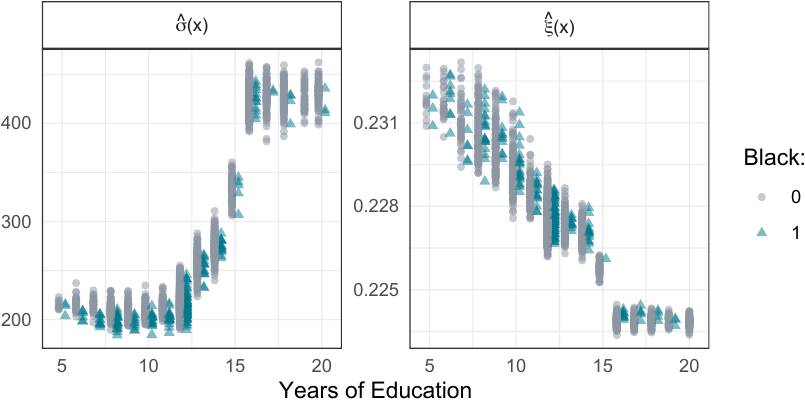}
  \caption{Estimated GPD parameters $\hat\theta(x)$ as a function of the years of education for the black (triangles) and white (circles) subgroups.
  }
  \label{fig:erf-params}
\end{figure}

Figure~\ref{fig:erf-predicts} compares the ERF quantile predictions to those of the other methods at levels $\tau = 0.9, 0.995$.
We removed all the quantiles above 6,000 predicted by GRF.
The extrapolation methods retain a good shape of the quantile function even for high levels. This does not hold for GRF, whose profile worsens as $\tau$ increases, and the discrete structure of the largest training observations becomes visible.
The unconditional method captures the variability of the conditional quantiles for $\tau = 0.9$, but it loses flexibility for larger values of $\tau$. The reason for this is that the unconditional method cannot produce different scale parameters of the GPD, while Figure~\ref{fig:erf-params} indicates that this is necessary for this data set.
ERF and GBEX model well the variability of the conditional quantiles for all values of $\tau$.
\begin{figure}[!tb]
  \centering
  \hspace*{-.5in}
  \includegraphics[scale=.8]{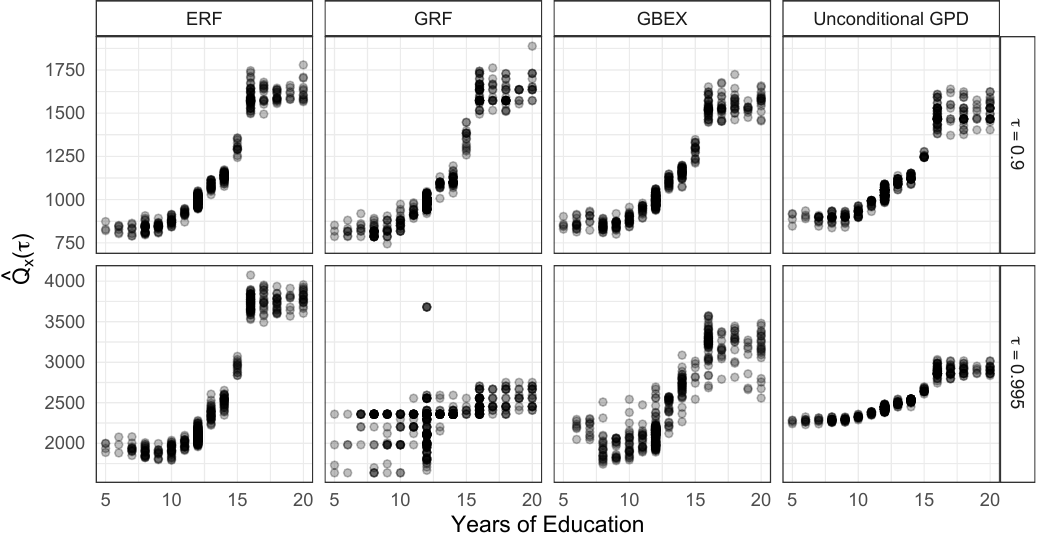}
  \caption{Predicted quantiles at levels $\tau = 0.9, 0.995$ for ERF, GRF, GBEX, and the unconditional method.
  }
  \label{fig:erf-predicts}
\end{figure}
After the exploratory analysis, we assess the quantitative performance of ERF and the other methods.
We consider the prediction metric proposed by~\citet{huixia2013},
\begin{align}\label{eq:wang-loss}
  \mathcal R_{n}\left(\hat Q_{\cdot}(\tau)\right) : =\frac{\sum_{i = 1}^{n}\I \{ Y_i < \hat  Q_{X_i}(\tau)\} - n\tau}{\sqrt{n\tau (1 - \tau)}},
\end{align}
where $n$ is the number of test observations, and $\hat Q_{\cdot}(\tau)$ is the $\tau$-th conditional quantile estimated on the training data set. This metric compares the normalized estimated proportion of observations with $Y_i < \hat Q_{X_i}(\tau)$ with the theoretical level $\tau$.
Using the true quantile function $Q_\cdot(\tau)$,
the random variable $\I\{Y_i < Q_{X_i}(\tau)\}$ follows a Bernoulli distribution with expectation $\tau$ and variance $\tau (1 - \tau)$, and by the central limit theorem the metric with oracle quantile function $\mathcal R_n(Q_\cdot(\tau))$ is asymptotically standard normal.
We partition the 32,512 observations not used in the exploratory analysis into ten random folds. On each fold, we fit the different methods and evaluate them on the left-out observations, using the absolute value of~\eqref{eq:wang-loss}. Unlike classical cross-validation, we fit the methods using a single fold and validate them on the remaining ones; this allows us to have enough observations to gauge their performance for high quantile levels $\tau$.
Figure~\ref{fig:wage-cv} shows the performance of ERF, GRF, GBEX, and the unconditional method over the ten repetitions for different quantile levels.
The shaded area represents the 95\% interval of the absolute value of a standard normal distribution, corresponding to the 95\% confidence level of the oracle method with true quantile function.
We observe that both ERF and GBEX have very good performance compared to the oracle for increasing quantile levels, and they outperform the unconditional method for large values of $\tau$. This is because they are flexible to model the scale and shape as a function of the predictors, unlike the unconditional method.
While GRF performs well for the quantile level $\tau = 0.9$, it worsens quite quickly for larger values of~$\tau$. This is expected since GRF does not rely on extrapolation results from extreme value theory and cannot accurately predict very high quantiles.

For the same data set, \cite{angrist2009b} consider the natural logarithm of the wage as a response variable for quantile regression with fixed, non-extreme quantile levels. In Appendix~E.1 %
we perform our analysis above for extreme quantiles again with this log-transformed response since it highlights several interesting properties of the ERF algorithm. Figure~S15 %
in Appendix~E.2 %
shows that the flexible methods ERF and GBEX have the desirable property that the predictions do not change much under marginal transformations. The unconditional method, on the other hand, seems to be sensitive to marginal transformations; see Appendix~E.1 %
for details. We thus advise to use flexible extrapolation methods such as ERF or GBEX that perform well on any marginal distributions.
\begin{figure}[!tb]
  \centering
  \hspace*{-0.75in}
  \includegraphics[scale=.8]{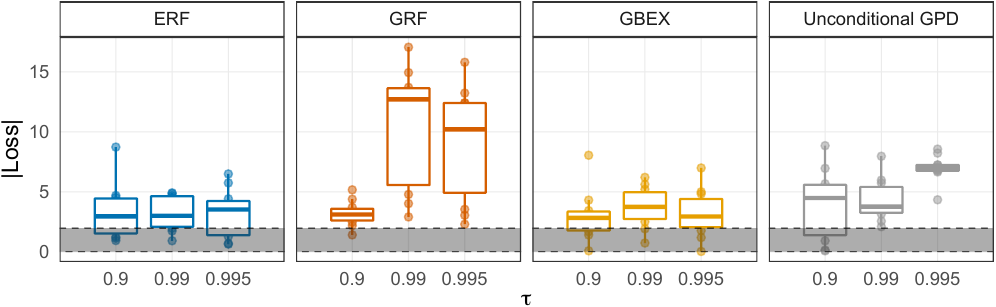}
  \caption{Absolute value of the loss~\eqref{eq:wang-loss} for the different methods fitted on the original response of the U.S.~wage data. The shaded area represents the 95\% interval of the absolute value of a standard normal distribution.}
  \label{fig:wage-cv}
\end{figure}

\section*{Acknowledgements}
We thank Alberto Quaini, Stanislav Volgushev and Chen Zhou for helpful discussions. We are also grateful to the editorial team, two anonymous referees, and the code referee for comments which helped us
to significantly improve the paper.
SE was supported by a research grant (186858) from the Swiss National Science Foundation~(SNSF). NG was supported by a research grant (210976) from the SNSF.

\begin{appendix}
  \section{Proofs}\label{app:proof2}

  \subsection{Random forests}\label{para:forest}
  Here we recall the main facts of the random forests proposed by~\citet{athe2019} in their Specification~1.
  The forest is honest and built via subsampling as follows.
  Each tree $b = 1, \dots, B$ in the forest is built as follows.
  Subsample without replacement $\mathcal{S}_b \subseteq \{1,\dots, n\}$ observations such that $|\mathcal{S}_b| = s < n$, with $s \to \infty$ and $s/n \to 0$ as $n \to \infty$.
  Partition $\mathcal{S}_b = \mathcal{I}_b \cup \mathcal{J}_b$, where $\mathcal{I}_b \cap \mathcal{J}_b = \varnothing$ and $|\mathcal{I}_b| = \lfloor s/2 \rfloor$ and  $|\mathcal{J}_b| = \lceil s/2 \rceil$. The observations in $\mathcal{J}_b$ are used to split the predictor space to construct the final leaves $L_{b}(x)$, for all $x \in \mathcal{X}$.
  The observations in $\mathcal{I}_b$ are used to make predictions.
  Furthermore, the forest consists of $B_n = {n\choose s}$ trees fitted on all possible subsamples of size $s$.
  All trees in the forest are symmetric, in the sense that they are invariant to
  permuting the indices of training observations.
  Moreover, they make balanced splits, in the sense that
  every split puts at least a fraction $\omega$ of the observations in the parent node into each child, for some $\omega > 0$.
  They are randomized in such a way that, at every split,
  the probability that the tree splits on the $j$-th feature is bounded from below by some $\pi > 0$.

In practice, one builds a forest by growing a fixed number of trees on subsamples of size $s < n$.
The following results instead hold for forests made of $n\choose s$ trees fitted on all possible subsamples of size~$s$. Similarly to \citet{wager2018estimation}, we assume that $B$ is large enough so that the Monte Carlo effect is negligible.

  We recall the main definitions of similarity weights for a forest and the underlying trees.
  For a given predictor value $x \in \mathcal{X}$, the forest similarity weights are defined by
  \begin{align*}
    w_i(x) \coloneqq \frac{1}{B} \sum_{b = 1}^{B} w_{i, b}(x),
  \end{align*}
  where $w_{i, b}(x)$ are the weights of the underlying trees $b = 1, \dots, B = {n\choose k}$.
  For each tree, the corresponding weights are defined by
  \begin{align}\label{eq:weights}
    w_{i, b}(x) \coloneqq & \ \frac{\1\left\{ X_i \in L_b(x), i \in \mathcal{I}_b \right\}}{|L_b(x)|}, \\
    |L_b(x)| \coloneqq         & \ \sum_{i = 1}^{n} \1\{X_i \in L_b(x), i \in \mathcal{I}_b\}.
  \end{align}
  Each tree is constructed such that each leaf contains between $\kappa$ and $2\kappa-1$ observations. Therefore, the leaf size $|L_b(x)|$ is always non-zero.

  Here, we restate a result about the diameter of the leaf of a single tree, which is defined as $\diam(L_b(x)) \coloneqq \sup\{\norm{y - x}_2 : y \in L_b(x)\}$. It can be found in the proof of Theorem~3 in \citet{wager2018estimation}.  

  \begin{lemma}[Leaf's diameter convergence in probability]
    Let $b = 1, \dots, B$ denote a tree in the forest and let $x \in \mathcal{X}$ be a fixed predictor point. 
    Let $s < n$ denote the number of observations used to grow the tree that satisfy~(3.6). %
    Then, for $s$ large enough, the diameter of the leaf $L_b(x)$ satisfies 
    \begin{align}\label{eq:diam}
      \P\left[ \diam\left( L_b(x) \right) > C_1 s^{-0.51 C_3}\right] \leq C_2s^{-0.50C_3},
    \end{align}
    where $C_1$, $C_2$ are positive constants depending on the parameters of the Specification~1 of~\citet{athe2019}, and 
    \begin{equation}
      \label{eq:c3}
      C_3 \coloneqq \frac{\log((1-\omega)^{-1})}{\log(\omega^{-1})}\frac{\pi}{p},\ 0 \leq \omega \leq 0.2. 
    \end{equation}
  \end{lemma}

  As a simple corollary, we can upper bound the expectation of the diameter of a leaf.

  \begin{cor}[Leaf's diameter convergence in expectation]\label{cor:diam}
    For $s$ large enough, the expected value of the diameter of the leaf satisfies
    \begin{align}
      \E&\left[ \diam\left( L_b(x) \right) \right] = \mathcal{O}\left( s^{-0.5 C_3} \right),
    \end{align}
  where $C_3$ is defined in~\eqref{eq:c3}.
  \end{cor}

  \begin{proof}
    We can write
    \begin{align}\label{eq:exp-diam}
      \begin{split}
        \E&\left[ \diam\left( L_b(x) \right) \right]\\
        = &\ \E\left[ \diam\left( L_b(x) \right) \mid \diam\left( L_b(x) \right) > C_1 s^{-0.51 C_3}\right] \P\left( \diam\left( L_b(x) \right) > C_1 s^{-0.51 C_3}  \right)\\
        & + \E\left[ \diam\left( L_b(x) \right) \mid \diam\left( L_b(x) \right) \leq C_1 s^{-0.51 C_3}\right] \P\left( \diam\left( L_b(x) \right) \leq C_1 s^{-0.51 C_3}  \right)\\
        \leq&\
        |\mathcal{X}| C_2 s^{-0.50 C_3} + C_1 s^{-0.51 C_3},
      \end{split}
    \end{align}
    where $|\mathcal{X}| < \infty$ is the area of the compact predictor space,
    and $C_1$, $C_2$ are positive constants depending on the parameters of the Specification~1 of~\citet{athe2019}.
  \end{proof}

  Here, we restate a result from~\citet{wager2018estimation} who show that the variance of a forest $T_n(x)$ is at most $s/n$ times the variance of a tree $T_{n, b}(x)$.

  \begin{lemma}[Variance of a forest]\label{lem:varforest}
  Let $x \in \mathcal{X}$ denote a fixed predictor point and let $s < n$ denote the number of observations used to grow the tree that satisfy~(3.6). %
  Denote by $T_n(x)$ a forest grown according to Specification~1 (see Section~\ref{para:forest}), and by $T_{n, b}(x)$ a tree of the forest, for $b = 1, \dots, B_n = {n\choose s}$.
  Then, the variance of a forest $T_n(x)$ is at most $s/n$ times the variance of a tree $T_{n, b}(x)$, that is
    \begin{align}\label{eq:varforest3}
      \limsup_{n\to\infty} \frac{n}{s}\frac{\V\left[T_{n}(x)\right]}{\V\left[T_{n, b}(x)\right]} \leq 1.
    \end{align}
  \end{lemma}

  \subsection{Proof of Theorem~1}\label{proof:consistency-new}
  The proof is inspired by \citep[][proof of Theorem~2.1]{zhou2009} who showed consistency of the maximum likelihood estimator for the GPD in the unconditional case.
  The main technical difficulty and difference with the proof from~\citet{zhou2009} is to show consistency of the terms in Propositions~\ref{prop:hill}--\ref{prop:gn-tilde}.
  Here, we deal with estimators that are localized in the predictor space using similarity weights estimated with a generalized random forest \citep{athe2019}.

  \begin{proof}
    Fix the predictor value $x \in \mathcal{X}$, and recall the weighted negative log-likelihood in~(3.3) %
    defined as
    \begin{align*}
      L_n(\theta; x) \coloneqq \sum_{i = 1}^{n} w_i(x) \ell_\theta(Y_i - \hat{Q}_{X_i}(\tau_n)) 1\{Y_i > \hat{Q}_x(\tau_n)\},
    \end{align*}
    where
    \begin{align*}
      \ell_\theta(z) =   \log \sigma + \left(1 + \frac{1}{\xi}\right)\log \left(1 + \frac{\xi}{\sigma}z\right),\quad z > 0,
    \end{align*}
    and $w_i(x) \coloneqq w_n(x, X_i)$.
   To compute the local minimum over $\Theta$, consider the first order conditions $\nabla \ell_\theta(z) = 0$, that is
    \begin{align*}
      \partial_\sigma \ell_\theta(z) \coloneqq &\
      \frac{1}{\sigma}  - \frac{1}{\sigma} \left( \frac{1 + \xi}{\xi}  \right)\frac{\xi/\sigma z}{1 + \xi/\sigma z} = 0
      \\
      \Rightarrow &\  \left( \frac{1 + \xi}{\xi}  \right)\frac{\xi/\sigma z}{1 + \xi/\sigma z} = 1,\numberthis \label{eq:foc-sigma}
      \\
      \Rightarrow &\ \frac{1}{1 + \xi/\sigma z} = \frac{1}{1 + \xi},
    \end{align*}
    and
    \begin{align*}
      \partial_\xi \ell_\theta(z) \coloneqq &\
      -\frac{1}{\xi^2} \log\left( 1 + \frac{\xi}{\sigma}z \right) + \left(\frac{1 + \xi}{\xi} \right)\frac{z/\sigma}{1 + \xi/\sigma z} = 0,
      \\
      \Rightarrow &\
      \frac{1}{\xi} \log\left( 1 + \frac{\xi}{\sigma}z \right) = \left(\frac{1 + \xi}{\xi} \right)\frac{\xi/\sigma z}{1 + \xi/\sigma z} = 1, \numberthis\label{eq:foc-xi}
      \\
      \Rightarrow &\
      \log\left( 1 + \frac{\xi}{\sigma}z \right) = \xi,
    \end{align*}
    where in~\eqref{eq:foc-xi} we used~\eqref{eq:foc-sigma}.
    Therefore, from~\eqref{eq:foc-sigma} and~\eqref{eq:foc-xi}, since $\xi > 0$, we can express the first order conditions  $\nabla L_n(\theta; x) = 0$ as
    \begin{align}\label{eq:focs}
      \begin{split}
        \sum_{i = 1}^{n} &
        \tilde w_{i}(x) \1\left\{ Z_i > 0 \right\}\log\left(1 + \frac{\xi}{\sigma}Z_i\right) = \xi,\\
        \sum_{i = 1}^{n} & \tilde w_{i}(x) \1\left\{ Z_i > 0 \right\}\frac{1}{1 + \xi/\sigma Z_i} = \frac{1}{\xi + 1},
      \end{split}
    \end{align}
    where, for all $i = 1, \dots, n$, $x \in \mathcal{X}$, we define
    \begin{align}\label{eq:zi-wi}
       & Z_i \coloneqq  Y_i - \hat Q_{X_i}(\tau_n),                                                     \\
       & \tilde w_{i}(x) \coloneqq  \frac{w_i(x)}{\sum_{j = 1}^{n} w_j(x) \1\{Z_j > 0\}}.
    \end{align}
    The bivariate search for zeros over $\Theta$ in~\eqref{eq:focs} can be cast to a univariate search using the parametrization $t = \xi/\sigma$ proposed by~\citet{davison1984}.
    Define the functions
    \begin{align}
        f_n(t) \coloneqq & \sum_{i = 1}^{n}
        \tilde w_{i}(x)\1\left\{ Z_i > 0 \right\}\log\left(1 + t Z_i\right) + 1,\label{eq:fn-def}\\
        g_n(t) \coloneqq & \sum_{i = 1}^{n}  \tilde w_{i}(x)\1\left\{ Z_i > 0\right\} \frac{1}{1 + t Z_i},
        \label{eq:gn-def}
    \end{align}
    where $t > 0$. Then, \citet{grimshaw1993} proposes to solve the equations in~\eqref{eq:focs} as follows.
    \begin{enumerate}
      \item Find a non-zero root $t_{n, x}^*$
      of $h_n(t) \coloneqq f_n(t)g_n(t) - 1$;
      \item\label{item:2} Define the estimator of the shape parameter $\hat{\xi}(x) \coloneqq f_n(t_{n, x}^*) -1$;
      \item\label{item:3} Define the estimator of the scale parameter $\hat\sigma(x) = \hat{\xi}(x) / t_{n, x}^*$. 
    \end{enumerate}
      The proof follows the one from~\citep[][see proof of Theorem~2.1]{zhou2009} and is split into two parts. In the first part, we show the existence of a solution $t_{n, x}^*$ with probability converging to 1 as $n\to \infty$.
      In the second part, we show that by plugging $t_{n, x}^*$ into Steps~\ref{item:2} and~\ref{item:3} consistently estimates the parameters $\xi(x)$ and $\sigma(x)$.

      Before starting with the proof, we define the quantity
    \begin{align}
      \tilde{g}_n \coloneqq \sum_{i = 1}^{n} \tilde{w}_i(x) \1\{Z_i > 0\}\left( \frac{1}{1 + Z_i / \hat{Q}_x(\tau_n)} \right)^2,
      \label{eq:gn-tilde}
    \end{align}
    and, for $\delta > 0$, the functions
    \begin{align}
      D_1(\delta) \coloneqq &\ 
      \frac{\xi(x)}{\left( 1 + \xi(x) \right)^2} - \frac{\xi(x)}{(2\xi(x) + 1)(1 + \delta)},
      \label{eq:d1}
      \\
      D_2(\delta) \coloneqq &\ 
      \frac{\log(1 - \delta)}{\delta} \frac{\xi(x)}{\left( 1 + \xi(x) \right)^2} +
      \frac{\xi(x)}{2\xi(x) + 1}.
      \label{eq:d2}
    \end{align}

      We now show existence of a solution with probability converging to 1.
      First, fix an arbitrary $\delta \in (0, 1/2)$ satisfying $D_1(\delta) < 0$ and $D_2(\delta) > 0$.
      By following~\citet{zhou2009}, consider the approximate solution 
      \begin{align}\label{eq:approx-sol}
        t_{n, x} \coloneqq \frac{\xi(x)}{\xi(x)\hat{Q}_x(\tau_n)} = \frac{1}{\hat{Q}_x(\tau_n)},
      \end{align}
    motivated by the fact that when $\xi(x) > 0$ it holds that $\sigma(x) \sim \xi(x) Q_{x}(\tau_n)$ as $n \to \infty$.
    Moreover, define the perturbed solutions 
    \begin{align}\label{eq:perturbed-sol}
      t_{n, x}^{(\delta)} \coloneqq &\ \frac{1 + \delta}{\hat{Q}_x(\tau_n)},\quad
      t_{n, x}^{(-\delta)} \coloneqq  \frac{1 - \delta}{\hat{Q}_x(\tau_n)}.
    \end{align}
    Then, following~\citep[][Equations~(14) and~(15)]{zhou2009}, for any $\delta_n \in (0,\delta)$ we can bound the function $f_n$ by
    \begin{align}\label{eq:bound-fn}
    \begin{split}
      f_n(t_{n, x}^{(\delta_n)}) 
      <&\ f_n(t_{n, x}) + \delta_n \left( 1 - g_n(t_{n, x}) \right),
      \\
      f_n(t_{n, x}^{(-\delta_n)})
      >&\ f_n(t_{n, x}) + \frac{\log(1-\delta)}{\delta}\delta_n(1 - g_n(t_{n, x})),
    \end{split}
    \end{align}
    and the function $g_n$ by
    \begin{align}\label{eq:bound-gn}
    \begin{split}
      g_n(t_{n, x}^{(\delta_n)}) 
      <&\ g_n(t_{n, x}) - \frac{\delta_n}{1 + \delta}(g_n(t_{n, x}) - \tilde{g}_n),
      \\
      g_n(t_{n, x}^{(-\delta_n)}) 
      >&\ g_n(t_{n, x}) + \delta_n(g_n(t_{n, x}) - \tilde{g}_n).
    \end{split}
    \end{align}
    Hence, for any $\delta_n \in (0,\delta)$, from~\eqref{eq:bound-fn} and~\eqref{eq:bound-gn} and from the definition of $h_n$ we have that
  \begin{align}\label{eq:h-bound}
  \begin{split}
    h_n(t_{n, x}^{(\delta_n)})
    <&\ f_n(t_{n, x})g_n(t_{n, x}) - 1 + \delta_n D_{1, n},
    \\
    h_n(t_{n, x}^{(-\delta_n)})
    >&\ f_n(t_{n, x})g_n(t_{n, x}) -1 + \delta_n D_{2, n},
  \end{split}
  \end{align}
  where we define
  \begin{align}\label{eq:d1n-d2n}
    D_{1, n} \coloneqq &\ g_n(t_{n, x})(1 - g_n(t_{n, x})) - f_n(t_{n, x})\frac{g_n(t_{n, x}) - \tilde{g}_n}{1 + \delta},
    \\
    D_{2, n} \coloneqq &\ \frac{\log(1 - \delta)}{\delta} g_n(t_{n, x})(1 - g_n(t_{n, x})) + f_n(t_{n, x})(g_n(t_{n, x}) - \tilde{g}_n).
  \end{align}
  From Proposition~\ref{prop:hill}, it holds that $f_n(t_{n, x}) \stackrel{\P}{\to} 1 + \xi(x)$. From Proposition~\ref{prop:gn}, it holds that $g_n(t_{n, x})\stackrel{\P}{\to} (1 + \xi(x))^{-1}$.
  From Proposition~\ref{prop:gn-tilde}, it holds that $\tilde{g}_n \stackrel{\P}{\to}(2\xi(x) + 1)^{-1}$.
  Thus, from the continuous mapping theorem, we have that $D_{1, n} \stackrel{\P}{\to} D_{1}(\delta) < 0$ and $D_{2, n} \stackrel{\P}{\to} D_{2}(\delta) > 0$.
  Define the sequence
  \begin{align}\label{eq:def-dn}
    \delta_n \coloneqq |f_n(t_{n, x})g_n(t_{n, x}) - 1| \max\left\{ -\frac{1}{D_{1, n}}, \frac{1}{D_{2, n}}\right\},
  \end{align}
  and note that $\delta_n \stackrel{\P}{\to} 0$.
  We are now ready to show that the probability of having a solution to $h_n(t) = 0$ in the interval $[t_{n, x}^{(-\delta_n)}, t_{n, x}^{(\delta_n)}]$ converges to 1, i.e., 
  \begin{align}\label{eq:def-existence}
    \P\left(\left\{  h_n(t) = 0\ \text{for some}\ t \in [t_{n, x}^{(-\delta_n)}, t_{n, x}^{(\delta_n)}] \right\} \right) \to 1,\ \text{as}\ n \to \infty.
  \end{align}
  Define the event $E_n \coloneqq \left\{ D_{1, n} < 0,\ D_{2, n} > 0,\ 0 < \delta_n < \delta \right\}$ and note that $P(E_n) \to 1$ as $n \to \infty$.
  On this event, from~\eqref{eq:h-bound} and~\eqref{eq:def-dn}, we have that
  \begin{align*}
    h_n(t_{n, x}^{(\delta_n)})
    <&\ \left( f_n(t_{n, x})g_n(t_{n, x}) - 1 \right) - |f_n(t_{n, x})g_n(t_{n, x}) - 1| \leq 0,
    \\
    h_n(t_{n, x}^{(-\delta_n)})
    >&\ \left( f_n(t_{n, x})g_n(t_{n, x}) -1 \right) + |f_n(t_{n, x})g_n(t_{n, x}) - 1| \geq 0.
  \end{align*}
  Thus, on the event $E_n$ there exists a solution $t_{n, x}^*$ lying in the interval $[t_{n, x}^{(-\delta_n)}, t_{n, x}^{(\delta_n)}]$, and therefore~\eqref{eq:def-existence} holds.

  We now show that the estimators $\hat{\xi}(x) \coloneqq f_n(t_{n, x}^*)-1$ and $\hat\sigma(x) = \hat\xi(x)/t_{n, x}^*$ are consistent.
  Since $t\mapsto f_n(t)$ is an increasing function, using~\eqref{eq:bound-fn} we have that
  \begin{align*}
    f_n(t_{n, x}) + \frac{\log(1-\delta)}{\delta}\delta_n(1 - g_n(t_{n, x})) < &\ f_n(t_{n, x}^{(-\delta_n)})\\
    \leq &\ f_n(t_{n, x}^*)\\
    \leq &\ f_n(t_{n, x}^{(\delta_n)})
    < f_n(t_{n, x}) + \delta_n \left( 1 - g_n(t_{n, x}) \right).
  \end{align*}
  From the consistency of $f_n(t_{n, x})$, $g_n(t_{n, x})$ and $\delta_n$, the continuous mapping theorem implies that
  $f_n(t_{n, x}^*) \stackrel{\P}{\to} \xi(x) + 1$, and therefore $\hat\xi(x) \stackrel{\P}{\to}\xi(x)$.
  Consider now $\hat\sigma(x) \coloneqq \hat\xi(x)/t_{n, x}^*$ which can be bounded by
  \begin{align*}
    \frac{\hat\xi(x)\hat{Q}_x(\tau_n)}{1 + \delta_n}=\frac{\hat\xi(x)}{t_{n, x}^{(\delta_n)}}< \hat\sigma(x) < \frac{\hat\xi(x)}{t_{n, x}^{(-\delta_n)}}
    =\frac{\hat\xi(x)\hat{Q}_x(\tau_n)}{1 - \delta_n}.
  \end{align*}
  Therefore, from the consistency of $\hat\xi(x)$ and the consistency of $\hat{Q}_x(\tau_n)$ from Assumption~2, %
  we have that
  \begin{align*}
    \frac{1}{1 + \delta_n}\frac{\hat\xi(x)\hat{Q}_x(\tau_n)}{\xi(x)Q_x(\tau_n)}
    < \frac{\hat\sigma(x)}{\xi(x)Q_x(\tau_n)} 
    < \frac{1}{1 - \delta_n}\frac{\hat\xi(x)\hat{Q}_x(\tau_n)}{\xi(x)Q_x(\tau_n)},
  \end{align*}
  which implies that $\hat\sigma(x)/(\xi(x)Q_x(\tau_n)) \stackrel{\P}{\to} 1$.
  \end{proof}

  \subsection{Proof of Theorem~2}\label{proof:consistency-pen}
  \begin{proof}
    Fix the predictor value $x \in \mathcal{X}$.
    The first order conditions of the penalized log-likelihood--see~(3.9)--%
    can be simplified to
    \begin{align}\label{eq:focs-pen}
      \begin{split}
        \sum_{i = 1}^{n} &
        \tilde w_{i}(x) \1\left\{ Z_i > 0 \right\}\log\left(1 + \frac{\xi}{\sigma}Z_i\right) = \xi + 2\lambda_n\frac{\xi^2 (\xi - \hat\xi_n)}{T_n(x)},\\
        \sum_{i = 1}^{n} & \tilde w_{i}(x) \1\left\{ Z_i > 0 \right\}\frac{1}{1 + \xi/\sigma Z_i} = \frac{1}{\xi + 1},
      \end{split}
    \end{align}
    where $Z_i$ and $w_i(x)$ are defined in~\eqref{eq:zi-wi} and 
    \begin{align*}
       T_n(x) \coloneqq \frac{n}{k}\sum_{i = 1}^{n} w_i(x) \1\{Z_i > 0\}.
    \end{align*}
    The bivariate search for zeros over $\Theta$ in~\eqref{eq:focs-pen} can be cast to a univariate search using the parametrization $t = \xi/\sigma$ proposed by~\citet{davison1984}.
    Define the functions $t \mapsto f_n(t)$ and $t \mapsto g_n(t)$ as in~\eqref{eq:fn-def}, and let $h_n(t) \coloneqq f_n(t)g_n(t) - 1$ where $t > 0$. Then, following~\citet{grimshaw1993}, we solve the equations in~\eqref{eq:focs-pen} as follows.
    \begin{enumerate}
      \item Find a $t_{n, x}^*$ satisfying 
      \begin{align}\label{eq:h-pen}
        h_n(t_{n, x}^*) = 2\lambda_n \frac{\xi^2 (\xi - \hat\xi_n)}{(1 + \xi)T_n(x)};
      \end{align}
      \item\label{item:2-pen} Define the estimator of the shape parameter $\hat\xi_{\mathrm{pen}}(x) \coloneqq 1 / g_n(t_{n, x}^*) - 1$;
      \item\label{item:3-pen} Define the estimator of the scale parameter $\hat\sigma_\mathrm{pen}(x) = \hat\xi_{\mathrm{pen}}(x) / t_{n, x}^*$. 
    \end{enumerate}
    Let $t_{n, x}$ denote the approximate solution defined in~\eqref{eq:approx-sol}, and for any $\delta \in (0, 1/2)$ let $t_{n, x}^{(\delta)}$ and $t_{n, x}^{(-\delta)}$ denote the perturbed solutions defined in~\eqref{eq:perturbed-sol}.

    Fix $\delta \in (0, 1/2)$ such that $D_1(\delta) < 0$ and $D_2(\delta) > 0$, where $D_1(\delta)$ and $D_2(\delta)$ are defined in~\eqref{eq:d1} and~\eqref{eq:d2}.
    
    By assumption, $\lambda_n = o(1)$ and $\hat\xi = \mathcal{O}_\P(1)$. Moreover, from Lemma~\ref{lem:denom} it holds that $T_n(x) = 1 + o_\P(1)$.
    Therefore, there exists a sequence $\varepsilon_n > 0$ satisfying as $n \to \infty$
    \begin{align}
      \P\left( \left|2\lambda_n \frac{\xi^2 (\xi - \hat\xi)}{(1 + \xi)T_n(x)} \right| \leq \varepsilon_n \right) \to 1.
    \end{align}
    Define the sequence,
    \begin{align}\label{eq:def-dn-pen}
      \delta_n \coloneqq \left( |f_n(t_{n, x})g_n(t_{n, x}) - 1| + \varepsilon_n \right) \max\left\{ -\frac{1}{D_{1, n}}, \frac{1}{D_{2, n}}\right\},
    \end{align}
    where $D_{1, n} \stackrel{\P}{\to} D_1(\delta) < 0$ and $D_{2, n} \stackrel{\P}{\to} D_2(\delta) > 0$ are defined in~\eqref{eq:d1n-d2n}, and note that $\delta_n \stackrel{\P}{\to} 0$.
    We now show that the probability of finding a $t$ in the interval $[t_{n, x}^{(-\delta_n)}, t_{n, x}^{(\delta_n)}]$ satisfying~\eqref{eq:h-pen} converges to 1, i.e., as $n \to \infty$
  \begin{align}\label{eq:def-existence-pen}
    \P\left(\left\{\text{there exists some}\ t \in [t_{n, x}^{(-\delta_n)}, t_{n, x}^{(\delta_n)}]\ \text{s.t.~\eqref{eq:h-pen} is satisfied} \right\} \right) \to 1.
  \end{align}
  Define the event 
  \begin{align}\label{eq:event-pen}
    E_n \coloneqq \left\{ D_{1, n} < 0,\ D_{2, n} > 0, \left|2\lambda_n \frac{\xi^2 (\xi - \hat\xi)}{(1 + \xi)T_n(x)} \right| \leq \varepsilon_n, \ 0 < \delta_n < \delta \right\},
  \end{align} and note that $P(E_n) \to 1$ as $n \to \infty$.
  On this event, from~\eqref{eq:h-bound} and~\eqref{eq:def-dn-pen}, we have that
  \begin{align*}
    h_n(t_{n, x}^{(\delta_n)})
    <&\ \left( f_n(t_{n, x})g_n(t_{n, x}) - 1 \right) - |f_n(t_{n, x})g_n(t_{n, x}) - 1| - \varepsilon_n \leq -\varepsilon_n,
    \\
    h_n(t_{n, x}^{(-\delta_n)})
    >&\ \left( f_n(t_{n, x})g_n(t_{n, x}) -1 \right) + |f_n(t_{n, x})g_n(t_{n, x}) - 1| + \varepsilon_n \geq \varepsilon_n.
  \end{align*}
  Therefore, on the event $E_n$ there exists a solution $t_{n, x}^*$ lying in the interval $[t_{n, x}^{(-\delta_n)}, t_{n, x}^{(\delta_n)}]$ that satisfies~\eqref{eq:h-pen}, and therefore~\eqref{eq:def-existence-pen} holds.

  We now show that the estimators $\hat\xi_{\mathrm{pen}}(x) \coloneqq 1 / g_n(t_{n, x}^*) - 1$ and $\hat\sigma_\mathrm{pen}(x)\coloneqq \hat\xi_{\mathrm{pen}}(x) / t_{n, x}^*$ are consistent.
  Define the event $\tilde{E}_n \coloneqq E_n \cap \{\hat{Q}_x(\tau_n) > 0\}$ and notice that $\P(\tilde{E}_n) \to 1$ as $n \to \infty$ by Assumption~2 in the main text %
  and Lemma~\ref{lem:flucts-1}. Furthermore, on the event $\tilde{E}_n$, if $Z_i > 0$ we have that $a_i \coloneqq Z_i / \hat{Q}_x(\tau_n) > 0$ for all $i = 1, \dots, n$.
  Therefore, it follows that
  \begin{align}\label{eq:bound-gn-pen-1}
  \begin{split}
    g_n(t_{n, x}^{(\delta_n)}) - g_n(t_{n, x})
    = &\ \sum_{i = 1}^{n}  \tilde w_{i}(x)\1\left\{ Z_i > 0\right\}\left(  \frac{1}{1 +  a_i(1 + \delta_n)} - \frac{1}{1 + a_i} \right)
    \\
    = &\ \sum_{i = 1}^{n}  \tilde w_{i}(x)\1\left\{ Z_i > 0\right\}\frac{ - \delta_n a_i}{\left( 1 +  a_i(1 + \delta_n) \right)(1 + a_i)} > - \delta_n,
  \end{split}
  \end{align}
  and
  \begin{align}\label{eq:bound-gn-pen-2}
  \begin{split}
    g_n(t_{n, x}^{(-\delta_n)}) - g_n(t_{n, x})
    = &\ \sum_{i = 1}^{n}  \tilde w_{i}(x)\1\left\{ Z_i > 0\right\}\left(  \frac{1}{1 +  a_i(1 - \delta_n)} - \frac{1}{1 + a_i} \right)\\
    = &\ \sum_{i = 1}^{n}  \tilde w_{i}(x)\1\left\{ Z_i > 0\right\}\frac{\delta_n a_i}{\left( 1 +  a_i(1 - \delta_n) \right)(1 + a_i)} < \delta_n.
  \end{split}
  \end{align}
  Since $t\mapsto g_n(t)$ is a decreasing function, using~\eqref{eq:bound-gn-pen-1} and~\eqref{eq:bound-gn-pen-2}, on the event $\tilde{E}_n$ we have that
  \begin{align*}
    g_n(t_{n, x}) - \delta_n < &\ g_n(t_{n, x}^{(\delta_n)})
    \leq  g_n(t_{n, x}^*)
    \leq g_n(t_{n, x}^{(-\delta_n)})
    < g_n(t_{n, x}) + \delta_n.
  \end{align*}
  From the consistency of $g_n(t_{n, x})$ in Proposition~\ref{prop:gn} and the fact that $\delta_n \stackrel{\P}{\to} 0$,
  the continuous mapping theorem implies that 
  $g_n(t_{n, x}^*) \stackrel{\P}{\to} (1 + \xi(x))^{-1}$, and therefore $\hat\xi_{\mathrm{pen}}(x) \stackrel{\P}{\to}\xi(x)$.
  Consider now $\hat\sigma_{\mathrm{pen}}(x) \coloneqq \hat\xi_{\mathrm{pen}}(x)/t_{n, x}^*$ which can be bounded by 
  \begin{align*}
    \frac{\hat\xi_{\mathrm{pen}}(x)\hat{Q}_x(\tau_n)}{1 + \delta_n}
    =
    \frac{\hat\xi_{\mathrm{pen}}(x)}{t_{n, x}^{(\delta_n)}}
    < 
    \hat\sigma_{\mathrm{pen}}(x) 
    < 
    \frac{\hat\xi_{\mathrm{pen}}(x)}{t_{n, x}^{(-\delta_n)}}
    =
    \frac{\hat\xi_{\mathrm{pen}}(x)\hat{Q}_x(\tau_n)}{1 - \delta_n}.
  \end{align*}
  Therefore, from the consistency of $\hat\xi_{\mathrm{pen}}(x)$ and the consistency of $\hat{Q}_x(\tau_n)$ from Assumption~2 %
  in the main text, we have that
  \begin{align*}
    \frac{1}{1 + \delta_n}\frac{\hat\xi_{\mathrm{pen}}(x)\hat{Q}_x(\tau_n)}{\xi(x)Q_x(\tau_n)}
    < \frac{\hat\sigma_{\mathrm{pen}}(x)}{\xi(x)Q_x(\tau_n)} 
    < \frac{1}{1 - \delta_n}\frac{\hat\xi_{\mathrm{pen}}(x)\hat{Q}_x(\tau_n)}{\xi(x)Q_x(\tau_n)},
  \end{align*}
  which implies that $\hat\sigma_{\mathrm{pen}}(x)/(\xi(x)Q_x(\tau_n)) \stackrel{\P}{\to} 1$.
  \end{proof}

  \subsection{Proof of Corollary~1}\label{proof:cor-hill}
  \begin{proof}
    Fix $x \in \mathcal{X}$. By~\eqref{eq:fn}, we have that $\hat \xi_H(x) = T_2(x)(f_n(t_{n, x}) - 1)$.
    By Proposition~\ref{prop:hill}, it holds that $f_n(t_{n, x}) - 1 \stackrel{\P}{\to} \xi(x)$.
    By Lemma~\ref{lem:denom}, it holds that $T_2(x) \stackrel{\P}{\to} 1$.
    Therefore, by the continuous mapping theorem, it holds that $\hat \xi_H(x) \stackrel{\P}{\to} \xi(x)$.
  \end{proof}

  \subsection{Main results}

  \begin{proposition}[Local Hill estimator]\label{prop:hill}
    Define the approximate solution $t_{n, x} \coloneqq \xi(x) / (\xi(x) \hat{Q}_x(\tau_n))$.
    Then, it holds that
    \begin{align*}
      f_n(t_{n, x}) - 1
      \stackrel{\P}{\to} \xi(x).
    \end{align*}
  \end{proposition}

  \begin{proof}
    Fix $\eta > 0$,
    and define the event
    \begin{align}\label{eq:event}
      A_{n}(\eta) \coloneqq \left\{ \left|\frac{\hat{Q}_x(\tau_n)}{Q_x(\tau_n)} - 1\right| <  \eta\ \text{and}\ Q_x(\tau_n) > 1,\ \text{for all}\ x \in \mathcal{X} \right\},
    \end{align}
    and note that by Assumption~2 in the main text %
    and Lemma~\ref{lem:flucts-1}
    it holds that $\P(A_{n}(\eta)) \to 1$ as $n \to \infty$.
    Rewrite
    \begin{align}\label{eq:fn}
      \begin{split}
        f_n(t_{n, x}) - 1
        = &\
        \sum_{i = 1}^{n}\tilde w_i(x) \1\{Z_i > 0\}\log\left(1 + t_{n, x} Z_i\right)\\
        = &\
        \sum_{i = 1}^{n}\tilde w_i(x) \1\{Z_i > 0\}\log\left(1 + Z_i / \hat{Q}_x(\tau_n)\right)\\
        = &
        \frac{\frac{n}{k}\sum_{i = 1}^{n}  w_i(x) \1\{Z_i > 0\}\log\left(1 + Z_i / \hat{Q}_x(\tau_n)\right)}{\frac{n}{k}\sum_{i = 1}^{n} w_i(x) \1\{Z_i > 0\}}\\
        \eqqcolon & \frac{T_1(x)}{T_2(x)}.
      \end{split}
    \end{align}
    By Lemma~\ref{lem:denom}, the denominator in~\eqref{eq:fn} is such that
    \begin{align*}
      T_{2}(x)\stackrel{\P}{\to} 1.
    \end{align*}
    Therefore, in the sequel, we study the behavior of the numerator.
    Consider a fixed tree $b = 1, \dots, B$, where $B \coloneqq {n\choose s}$ making predictions at a fixed point $x \in \mathcal{X}$ with weights defined as
    \begin{align*}
      w_{i, b}(x) \coloneqq & \ \frac{\1\left\{ X_i \in L_b(x), i \in \mathcal{I}_b \right\}}{|L_b(x)|}, \\
      |L_b(x)| \coloneqq         & \ \sum_{i = 1}^{n} \1\{X_i \in L_b(x), i \in \mathcal{I}_b\},
    \end{align*}
    where $L_b(x)$ denotes the estimated leaf containing $x$ in the tree $b$ and its size $|L_b(x)|$ is always non-zero by construction (see Section~\ref{para:forest}).
    By Lemma~\ref{lem:log-bound}, for all observations satisfying $Z_i > 0$ and $w_{i, b}(x) > 0$,
    on the event $A_n(\eta)$ it holds that
    \begin{align}
      \begin{split}
        \Bigg|\log\left( 1 + Z_i/ \hat{Q}_x(\tau_n) \right) - \log\left( \frac{Y_i}{Q_x(\tau_n)} \right)\Bigg|
        \leq &\ 
        C_n \norm{X_i - x}_2
        + \log\left(\frac{1+\eta}{(1 - \eta)^2}  \right),
      \end{split}
    \end{align}
    where $C_n > 0$ is the
    sequence defined in Lemma~\ref{lem:flucts-1}.
    Therefore, it holds that 
    \begin{align}
      \P   & \left( \left| T_1(x)  - \xi(x)\right| >  \varepsilon \right) \notag \\
      \leq & \ \P\left( \left| T_1(x)  - \xi(x) \right| >  \varepsilon, A_n(\eta) \right)
      + \P\left( A_n(\eta)^c \right) \notag\\
      \leq & \ \P\left( \left| \sum_{i = 1}^{n} w_i(x) \1\{Z_i > 0\}\frac{n}{k}\log\left(\frac{Y_i }{Q_x(\tau_n)}\right) - \xi(x) \right|  > \varepsilon/3,\ A_n(\eta)\right) \tag{$I$} \label{eq:t1}\\
      & + \P\left(  \sum_{i = 1}^{n} w_i(x) \1\{Z_i > 0\} \frac{n}{k}C_n \norm{X_i - x}_2  > \varepsilon/3,\ A_n(\eta)\right) \tag{$II$}\label{eq:t2}\\
      & + \P\left( \log\left(\frac{1+\eta}{(1 - \eta)^2}  \right) \sum_{i = 1}^{n} w_i(x) \1\{Z_i > 0\} \frac{n}{k}>\varepsilon/3,\ A_n(\eta)\right) \tag{$III$}\label{eq:t3}\\
      & + \P\left( A_n(\eta)^c\right). \notag
    \end{align}

    Consider term~\eqref{eq:t1}. Recall we have the stochastic representation $Y_i \stackrel{d}{=} Q_{X_i}(U_i)$,
    where $U_i$ are standard uniform random variables independent of $X_i$, for $i = 1, \dots, n$.
    We have that 
    \begin{align*}
      \P&\left( \left| \sum_{i = 1}^{n}w_i(x) \1\{Z_i > 0\}\frac{n}{k}\log\left(\frac{Y_i }{Q_x(\tau_n)}\right) - \xi(x) \right|  > \frac{\varepsilon}{3},\ A_n(\eta) \right)\\
      \leq &\ 
      \P\left( \left| \sum_{i = 1}^{n} w_i(x) \1\{U_i > \tau_n\}\frac{n}{k}\log\left(\frac{Q_x(U_i)}{Q_x(\tau_n)}\right) - \xi(x) \right|  > \frac{\varepsilon}{6},\ A_n (\eta)\right)\tag{$IV$}\label{eq:t4} \\
      &+ \P\left( \left| \sum_{i = 1}^{n}w_i(x) \frac{n}{k}
      \left\{\1\{Z_i > 0\}\log\left(\frac{Y_i }{Q_x(\tau_n)}\right)\right.\right.\right. \tag{$V$}\label{eq:t5}\\
      &\qquad\qquad\qquad\qquad\qquad - \left.\left.\left.  \1\{U_i > \tau_n\} \log\left(\frac{Q_x(U_i)}{Q_x(\tau_n)} \right) \right\}  \right|  > \frac{\varepsilon}{6}, A_n (\eta)\right).
    \end{align*}
    Using Lemma~\ref{lem:f-4}, it holds that~\eqref{eq:t4}$\to 0$ as $n\to \infty$.
    Using Lemma~\ref{lem:f-5}, it holds that~\eqref{eq:t5}$\to 0$ as $n\to \infty$.
    Therefore, it follows that~\eqref{eq:t1}$\to 0$ as $n \to \infty$.

    Consider term~\eqref{eq:t2}. Using Lemma~\ref{lem:f-2}, it holds that~\eqref{eq:t2}$\to 0$ as $n\to \infty$.

    Consider term~\eqref{eq:t3}.
    By Lemma~\ref{lem:denom} it holds that
    \begin{align*}
      \log & \left(\frac{1+\eta}{(1 - \eta)^2}  \right)  T_2(x) \stackrel{\P}{\to} \log\left( \frac{1+\eta}{(1-\eta)^2} \right).
    \end{align*}
    Since $\eta > 0$ is arbitrary, it follows that~\eqref{eq:t3}$\to 0$ as $n \to \infty$.
    Putting everything together, we have that $f_n(t_{n, x}) - 1 \stackrel{\P}{\to} \xi(x)$.

  \end{proof}

  \begin{proposition}[$g_n$ converges in probability]\label{prop:gn}
    Define the approximate solution $t_{n, x} \coloneqq \xi(x) / (\xi(x) \hat{Q}_x(\tau_n))$.
    Then, it holds that
    \begin{align*}
      g_n(t_{n, x})
      \stackrel{\P}{\to} \frac{1}{1+\xi(x)}.
    \end{align*}    
  \end{proposition}

  \begin{proof}
    Rewrite
    \begin{align}\label{eq:gn}
      \begin{split}
        g_n(t_{n, x})
        = &\
        \sum_{i = 1}^{n}\tilde w_i(x) \1\{Z_i > 0\}\frac{1}{1 + t_{n, x} Z_i}\\
        = &\
        \sum_{i = 1}^{n}\tilde w_i(x) \1\{Z_i > 0\}\frac{1}{1 + Z_i / \hat{Q}_x(\tau_n)}\\
        = &
        \frac{\frac{n}{k}\sum_{i = 1}^{n}  w_i(x) \1\{Z_i > 0\}\frac{\hat{Q}_x(\tau_n)}{ \hat{Q}_x(\tau_n) + Z_i}}{\frac{n}{k}\sum_{i = 1}^{n} w_i(x) \1\{Z_i > 0\}}\\
        \eqqcolon & \frac{T_1(x)}{T_2(x)}.
      \end{split}
    \end{align}
    By Lemma~\ref{lem:denom}, the denominator in~\eqref{eq:fn} is such that
    \begin{align*}
      T_{2}(x)\stackrel{\P}{\to} 1.
    \end{align*}
    Therefore, in the sequel, we study the behavior of the numerator.
    We now split the numerator between those observations that are `close' to the predictor value $x \in \mathcal{X}$ and those that are not. Define $\delta_n \coloneqq C_1 s^{-0.51 C_3} \to 0$ as $n \to \infty$. 
    We rewrite
    \begin{align*}
      T_1(x)
       = &\ \sum_{i = 1}^{n} w_i(x) \1\{\norm{X_i - x}_2 \leq \delta_n \}\1\{Z_i > 0\}\frac{n}{k}\frac{1}{1 + Z_i/\hat{Q}_x(\tau_n)}\\
      &+ \sum_{i = 1}^{n} w_i(x) \1\{\norm{X_i - x}_2 > \delta_n \}\1\{Z_i > 0\}\frac{n}{k}\\
      \eqqcolon &\ \tilde{T}_1(x) + \tilde{T}_2(x),
    \end{align*}
    where $\tilde{T}_1(x) \stackrel{\P}{\to} (1 + \xi(x))^{-1}$ by Lemma~\ref{lem:gn-1} and $\tilde{T}_2(x) \stackrel{\P}{\to} 0$ by Lemma~\ref{lem:gn-2}.
    
  \end{proof}

  \begin{proposition}[$\tilde{g}_n$ converges in probability]\label{prop:gn-tilde}
    It holds that
    \begin{align*}
      \tilde{g}_n \coloneqq \sum_{i = 1}^{n} \tilde{w}_i(x) \1\{Z_i > 0\}\left( \frac{1}{1 + Z_i / \hat{Q}_x(\tau_n)} \right)^2 \stackrel{\P}{\to} \frac{1}{2\xi(x) + 1}.
    \end{align*}
  \end{proposition}
  \begin{proof}    
    The proof is similar to the proof of Proposition~\ref{prop:gn}, and we therefore omit it.
  \end{proof}

  \begin{lemma}[Denominator converges to one]\label{lem:denom}
    Let $T_2(x) = \frac{n}{k}\sum_{i = 1}^{n}w_i(x) \1\{ Y_i > \hat{Q}_{X_i}(\tau_n) \}$.
    Then, it holds that
    \begin{align*}
      T_2(x)\stackrel{\P}{\to} 1.
    \end{align*}

  \end{lemma}

  \begin{proof}
    Fix $\varepsilon, \eta > 0$. We can write
    \begin{align*}
      \P   \left( \left|T_2(x) - 1\right| > \varepsilon\right)                                                                     
      \leq & \ \P\left( \left|T_2(x) - 1\right| > \varepsilon,\ A_{n}(\eta)\right) + \P\left(A_{n}(\eta)^c  \right),
    \end{align*}
    where $A_n(\eta)$ is the event defined in~\eqref{eq:event}. 
    We want to show that $\E[T_2(x)] \to 1$ and $\V[T_2(x)] \to 0$ as $n \to \infty$ when $A_{n}(\eta)$ holds.
    On the event $A_{n}(\eta)$, we have that
    \begin{align*}
      \1\{ Y_i > {Q}_{X_i}(\tau_n) (1 + \eta)\} < \1\{ Y_i > \hat{Q}_{X_i}(\tau_n) \} < \1\{ Y_i > {Q}_{X_i}(\tau_n) (1 - \eta)\}.
    \end{align*}
    Fix a tree $b = 1, \dots, {n\choose s}$ and define
    \begin{align}\label{eq:t2-tree}
      T_{2, b}(x) \coloneqq  \sum_{i = 1}^{n}w_{i,b}(x) \frac{n}{k} \1\{ Y_i > \hat{Q}_{X_i}(\tau_n) \}.
    \end{align}

    We now consider the expectation of $T_{2, b}(x)$. Fix $\zeta > 0$.
    On the event $A_n(\eta)$, using Lemma~\ref{lem:reg-var-1}, for $n$ large enough we have 
    \begin{align*}
      (1 + \eta)^{-1/\xi_{-}-\zeta}
      < &\ \E\left[T_{2, b}(x)  \right]\\
      < &\ \E\left[ \sum_{i = 1}^{n}w_{i,b}(x) \frac{n}{k} \1\{ Y_i > {Q}_{X_i}(\tau_n)(1 - \eta) \}  \right]\\
      = &\ \sum_{i=1}^{n}\E\left[ w_{i, b}(x) \frac{n}{k}\E\left[ \1\{ Y_i > {Q}_{X_i}(\tau_n)(1 - \eta) \} \mid X_i\right]  \right]\label{eq:step-honesty}\numberthis\\
      \leq &\ (1 - \eta)^{-1/\xi_{-}-\zeta},\label{eq:step-add-one}\numberthis,
    \end{align*}
    where in~\eqref{eq:step-honesty} we used that honesty implies that $Y_i \ind w_{i, b}(x)$ conditional on $X_i$, and in~\eqref{eq:step-add-one} we used that the weights $w_{i, b}(x)$ add up to one.
    The expectation of the forest $T_2(x)$ is equal to the expectation of a single tree $T_{2, b}(x)$. 
    Since $\eta, \zeta > 0$ are arbitrary, it follows that $\E[T_2(x)] \to 1$.

    We now consider the variance of $T_{2, b}(x)$.
    Fix $\zeta > 0$.
    On the event $A_n(\eta)$, using Lemma~\ref{lem:reg-var-1}, for $n$ large enough we have
    \begin{align*}
      \V&\left[ T_{2, b}(x) \right]
      \leq  \E\left[ T_{2, b}(x)^2 \right]\numberthis\label{eq:var-tb}\\
      < &\ \E\left[ \left( \sum_{i = 1}^{n}w_{i,b}(x) \frac{n}{k} \1\{ Y_i > {Q}_{X_i}(\tau_n)(1 - \eta) \}  \right)^2 \right]\\
      = &\ \sum_{i = 1}^{n} \E\left[ w_{i, b}(x)^2 \left( \frac{n}{k} \right)^2 \P\left( Y_i > Q_{X_i}(\tau_n)(1 - \eta) \mid X_i \right) \right]\\
      &+ \sum_{i\neq j}\E\left[ w_{i,b}(x)w_{j,b}(x)
      \left( \frac{n}{k} \right)^2
      \P\left( Y_i > Q_{X_i}(\tau_n)(1 - \eta) \mid X_i\right)
      \P\left( Y_j > Q_{X_j}(\tau_n)(1 - \eta) \mid X_j\right)
      \right]\\
      \leq &\ \left( \frac{n}{k}(1 - \eta)^{-1/\xi_{-}-\zeta} + (1 - \eta)^{-2/\xi_{-}-2\zeta} \right).
    \end{align*}
    Using Lemma~\ref{lem:varforest}, the variance of the forest is at most $s/n$ the variance of a tree. Therefore, using~(3.6), %
    we have that
    \begin{align*}
      \V\left[ T_2(x) \right] \leq \frac{s}{n} \V\left[ T_{2, b}(x) \right] \leq  \left( \frac{s}{k} (1 - \eta)^{-1/\xi_{-}-\zeta} + \frac{s}{n} (1 - \eta)^{-2/\xi_{-}-2\zeta} \right) \to 0,
    \end{align*}
    as $n \to \infty$.
  \end{proof}

  \begin{lemma}[Term~(\ref{eq:t2}) of $f_n$]\label{lem:f-2}
    It holds that
    \begin{align*}
      \sum_{i = 1}^{n} w_i(x) \1\{Z_i > 0\} \frac{n}{k}C_n \norm{X_i - x}_2 \stackrel{\P}{\to} 0.
    \end{align*}
  \end{lemma}

  \begin{proof}
    Fix $\eta > 0$. On the event $A_n(\eta)$ defined in~\eqref{eq:event}, it holds
    \begin{align*}
      0 
      \leq &\
      \sum_{i = 1}^{n} w_{i}(x) \frac{n}{k}C_n \1\{Y_i > \hat{Q}_{X_i}(\tau_n)\} \norm{X_i - x}_2 \\
      \leq &\
      \sum_{i = 1}^{n} w_i(x) \frac{n}{k}C_n \1\{Y_i > Q_{X_i}(\tau_n)(1 - \eta)\} \norm{X_i - x}_2 \\
      \coloneqq &\ T(x).
    \end{align*}
  We will show that $\E[T(x)] \to 0$ and $\V[T(x)] \to 0$ on the event $A_n(\eta)$.
  Fix a tree $b = 1, \dots, B$, and define
  \begin{align*}
    T_{b}(x) 
    \coloneqq 
   \sum_{i = 1}^{n} w_{i, b}(x)  \frac{n}{k}C_n  \1\{Y_i > Q_{X_i}(\tau_n)(1 - \eta)\} \norm{X_i - x}_2.
  \end{align*}
  Fix $\zeta > 0$. By Lemma~\ref{lem:reg-var-1}, for $n$ large enough we have 
  \begin{align}
    \sup_{x\in\mathcal{X}}\frac{n}{k} \P\left( Y > {Q}_{X}(\tau_n)(1 - \eta) \mid X = x \right)
    <
    (1 - \eta)^{-1/\xi_{-}-\zeta}.
  \end{align}
  
    We now consider the expectation of $T_{b}$.  On the event $A_n(\eta)$, using Lemma~\ref{lem:reg-var-1}, for $n$ large enough we have
  \begin{align*}
    0 \leq \E\left[T_{b}(x)  \right]
    = &\  \sum_{i = 1}^{n} \E \left[ w_{i, b}(x)  C_n \norm{X_i - x}_2  \frac{n}{k}\E\left[  \1\{Y_i > Q_{X_i}(\tau_n)(1 - \eta)\} \mid X_i\right]\right]\\
    < &\ (1 - \eta)^{-1/\xi_{-}-\zeta}\ C_n  \E \left[ \sum_{i = 1}^{n}  w_{i, b}(x)  \norm{X_i - x}_2 \right]\\
    \leq &\ (1 - \eta)^{-1/\xi_{-}-\zeta}\ C_n \E[\diam(L_b(x))].
  \end{align*}
  The expectation of the forest $T(x)$ is equal to the expectation of a single tree $T_{b}(x)$. 
  Moreover, Corollaries~\ref{cor:diam} and~\ref{cor:flucts-2} imply that $C_n \E[\diam(L_b(x))] \to 0$ as $n\to \infty$.
  Since $\eta, \zeta > 0$ are arbitrary, it follows that $\E[T(x)] \to 0$.

  We now consider the variance of $T_{b}$.
  With similar calculations as in~\eqref{eq:var-tb},
  on the event $A_n(\eta)$, using Lemma~\ref{lem:reg-var-1}, for $n$ large enough we have
  \begin{align*}
    \V&\left[ T_{b}(x) \right]
      \leq  \E\left[ T_{b}(x)^2 \right]\\
      < &\ \E\left[ \left(  \sum_{i = 1}^{n} w_{i, b}(x)  \frac{n}{k}C_n  \1\{Y_i > Q_{X_i}(\tau_n)(1 - \eta)\} \norm{X_i - x}_2\right)^2 \right]\\
      \leq &\ C_n^2\ \E\left[ \diam(L_b(x))^2 \right] \left( \frac{n}{k}   (1 - \eta)^{-1/\xi_{-}-\zeta} +   (1 - \eta)^{-2/\xi_{-}-2\zeta} \right).
  \end{align*}
  Using Lemma~\ref{lem:varforest}, the variance of the forest is at most $s/n$ the variance of a tree. Therefore, using~(3.6), %
  we have that
  \begin{align*}
    \V\left[ T(x) \right] 
    \leq &\ \frac{s}{n} \V\left[ T_{b}(x) \right]\\ 
    \leq &\ C_n^2\ \E\left[ \diam(L_b(x))^2 \right] \left( \frac{s}{k} (1 - \eta)^{-1/\xi_{-}-\zeta} + \frac{s}{n} (1 - \eta)^{-2/\xi_{-}-2\zeta}\right) \to 0,
  \end{align*}
  as $n \to \infty$.
  Here we used~(3.6), Corollary~\ref{cor:flucts-2} and the fact that 
  \[
    \E[\diam(L_b(x))^2] = \mathcal{O}\left( s^{-0.5 C_3} \right),
  \]
  which can be easily verified from~\eqref{eq:exp-diam}.
  \end{proof}

  \begin{lemma}[Term~(\ref{eq:t4}) of $f_n$]\label{lem:f-4}
    It holds that 
    \[
      T(x) \coloneqq \sum_{i = 1}^{n} w_i(x) \1\{U_i > \tau_n\}\frac{n}{k}\log\left(\frac{Q_x(U_i)}{Q_x(\tau_n)}\right) \stackrel{\P}{\to} \xi(x).
    \]
    
  \end{lemma}

  \begin{proof}
    First, from~\citep[][Equation~(1.5)]{hsing1991} it holds, as $n \to \infty$, that
    \begin{align*}
      \E&\left[ \frac{n}{k}\1\{U_i > \tau_n\}\log\left(\frac{Q_x(U_i)}{Q_x(\tau_n)}\right) \right] \to \xi(x),\label{eq:exp-hill} \numberthis\\
      \E&\left[ \frac{n}{k}\1\{U_i > \tau_n\}\log\left(\frac{Q_x(U_i)}{Q_x(\tau_n)}\right)^2 \right] \to 2\xi(x)^2. \label{eq:var-hill-0} \numberthis
    \end{align*}

    Define
    \begin{align*}
      T_b(x) \coloneqq \sum_{i = 1}^{n} w_{i, b}(x) \1\{U_i > \tau_n\}\frac{n}{k}\log\left(\frac{Q_x(U_i)}{Q_x(\tau_n)}\right).
    \end{align*}
    Consider the expectation of $T_b(x)$. {Honesty implies that $U_i \ind w_{i, b}(x)$ conditionally on $X_i$. Therefore, from convergence in~\eqref{eq:exp-hill}, for every $\varepsilon >0$ there exists a sample size $n_0$ such that for all $n > n_0$
    \begin{align*}
      |\E[T_b(x)] - \xi(x)| &\
       = \left|\sum_{i = 1}^{n}\E\left\{  \E\left[w_{i,b}(x) \frac{n}{k}\1\{U_i > \tau_n\}\log\left(\frac{Q_x(U_i)}{Q_x(\tau_n)}\right) \mid X_i, w_{i, b}(x) \right]\right\} - \xi(x) \right|\\
      \leq &\
      \sum_{i = 1}^{n}\E\left\{w_{i,b}(x)\left|\E\left[ \frac{n}{k}\1\{U_i > \tau_n\}\log\left(\frac{Q_x(U_i)}{Q_x(\tau_n)}\right) \mid X_i \right] - \xi(x)\right|\right\}\\
      = &\
      \sum_{i = 1}^{n}\E\left\{w_{i,b}(x) \left|\E\left[ \frac{n}{k}\1\{U_i > \tau_n\}\log\left(\frac{Q_x(U_i)}{Q_x(\tau_n)}\right) \right] - \xi(x)\right|\right\}\\
      < &\ \E\left\{\sum_{i = 1}^{n}w_{i,b}(x) \varepsilon\right\}< \varepsilon.
    \end{align*}
    }
    It follows that $\E[T_b(x)] \to \xi(x)$, and therefore $\E[T(x)] \to \xi(x)$, too.

    Consider the variance of $T_b(x)$.  
    From~\eqref{eq:var-hill-0}, we have that 
    \begin{align}\label{eq:var-hill}
      \E\left[ \left( \frac{n}{k}\1\{U_i > \tau_n\}\log\left(\frac{Q_x(U_i)}{Q_x(\tau_n)}\right) \right)^2 \right] = \mathcal{O}\left( \frac{n}{k} \right),
    \end{align}
    and thus $\V[T_b(x)] = \mathcal{O}(n/k)$.
    Using Lemma~\ref{lem:varforest}, the variance of the forest $T(x)$ is at most $s/n$ the variance of a tree. Therefore, using~(3.6), we have that $\V[T(x)] \leq s/n \V[T_b(X)] \to 0$.
  \end{proof}

  \begin{lemma}[Term~(\ref{eq:t5}) of $f_n$]\label{lem:f-5}
    It holds that
    \[
      \sum_{i = 1}^{n}w_i(x) \frac{n}{k}
      \left[\1\{Z_i > 0\}\log\left(\frac{Y_i }{Q_x(\tau_n)}\right)-  
      \1\{U_i > \tau_n\} \log\left(\frac{Q_x(U_i)}{Q_x(\tau_n)} \right) \right] 
      \stackrel{\P}{\to} 0.
    \]
  \end{lemma}

  \begin{proof}
    Fix $\varepsilon, \eta > 0$ and let $A_n(\eta)$ denote the event defined in~\eqref{eq:event}. We have that
    \begin{align*}
      \P&\left(\left| \sum_{i = 1}^{n}w_i(x) \frac{n}{k}
      \left[\1\{Z_i > 0\}\log\left(\frac{Y_i }{Q_x(\tau_n)}\right)-  
      \1\{U_i > \tau_n\} \log\left(\frac{Q_x(U_i)}{Q_x(\tau_n)} \right) \right] \right| > \varepsilon \right)\\
      \leq &\ 
      \P\left(\left| \sum_{i = 1}^{n}w_i(x) \frac{n}{k}
      \left[\1\{Y_i > \hat{Q}_{X_i}(\tau_n)\} - \1\{Y_i > {Q}_{X_i}(\tau_n)\}\right] \log\left(\frac{Y_i }{Q_x(\tau_n)}\right) \right| > \frac{\varepsilon}{2}, A_n(\eta) \right)\tag{$VI$}\label{eq:t6}\\
      &+ \P\left(\left| \sum_{i = 1}^{n}w_i(x) \frac{n}{k}
      \1\{Y_i > {Q}_{X_i}(\tau_n)\} \log\left(\frac{Y_i }{Q_x(U_i)}\right) \right| > \frac{\varepsilon}{2}, A_n(\eta) \right)\tag{$VII$}\label{eq:t7}\\
      &+ \P\left( A_n(\eta)^c \right).
    \end{align*}

  Using Lemma~\ref{lem:f-6}, it holds that~\eqref{eq:t6}$\to 0$ as $n\to \infty$.

  Using Lemma~\ref{lem:f-7}, it holds that~\eqref{eq:t7}$\to 0$ as $n\to \infty$.
  \end{proof}

  \begin{lemma}[Term~(\ref{eq:t6}) of $f_n$]\label{lem:f-6}
    It holds that
    \begin{align*}
      \sum_{i = 1}^{n}w_i(x) \frac{n}{k}
      \left[\1\{Y_i > \hat{Q}_{X_i}(\tau_n)\} - \1\{Y_i > {Q}_{X_i}(\tau_n)\}\right] \log\left(\frac{Y_i }{Q_x(\tau_n)}\right) \stackrel{\P}{\to} 0.
    \end{align*}
  \end{lemma}

  \begin{proof}
    Fix $\eta > 0$.
    On the event $A_n(\eta)$ defined in~\eqref{eq:event}, it holds for all $i = 1, \dots, n$,
    \begin{align*}
      \big|\1\{&Y_i > \hat{Q}_{X_i}(\tau_n)\} - \1\{Y_i > {Q}_{X_i}(\tau_n)\}\big|
      \\
      \leq &\ 
      \left|\1\{Y_i > {Q}_{X_i}(\tau_n)(1-\eta)\} - \1\{Y_i > {Q}_{X_i}(\tau_n)(1 + \eta)\}\right|
      \\
      = &\
      \1\{{Q}_{X_i}(\tau_n)(1-\eta) < Y_i < {Q}_{X_i}(\tau_n)(1 + \eta)\}.
    \end{align*}
  Therefore, it follows that
  \begin{align*}
    \left|\log\left( \frac{Y_i}{Q_x(\tau_n)} \right)\right| \leq \left|\log\left( \frac{Q_{X_i}(\tau_n)}{Q_x(\tau_n)} \right)\right| + \left|\log(1-\eta)\right|,\ \text{if}\ Y_i < Q_x(\tau_n),
    \\
    \left|\log\left( \frac{Y_i}{Q_x(\tau_n)} \right)\right| \leq \left|\log\left( \frac{Q_{X_i}(\tau_n)}{Q_x(\tau_n)} \right)\right| + \left|\log(1+\eta)\right|,\ \text{if}\ Y_i \geq Q_x(\tau_n).
  \end{align*}
  Since $|\log(1 + \eta)| < |\log(1 - \eta)|$, using Lemma~\ref{lem:flucts-1}, we have for all $i = 1, \dots, n$,
  \begin{align*}
    \left|\log\left( \frac{Y_i}{Q_x(\tau_n)} \right)\right| 
    \leq \left|\log\left( \frac{Q_{X_i}(\tau_n)}{Q_x(\tau_n)} \right)\right| + \left|\log(1-\eta)\right|
    \leq C_n \norm{X_i - x}_2 + \left|\log(1 - \eta)\right|.
  \end{align*}
  Fix $\zeta > 0$. By Lemma~\ref{lem:reg-var-1}, for $n$ large enough and all $x \in \mathcal{X}$ we have 
  \begin{align*}
    \frac{n}{k} &\P\left( Q_X(\tau_n)(1 - \eta) < Y < Q_X(\tau_n)(1 + \eta) \mid X = x \right)
    \\
    = &\
    \frac{n}{k} \left[ \P\left( Y > Q_X(\tau_n)(1 - \eta)\mid X = x \right) - \P\left( Y > Q_X(\tau_n)(1 + \eta) \mid X = x \right) \right]
    \\
    \leq &\
    (1 - \eta)^{-1/\xi_{-}-\zeta}
    -
    (1 + \eta)^{-1/\xi_{-}-\zeta}.
  \end{align*}
  Therefore, on the event $A_n(\eta)$, we have that
  \begin{align*}
    0 \leq &\
    \sum_{i = 1}^{n}w_i(x) \frac{n}{k}
    \left|\1\{Y_i > \hat{Q}_{X_i}(\tau_n)\} - \1\{Y_i > {Q}_{X_i}(\tau_n)\}\right| \left| \log\left(\frac{Y_i }{Q_x(\tau_n)}\right)\right|
    \\
    \leq &\
    \sum_{i = 1}^{n}w_i(x) \frac{n}{k}
    \left|\1\{{Q}_{X_i}(\tau_n)(1 - \eta) < Y_i < {Q}_{X_i}(\tau_n)(1 + \eta)\}\right|\left(  C_n \norm{X_i - x}_2 + \left|\log(1 - \eta)\right| \right)
    \\
    \eqqcolon &\ T(x).
  \end{align*}
  With similar calculations as in Lemma~\ref{lem:f-2}, it follows that
  $\E[T(x)] \to 0$ and $\V[T(x)] \to 0$ as $n \to \infty$.
  \end{proof}

  \begin{lemma}[Term~(\ref{eq:t7}) of $f_n$]\label{lem:f-7}
    It holds that
    \begin{align*}
      T(x) \coloneqq \sum_{i = 1}^{n}w_i(x) \frac{n}{k}
      \1\{Y_i > {Q}_{X_i}(\tau_n)\} \log\left(\frac{Y_i }{Q_x(U_i)}\right) \stackrel{\P}{\to} 0.
    \end{align*}
  \end{lemma}

  \begin{proof}
    We will show that $\E[T(x)] \to 0$ and $\V[T(x)] \to 0$ as $n\to\infty$.
    Recall the stochastic representation $Y_i \stackrel{d}{=}Q_{X_i}(U_i)$.
    Fix a tree $b = 1, \dots, B$ and define
    \begin{align*}
      T_b(x) \coloneqq \sum_{i = 1}^{n}w_{i, b}(x) \frac{n}{k}
      \1\{U_i > \tau_n\} \log\left(\frac{Q_{X_i}(U_i)}{Q_x(U_i)}\right).
    \end{align*}
    From Lemma~\ref{lem:flucts-1}, by plugging in
    $U_i > \tau_n$ in place of $\tau_n$,
    we have that
    \begin{align*}
      \left|\log\left( \frac{Q_{X_i}(U_i)}{Q_{x}(U_i)}\right)\right|
      \leq &\ 
      \left\{ \log\left( \frac{1}{1-U_i} \right) (L_\xi + L_\alpha) + L_{c} \right\} \norm{X_i-x}_2.
    \end{align*}
    Moreover, note that conditional on the event $U_i > \tau_n$ we have the stochastic representation
    \begin{align*}
      \log\left( \frac{1}{1-U_i} \right) \stackrel{d}{=} E_i + \log\left( \frac{n}{k} \right),
    \end{align*}
    where $E_i \sim \mathrm{Exp}(1)$. 
    Therefore, we have that 
    \begin{align*}
      \E&\left[\left|\log\left(\frac{Q_{X_i}(U_i)}{Q_x(U_i)}\right)\right| \1\left\{ U_i > \tau_n \right\} \mid X_i\right]\\
      \leq &\ 
      \E\left[ \left\{ \log\left( \frac{1}{1-U_i} \right) (L_\xi + L_\alpha) + L_{c} \right\} \1\left\{ U_i > \tau_n \right\} \mid X_i\right] \norm{X_i-x}_2
      \\
      = &\ \E\left[ \left\{ \log\left( \frac{1}{1-U_i} \right) (L_\xi + L_\alpha) + L_{c} \right\} \mid U_i > \tau_n\right] \P\left( U_i > \tau_n \right) \norm{X_i - x}_2
      \\
      = &\ \E\left[\log\left( \frac{1}{1-U_i} \right) \mid U_i > \tau_n\right] \frac{k}{n} \norm{X_i - x}_2 (L_\xi + L_\alpha)
      + L_{c} \frac{k}{n} \norm{X_i - x}_2
      \\
      = &\ \left( 1 + \log\left( \frac{n}{k} \right) \right)
      \frac{k}{n} \norm{X_i - x}_2 (L_\xi + L_\alpha)
      + L_{c} \frac{k}{n} \norm{X_i - x}_2.
    \end{align*}
    Moreover, we have that
    \begin{align*}
      \E&\left[\left\{\log\left(\frac{Q_{X_i}(U_i)}{Q_x(U_i)}\right)\right\}^2 \1\left\{ U_i > \tau_n \right\} \mid X_i\right]
      \\
      \leq &\ 
      \E\left[ \left\{ \log\left( \frac{1}{1-U_i} \right) (L_\xi + L_\alpha) + L_{c} \right\}^2 \mid U_i > \tau_n, X_i\right] \frac{k}{n}\norm{X_i-x}_2^2
      \\
      \leq &\ \E\left[\log\left( \frac{1}{1-U_i} \right)^2 \mid U_i > \tau_n\right] \frac{k}{n} \norm{X_i - x}_2^2 (L_\xi + L_\alpha)^2
      + L_{c}^2 \frac{k}{n} \norm{X_i - x}_2^2
      \\
      &+ 2\  \E\left[\log\left( \frac{1}{1-U_i} \right)\mid U_i > \tau_n\right]\frac{k}{n} \norm{X_i - x}_2^2 (L_\xi + L_\alpha)L_{c}
      \\
      = &\
      \left( 2 + \log\left( \frac{n}{k} \right)^2 + 2\log\left( \frac{n}{k} \right) \right)\frac{k}{n} \norm{X_i - x}_2^2 (L_\xi + L_\alpha)^2
      + L_{c}^2 \frac{k}{n} \norm{X_i - x}_2^2
      \\
      & + \left(2 + 2\log\left( \frac{n}{k} \right) \right)\frac{k}{n} \norm{X_i - x}_2^2 (L_\xi + L_\alpha)L_{c}.
    \end{align*}
    Consider the expectation of $T_b(x)$. Using similar calculations as in Lemma~\ref{lem:f-2}, it is easy to show that 
    \begin{align*}
      |\E[T_b(x)]| = \mathcal{O}\left(\E\left[ \diam(L_b(x)) \right]\ \log\left( \frac{n}{k} \right) \right),
    \end{align*}
    which implies that $\E[T(x)] = \E[T_b(x)] \to 0$ as $n \to \infty$.
    Consider the variance of $T_b(x)$. Using similar calculations as in Lemma~\ref{lem:f-2}, it is easy to show that
    \begin{align*}
      \V[T_b(x)] \leq&\ \E[T_b(x)^2]\\
      = &\
      \mathcal{O}\left( \frac{n}{k}\ \E\left[ \diam(L_b(x))^2 \right]\ \log\left( \frac{n}{k} \right)^2 \right).
    \end{align*}
    Using Lemma~\ref{lem:varforest}, the variance of the forest is at most $s/n$ the variance of a tree. Therefore, using~(3.6), we have that $\V[T(x)] \leq s/n \V[T_b(X)] \to 0$.

  \end{proof}

  \begin{lemma}[Leading term of $g_n$]\label{lem:gn-1}
    It holds that
    \begin{align*}
     \tilde{T}_1(x) \coloneqq \sum_{i = 1}^{n} w_i(x) \1\left\{ \norm{X_i - x}_2 \leq \delta_n  \right\}\1\{Z_i > 0\}\frac{n}{k}\frac{1}{1 + Z_i/\hat{Q}_x(\tau_n)} \stackrel{\P}{\to} \frac{1}{1 + \xi(x)}.
    \end{align*}
  \end{lemma}

  \begin{proof}
    Fix $\varepsilon, \eta > 0$. Define the random variable $\1_{i, \delta_n}(x) \coloneqq \1\left\{ \norm{X_i - x}_2 \leq \delta_n  \right\}$ and let $A_n(\eta)$ denote the event defined in~\eqref{eq:event}. 
    We can rewrite
    \begin{align*}
      \P&\left( \left|\tilde{T}_1(x) -  \frac{1}{1 + \xi(x)}\right| > \varepsilon\right)\\
      \leq &\
      \P\left(\left|\frac{n}{k}\sum_{i = 1}^{n} w_i(x) \1_{i, \delta_n}(x)\1\{U_i > \tau_n\}\frac{Q_x(\tau_n)}{Q_x(U_i)} - \frac{1}{1 + \xi(x)}\right| > \frac{\varepsilon}{5} \right)
      \tag{$I$}\label{eq:g-t1}\\
      &+ \P\left(\sum_{i = 1}^{n} w_i(x) \1_{i, \delta_n}(x)\1\{U_i > \tau_n\}\frac{n}{k}\left|\frac{Q_{X_i}(\tau_n)}{Y_i}-\frac{Q_x(\tau_n)}{Q_x(U_i)}\right| > \frac{\varepsilon}{5} \right)
      \tag{$II$}\label{eq:g-t2}\\
      &+ \P\left(\sum_{i = 1}^{n} w_i(x) \1_{i, \delta_n}(x)\frac{n}{k}
      \left|\1\{U_i > \tau_n\} - \1\{Y_i > \hat{Q}_{X_i}(\tau_n)\}\right|
      \left|\frac{Q_{X_i}(\tau_n)}{Y_i}\right| > \frac{\varepsilon}{5} , A_n(\eta) \right)
      \tag{$III$}\label{eq:g-t3}\\
      &+\P\left(\sum_{i = 1}^{n} w_i(x) \1_{i, \delta_n}(x)\1\{Y_i > \hat{Q}_{X_i}(\tau_n)\}\frac{n}{k}\left|\frac{\hat{Q}_{X_i}(\tau_n)}{Y_i}-\frac{Q_{X_i}(\tau_n)}{Y_i}\right| > \frac{\varepsilon}{5} , A_n(\eta) \right)
      \tag{$IV$}\label{eq:g-t4}\\
      &+\P\left(\sum_{i = 1}^{n} w_i(x) \1_{i, \delta_n}(x)\1\{Y_i > \hat{Q}_{X_i}(\tau_n)\}\frac{n}{k}\left|
        \frac{1}{1+Z_i/\hat{Q}_x(\tau_n)}-\frac{\hat{Q}_{X_i}(\tau_n)}{Y_i}\right| > \frac{\varepsilon}{5} , A_n(\eta) \right)
      \tag{$V$}\label{eq:g-t5}\\
      &+\P\left( A_n(\eta)^c \right).
    \end{align*}

    Consider term~\eqref{eq:g-t1}. Using Lemma~\ref{lem:g-1}, it holds that~\eqref{eq:g-t1}$\to 0$ as $n\to \infty$.

    Consider term~\eqref{eq:g-t2}. Using Lemma~\ref{lem:g-2}, it holds that~\eqref{eq:g-t2}$\to 0$ as $n\to \infty$.

    Consider term~\eqref{eq:g-t3}. Using Lemma~\ref{lem:g-3}, it holds that~\eqref{eq:g-t3}$\to 0$ as $n\to \infty$.

    Consider term~\eqref{eq:g-t4}. Using Lemma~\ref{lem:g-4}, it holds that~\eqref{eq:g-t4}$\to 0$ as $n\to \infty$.

    Consider term~\eqref{eq:g-t5}. Using Lemma~\ref{lem:g-5}, it holds that~\eqref{eq:g-t5}$\to 0$ as $n\to \infty$.

    Putting everything together, we have that $\tilde{T}_1(x) \stackrel{\P}{\to} (1+\xi(x))^{-1}$.

  \end{proof}

  \begin{lemma}[Remainder term of $g_n$]\label{lem:gn-2}
    It holds that
    \begin{align*}
      \sum_{i = 1}^{n} w_i(x) \1\{\norm{X_i - x}_2 > \delta_n \}\1\{Z_i > 0\}\frac{n}{k} \stackrel{\P}{\to} 0.
    \end{align*}
  \end{lemma}

  \begin{proof}
    This proof follows closely the proof of Lemma~\ref{lem:f-2}.
    Fix $\eta > 0$. On the event $A_n(\eta)$ defined in~\eqref{eq:event}, it holds
    \begin{align*}
      0 
      \leq &\
      \sum_{i = 1}^{n} w_i(x)\frac{n}{k}  \1\{\norm{X_i - x}_2 > \delta_n \}\1\{Y_i > \hat{Q}_{X_i}(\tau_n)\} \\
      \leq &\
      \sum_{i = 1}^{n} w_i(x)\frac{n}{k}  \1\{\norm{X_i - x}_2 > \delta_n \}\1\{Y_i > {Q}_{X_i}(\tau_n)(1-\eta)\} 
      \coloneqq T(x).
    \end{align*}
  We will show that $\E[T(x)] \to 0$ and $\V[T(x)] \to 0$ on the event $A_n(\eta)$.
  Fix a tree $b = 1, \dots, B$, and define
  \begin{align*}
    T_{b}(x) 
    \coloneqq 
    \sum_{i = 1}^{n} w_{i, b}(x)\frac{n}{k}  \1\{\norm{X_i - x}_2 > \delta_n \}\1\{Y_i > {Q}_{X_i}(\tau_n)(1-\eta)\} .
  \end{align*}
  We now consider the expectation of $T_b(x)$. 
  Fix $\zeta > 0$.
  With similar arguments as in Lemma~\ref{lem:f-2}, we have 
  \begin{align*}
    0 \leq \E\left[T_{b}(x)  \right]
    = &\  \sum_{i = 1}^{n} \E \left[ w_{i, b}(x) \1\left\{\norm{X_i - x}_2 > \delta_n\right\}  \frac{n}{k}\E\left[  \1\{Y_i > Q_{X_i}(\tau_n)(1 - \eta)\} \mid X_i\right]\right]\\
    < &\ (1 - \eta)^{-1/\xi_{-}-\zeta}\ \E \left[ \sum_{i = 1}^{n}  w_{i, b}(x)  \1\left\{\norm{X_i - x}_2 > \delta_n\right\}  \right]\\
    \leq &\ (1 - \eta)^{-1/\xi_{-}-\zeta}\ \P\left(\diam(L_b(x)) > \delta_n \right).
  \end{align*}
    Notice that for every observation $i$ satisfying that $w_{i, b}(x) > 0$ we have that $X_i \in L_b(x)$ and so $\norm{X_i - x}_2 \leq \diam(L_b(x))$.
    Therefore, the random variable $\1\{\diam(L_b(x)) > \delta_n\} \geq \1\{\norm{X_i - x}_2 > \delta_n\}$.
  The expectation of the forest $T(x)$ is equal to the expectation of a single tree $T_{b}(x)$. Furthermore, by~\eqref{eq:diam}, it holds that $\P\left(\diam(L_b(x)) > \delta_n \right) \to 0$ as $n \to \infty$.
  Since $\eta, \zeta > 0$ are arbitrary, it follows that $\E[T(x)] \to 0$.

  We now consider the variance of $T_{b}$.
  Fix $\zeta > 0$.
  With similar arguments as in Lemma~\ref{lem:f-2}, we have
  \begin{align*}
    \V&\left[ T_{b}(x) \right]
      \leq  \E\left[ T_{b}(x)^2 \right]\\
      < &\ \E\left[ \left(  \sum_{i = 1}^{n} w_{i, b}(x)  \frac{n}{k}  \1\{Y_i > Q_{X_i}(\tau_n)(1 - \eta)\} \1\left\{ \norm{X_i - x}_2 > \delta_n \right\}\right)^2 \right]\\
      \leq &\ \P\left(\diam(L_b(x))>\delta_n \right) \left( \frac{n}{k} (1 - \eta)^{-1/\xi_{-}-\zeta} + (1 - \eta)^{-2/\xi_{-}-2\zeta} \right).
  \end{align*}
  Using Lemma~\ref{lem:varforest}, the variance of the forest is at most $s/n$ the variance of a tree. Therefore, using~(3.6), we have that
  \begin{align*}
    \V\left[ T(x) \right] 
    \leq &\ \frac{s}{n} \V\left[ T_{b}(x) \right]\\ 
    \leq &\ \P\left(\diam(L_b(x))>\delta_n \right)  \left( \frac{s}{k} (1 - \eta)^{-1/\xi_{-}-\zeta} + \frac{s}{n}(1 - \eta)^{-2/\xi_{-}-2\zeta} \right) \to 0,
  \end{align*}
  as $n \to \infty$.
  \end{proof}

  \begin{lemma}[Term~(\ref{eq:g-t1}) of $g_n$]\label{lem:g-1}
    It holds that 
    \begin{align*}
      \frac{n}{k}\sum_{i = 1}^{n} w_i(x)\1\left\{ \norm{X_i - x}_2 \leq \delta_n  \right\}\1\{U_i > \tau_n\}\frac{Q_x(\tau_n)}{Q_x(U_i)} \stackrel{\P}{\to} \frac{1}{1 + \xi(x)}.
    \end{align*}
  \end{lemma}

  \begin{proof}
    Fix $\varepsilon > 0$, and consider
    \begin{align*}
      \P&\left( \left|\frac{n}{k}\sum_{i = 1}^{n} w_i(x)\1\left\{ \norm{X_i - x}_2 \leq \delta_n  \right\}\1\{U_i > \tau_n\}\frac{Q_x(\tau_n)}{Q_x(U_i)} - \frac{1}{1 + \xi(x)} \right| > \varepsilon \right)
      \\
      \leq &\
      \P\left( \left|
        \frac{n}{k}\sum_{i = 1}^{n} w_i(x)\1\{U_i > \tau_n\}\frac{Q_x(\tau_n)}{Q_x(U_i)} 
        - \frac{1}{1 + \xi(x)} \right| > \frac{\varepsilon}{2} \right)\label{eq:gn-conv-i}\numberthis
        \\
      &+
      \P\left(\left|
        \frac{n}{k}\sum_{i = 1}^{n} w_i(x)\left( \1\left\{ \norm{X_i - x}_2 \leq \delta_n  \right\} - 1 \right)\1\{U_i > \tau_n\}\frac{Q_x(\tau_n)}{Q_x(U_i)} \right| > \frac{\varepsilon}{2} \right).\label{eq:gn-conv-ii}\numberthis 
    \end{align*}
    Consider~\eqref{eq:gn-conv-ii}. We can upper bound it by
    \begin{align*}
      \P& \left(\sum_{i = 1}^{n} w_i(x)\left| \1\left\{ \norm{X_i - x}_2 \leq \delta_n  \right\} - 1 \right|\1\{U_i > \tau_n\}\frac{n}{k} > \frac{\varepsilon}{2}  \right)
      \\
      = &\ P\left(\sum_{i = 1}^{n} w_i(x)\1\left\{ \norm{X_i - x}_2 > \delta_n  \right\} \1\{U_i > \tau_n\}\frac{n}{k} > \frac{\varepsilon}{2}  \right) \to 0,
    \end{align*}
    by Lemma~\ref{lem:gn-2}. 
   
    Consider~\eqref{eq:gn-conv-i}.
    Define $T(x) \coloneqq\sum_{i = 1}^{n} w_i(x) \frac{n}{k}\1\{U_i > \tau_n\}\frac{Q_x(\tau_n)}{Q_x(U_i)}$. We will show that $\E[T(x)] \to (1 + \xi(x))^{-1}$ and $\V[T(x)] \to 0$ as $n \to \infty$.

    First, we show that 
    \begin{align}\label{eq:exp-gn}
      \E\left[\frac{n}{k} \1\{U_i > \tau_n\}\frac{Q_x(\tau_n)}{Q_x(U_i)}\right] \to \frac{1}{1 + \xi(x)},
    \end{align}
    as $n \to \infty$.
    Let $t_n = 1/(1-\tau_n) \to \infty$ as $n\to\infty$, and $y_n(u) = (1-\tau_n)/(1-u)$, which is greater or equal to 1 for $u \geq \tau_n$.  Furthermore, define $V_x(t) \coloneqq Q_x(1 - 1/t)$. Then, for any fixed $\varepsilon > 0$ there exists a sample size $n_0$ such that for all $n > n_0$ we have that 
    \begin{align*}
      \E&\left[  \1\{U_i > \tau_n\}\frac{Q_x(\tau_n)}{Q_x(U_i)}   \right]
      = 
      \int_{\tau_n}^{1} \frac{Q_x(\tau_n)}{Q_x(u)}  \mathrm{d}u
      \\
      = &\
      \int_{\tau_n}^{1} \frac{V_x(t_n)}{V_x(t_n y_n(u))}  \mathrm{d}u
      \\
      \leq &\ \frac{1}{1-\varepsilon} \int_{\tau_n}^{1} y_n(u)^{-\xi(x) + \varepsilon}  \mathrm{d}u \label{eq:V-reg-var}\numberthis
      \\
      = &\ \frac{1}{1-\varepsilon} \int_{\tau_n}^{1}\left(\frac{1-u}{1-\tau_n} \right)^{\xi(x) - \varepsilon} \mathrm{d}u 
      \\
      = &\
      \frac{1}{1-\varepsilon} \left[  \frac{\tau_n-1}{\xi(x) -\varepsilon +1} \left(\frac{1-u}{1-\tau_n} \right)^{\xi(x) - \varepsilon+1} \right]_{\tau_n}^1
      \\
      = &\ \frac{1}{1-\varepsilon} \left[  \frac{1-\tau_n}{\xi(x) -\varepsilon +1}  \right]
      \\
      = &\      
      \frac{k}{n} \frac{1}{(1+\xi(x) -\varepsilon)(1-\varepsilon)}.
    \end{align*}
    In~\eqref{eq:V-reg-var}, we use that $V_x$ is regularly varying at infinity with index $\xi(x)$, and the corresponding bound $V_x(ty)/V_x(t) \geq (1-\varepsilon)y^{\xi(x)-\varepsilon}$ for $t\geq t_0$ and $y\geq 1.$
    The lower bound can be established similarly.
    
    Furthermore, we have that
    \begin{align*}
      \E&\left[  \1\{U_i > \tau_n\}\left( \frac{Q_x(\tau_n)}{Q_x(U_i)} \right)^2   \right]
      = 
      \int_{\tau_n}^{1} \left( \frac{Q_x(\tau_n)}{Q_x(U_i)} \right)^2   \mathrm{d}u
      \\
      \leq &\ \left( \frac{1}{1-\varepsilon} \right)^2 \int_{\tau_n}^{1} y_n(u)^{-2\xi(x) + 2\varepsilon}  \mathrm{d}u
      \\
      = &\ \frac{k}{n} \frac{1}{(1+2\xi(x) -2\varepsilon)(1-\varepsilon)^2},
    \end{align*}
    so that we can upper bound 
    \begin{align}\label{eq:var-gn}
      \E\left[\left(\frac{n}{k} \1\{U_i > \tau_n\}\frac{Q_x(\tau_n)}{Q_x(U_i)}  \right)^2\right] = \mathcal{O}(n/k).
    \end{align}

    Define
    \begin{align*}
      T_b(x) \coloneqq \sum_{i = 1}^{n} w_{i, b}(x) \1\{U_i > \tau_n\}\frac{n}{k} \frac{Q_x(\tau_n)}{Q_x(U_i)}.
    \end{align*}
    Consider the expectation of $T_b(x)$. Honesty implies that $U_i \ind w_{i, b}(x)$ conditionally on $X_i$. Therefore, from convergence in~\eqref{eq:exp-gn}, for every $\varepsilon >0$ there exists a sample size $n_0$ such that for all $n > n_0$
    \begin{align*}
      \bigg|\E&[T_b(x)]- \frac{1}{1+\xi(x)}\bigg|\\ 
      \leq &\
      \sum_{i = 1}^{n}\E\left\{w_{i,b}(x)\left|\E\left[ \frac{n}{k}\1\{U_i > \tau_n\}\frac{Q_x(\tau_n)}{Q_x(U_i)} \mid X_i, w_{i, b}(x) \right] - \frac{1}{1 +\xi(x)}\right|\right\}\\
      = &\
      \sum_{i = 1}^{n}\E\left\{w_{i,b}(x) \left|\E\left[ \frac{n}{k}\1\{U_i > \tau_n\}\frac{Q_x(\tau_n)}{Q_x(U_i)}\right] - \frac{1}{1 +\xi(x)}\right|\right\}\\
      < &\ \E\left\{\sum_{i = 1}^{n}w_{i,b}(x) \varepsilon\right\}< \varepsilon.
    \end{align*}
    It follows that $\E[T_b(x)] \to \xi(x)$, and therefore $\E[T(x)] \to \xi(x)$, too.

    Consider the variance of $T_b(x)$.  Using~\eqref{eq:var-gn},
    we have that $\V[T_b(x)] = \mathcal{O}(n/k)$.
    Using Lemma~\ref{lem:varforest}, the variance of the forest $T(x)$ is at most $s/n$ the variance of a tree. Therefore, using~(3.6), we have that $\V[T(x)] \leq s/n \V[T_b(X)] \to 0$.

  \end{proof}
  \begin{lemma}[Term~(\ref{eq:g-t2}) of $g_n$]\label{lem:g-2}
    It holds that
    \begin{align*}
      \sum_{i = 1}^{n} w_i(x) \1_{i, \delta_n}(x)\1\{U_i > \tau_n\}\frac{n}{k}\left|\frac{Q_{X_i}(\tau_n)}{Y_i}-\frac{Q_x(\tau_n)}{Q_x(U_i)}\right| \stackrel{\P}{\to}0.
    \end{align*}
  \end{lemma}

  \begin{proof}
    Let $i$ be an observation satisfying  $\norm{X_i - x}_2 < \delta_n$.
    For $n$ large enough,
    using Lemma~\ref{lem:flucts-1}, we can make $|\log(Q_x(\tau_n)) - \log(Q_{X_i}(\tau_n))| \leq C_n \delta_n$ arbitrarily small.
    Moreover, using the mean value theorem, it holds that $|x - 1| \leq 2 |\log(x)|$ when $x$ is sufficiently small. Therefore, we can use the following upper bound,
    \begin{align*}
      \left|\frac{Q_{X_i}(\tau_n)}{Y_i}-\frac{Q_x(\tau_n)}{Q_x(U_i)}\right|
      \leq &\
      \left|\frac{Q_{X_i}(\tau_n) - Q_x(\tau_n)}{Y_i} \right| + \left| \frac{Q_x(\tau_n)}{Y_i} - \frac{Q_x(\tau_n)}{Q_x(U_i)}\right|\\
      \leq &\
      \left|1 - \frac{Q_x(\tau_n)}{Q_{X_i}(\tau_n)} \right| + \left|\frac{Q_x(\tau_n)}{Q_x(U_i)}\right|\left|\frac{Q_x(U_i)-Q_{X_i}(U_i)}{Q_{X_i}(U_i)} \right|\\
      \leq &\ 2 \left|\log\left( \frac{Q_x(\tau_n)}{Q_{X_i}(\tau_n)} \right)\right| + 2 \left|\log\left( \frac{Q_x(U_i)}{Q_{X_i}(U_i)} \right)\right|\\
      \leq &\ 2C_n \norm{X_i - x}_2 + 2\left|\log\left( \frac{Q_x(U_i)}{Q_{X_i}(U_i)} \right)\right|.
    \end{align*}
    We can then split the term as follows,
    \begin{align*}
      0 \leq &\ \sum_{i = 1}^{n} w_i(x) \1_{i, \delta_n}(x)\1\{U_i > \tau_n\}\frac{n}{k}\left|\frac{Q_{X_i}(\tau_n)}{Y_i}-\frac{Q_x(\tau_n)}{Q_x(U_i)}\right|\\
      \leq &\
      2 \sum_{i = 1}^{n} w_i(x) \1_{i, \delta_n}(x)\1\{U_i > \tau_n\}\frac{n}{k}C_n \norm{X_i - x}_2\\
      &+ 2 \sum_{i = 1}^{n} w_i(x) \1_{i, \delta_n}(x)\1\{U_i > \tau_n\}\frac{n}{k}\left|\log\left( \frac{Q_x(U_i)}{Q_{X_i}(U_i)} \right)\right|\\
      \leq &\ S_1(x) + S_2(x).
    \end{align*}
    The term $S_1(x) \stackrel{\P}{\to} 0$,
      by very similar arguments to the proof of
      by Lemma~\ref{lem:f-2}. The term $S_2(x) \stackrel{\P}{\to} 0$, by very similar arguments to the proof of Lemma~\ref{lem:f-7}.
  \end{proof}

  \begin{lemma}[Term~(\ref{eq:g-t3}) of $g_n$]\label{lem:g-3}
    It holds that
    \begin{align*}
      \sum_{i = 1}^{n} w_i(x) \1_{i, \delta_n}(x)\frac{n}{k}
      \left|\1\{U_i > \tau_n\} - \1\{Y_i > \hat{Q}_{X_i}(\tau_n)\}\right|
      \left|\frac{Q_{X_i}(\tau_n)}{Y_i}\right| \stackrel{\P}{\to} 0.
    \end{align*}
    
  \end{lemma}

  \begin{proof}
    Fix $\eta > 0$.
    On the event $A_n(\eta)$ defined in~\eqref{eq:event}, it holds for all $i = 1, \dots, n$,
    \begin{align*}
      \big|\1\{&Y_i > \hat{Q}_{X_i}(\tau_n)\} - \1\{Y_i > {Q}_{X_i}(\tau_n)\}\big|
      \\
      \leq &\ 
      \left|\1\{Y_i > {Q}_{X_i}(\tau_n)(1-\eta)\} - \1\{Y_i > {Q}_{X_i}(\tau_n)(1 + \eta)\}\right|
      \\
      = &\ \1\{{Q}_{X_i}(\tau_n)(1-\eta) < Y_i < {Q}_{X_i}(\tau_n)(1 + \eta)\}.
    \end{align*}
    Fix $\zeta > 0$. By Lemma~\ref{lem:reg-var-1}, for $n$ large enough and all $x \in \mathcal{X}$ we have 
    \begin{align*}
      \frac{n}{k} &\P\left( Q_X(\tau_n)(1 - \eta) < Y < Q_X(\tau_n)(1 + \eta) \mid X = x \right)
      \\
      = &\
      \frac{n}{k} \left[ \P\left( Y > Q_X(\tau_n)(1 - \eta)\mid X = x \right) - \P\left( Y > Q_X(\tau_n)(1 + \eta) \mid X = x \right) \right]
      \\
      \leq &\
      (1 - \eta)^{-1/\xi_{-} - \zeta}
      -
      (1 + \eta)^{-1/\xi_{-} - \zeta}.
    \end{align*}
    Therefore, on the event $A_n(\eta)$, we have that
    \begin{align*}
      0 \leq &\
      \sum_{i = 1}^{n} w_i(x) \1_{i, \delta_n}(x)\frac{n}{k}
      \left|\1\{U_i > \tau_n\} - \1\{Y_i > \hat{Q}_{X_i}(\tau_n)\}\right|
      \left|\frac{Q_{X_i}(\tau_n)}{\hat Q_{X_i}(\tau_n)} \frac{\hat Q_{X_i}(\tau_n)}{Y_i} \right|
      \\
      \leq &\
      \sum_{i = 1}^{n} w_i(x) \1_{i, \delta_n}(x)\frac{n}{k}\1\{{Q}_{X_i}(\tau_n)(1-\eta) < Y_i < {Q}_{X_i}(\tau_n)(1 + \eta)\}\left|\frac{1}{(1 - \eta)}\right|
      \\
      \eqqcolon &\ T(x).
    \end{align*}
    With similar calculations as in Lemma~\ref{lem:f-2}, it follows that
    $\E[T(x)] \to 0$ and $\V[T(x)] \to 0$ as $n \to \infty$.
  \end{proof}

  \begin{lemma}[Term~(\ref{eq:g-t4}) of $g_n$]\label{lem:g-4}
    It holds that
    \begin{align*}
      \sum_{i = 1}^{n} w_i(x) \1_{i, \delta_n}(x)\1\{Y_i > \hat{Q}_{X_i}(\tau_n)\}\frac{n}{k}\left|\frac{\hat{Q}_{X_i}(\tau_n)}{Y_i}-\frac{Q_{X_i}(\tau_n)}{Y_i}\right| \stackrel{\P}{\to} 0.
    \end{align*}
    
  \end{lemma}

  \begin{proof}
    Fix $\eta > 0$ and let $i$ be an observation satisfying  $\norm{X_i - x}_2 < \delta_n$ and $Z_i > 0$. Therefore, on the event $A_n(\eta)$ defined in~\eqref{eq:event}, we can upper bound 
    \begin{align*}
      \bigg|&\frac{\hat{Q}_{X_i}(\tau_n) - Q_{X_i}(\tau_n)}{Y_i}\bigg|
      = 
      \left|\frac{\hat{Q}_{X_i}(\tau_n) - Q_{X_i}(\tau_n)}{Q_{X_i}(\tau_n)} \right| \left|\frac{Q_{X_i}(\tau_n)}{Y_i}\right|
      \leq
      \eta \left|\frac{Q_{X_i}(\tau_n)}{Q_{X_i}(\tau_n)(1 - \eta)}\right|
      = \frac{\eta}{1 - \eta} .
    \end{align*}
    We can then upper bound
    \begin{align*}
      0
      \leq &\
      \sum_{i = 1}^{n} w_i(x) \1_{i, \delta_n}(x)\1\{Y_i > \hat{Q}_{X_i}(\tau_n)\}\frac{n}{k}\left|\frac{\hat{Q}_{X_i}(\tau_n)}{Y_i}-\frac{Q_{X_i}(\tau_n)}{Y_i}\right|
      \\
      \leq &\
      \sum_{i = 1}^{n} w_i(x) \1_{i, \delta_n}(x)\1\{Y_i > \hat{Q}_{X_i}(\tau_n)\}\frac{n}{k}\left( \frac{\eta}{1 - \eta} \right)
      \eqqcolon S(x).
    \end{align*}
    By Lemma~\ref{lem:denom}, the term $S(x) \stackrel{\P}{\to}0$.
  \end{proof}

  \begin{lemma}[Term~(\ref{eq:g-t5}) of $g_n$]\label{lem:g-5}
    It holds that
    \begin{align*}
      \sum_{i = 1}^{n} w_i(x) \1_{i, \delta_n}(x)\1\{Y_i > \hat{Q}_{X_i}(\tau_n)\}\frac{n}{k}\left|
        \frac{1}{1+Z_i/\hat{Q}_x(\tau_n)}-\frac{\hat{Q}_{X_i}(\tau_n)}{Y_i}\right| \stackrel{\P}{\to}0.
    \end{align*}
    
  \end{lemma}

  \begin{proof}
    Fix $\eta > 0$ and let $i$ be an observation satisfying  $\norm{X_i - x}_2 < \delta_n$ and $Z_i > 0$.
    For $n$ large enough,
    using Lemma~\ref{lem:flucts-1}, we can make $|\log(Q_x(\tau_n)) - \log(Q_{X_i}(\tau_n))| \leq C_n \delta_n$ arbitrarily small.
    Moreover, using the mean value theorem, it holds that $|x - 1| \leq 2 |\log(x)|$ when $x$ is sufficiently small. Therefore, on the event $A_n(\eta)$ defined in~\eqref{eq:event}, we can use the following upper bound, 
    \begin{align*}
      \bigg|&
      \frac{1}{1+Z_i/\hat{Q}_x(\tau_n)}-\frac{\hat{Q}_{X_i}(\tau_n)}{Y_i}\bigg|
      =
      \left|\frac{\hat{Q}_x(\tau_n)Y_i - \hat{Q}_{X_i}(\tau_n)\hat{Q}_x(\tau_n) - \hat{Q}_{X_i}(\tau_n)Z_i}{Y_i(\hat{Q}_x(\tau_n) + Z_i)} \right|
      \\
      = &\
      \left| \frac{Z_i\left( \hat{Q}_x(\tau_n) - \hat{Q}_{X_i}(\tau_n) \right)}{Y_i(\hat{Q}_x(\tau_n) + Z_i)}\right|
      \leq
      \left|1 - \frac{\hat{Q}_x(\tau_n)}{\hat{Q}_{X_i}(\tau_n)}\right|
      \leq 2 \left| \log\left( \frac{{Q}_x(\tau_n)}{{Q}_{X_i}(\tau_n)} \right)\right| + 2 \log\left( \frac{1 + \eta}{1 - \eta} \right)
      \\
      \leq &\
      2 C_n \norm{X_i - x}_2 + 2 \log\left( \frac{1 + \eta}{1 - \eta} \right).
    \end{align*}
    We can then split the term as follows,
    \begin{align*}
      0
      \leq &\
      \sum_{i = 1}^{n} w_i(x) \1_{i, \delta_n}(x)\1\{Y_i > \hat{Q}_{X_i}(\tau_n)\}\frac{n}{k}\left|
      \frac{1}{1+Z_i/\hat{Q}_x(\tau_n)}-\frac{\hat{Q}_{X_i}(\tau_n)}{Y_i}\right|
      \\
      \leq &\
      2 \sum_{i = 1}^{n} w_i(x) \1_{i, \delta_n}(x)\1\{Y_i > \hat{Q}_{X_i}(\tau_n)\} \frac{n}{k}C_n \norm{X_i - x}_2
      \\
      &+
      2 \sum_{i = 1}^{n} w_i(x) \1_{i, \delta_n}(x)\1\{Y_i > \hat{Q}_{X_i}(\tau_n)\} \frac{n}{k}\log\left( \frac{1 + \eta}{1 - \eta} \right)
      \\
      \eqqcolon &\ S_1(x) + S_2(x).
    \end{align*}
    The term $S_1(x) \stackrel{\P}{\to}0$ by Lemma~\ref{lem:f-2}.
    The term $S_2(x)\stackrel{\P}{\to}0$ by Lemma~\ref{lem:denom}.
  \end{proof}

  \subsection{Other results}
  
  \begin{lemma}[Quantile function is Lipschitz and eventually unbounded uniformly]\label{lem:flucts-1}
    Suppose Assumptions~1 and~3 from the main text hold. %
    Then, the quantile function $x \mapsto Q_x(\tau_n)$ has bounded fluctuations, that is,
    there exists a sequence $C_n > 0$ such that
    for all $x, y \in \mathcal{X}$ satisfies
    \begin{align}
      \left|\log(Q_x(\tau_n)) - \log(Q_y(\tau_n))\right| \leq C_n \norm{x - y}_2,
    \end{align}
  where $C_n \coloneqq  \log(n/k) (L_\xi + L_\alpha) + L_{c}$ and $L_\xi$, $L_\alpha$ and $L_{c}$ are the Lipschitz constants, see Assumption~3. %

  Moreover, the quantile function $Q_x$ is eventually uniformly unbounded, that is,
  \[  \inf_{x \in \mathcal{X}} Q_x(\tau_n)
  \to \infty, \qquad \text{as } n \to \infty. \]

  \end{lemma}

  \begin{proof}
   Suppose that for all $\tau \in (0, 1)$ it holds
    \begin{align*}
      Q_x(\tau) = (1-\tau)^{-\xi(x)}\ell_x\left( (1-\tau)^{-1} \right),
    \end{align*}
    where the slowly varying function $\ell_x: (0, 1) \to \R$ is normalized \citep[see][]{bingham1989} as in~(3.5). %
    Define $\tau_n \coloneqq 1 - k/n$, and note that
    \begin{align}
      Q_x(\tau_n) 
      =&\ (k/n)^{-\xi(x)} \ell_x(n/k)
      = (k/n)^{-\xi(x)} c(x) \exp\left\{ \int_{1}^{n/k}\frac{\alpha_x(t)}{t}\mathrm{d}t \right\}.
    \end{align}
    Therefore,
    \begin{align}\label{eq:ratio-0}
    \begin{split}
      \frac{Q_x(\tau_n)}{Q_y(\tau_n)} 
      = &\ \left( \frac{n}{k} \right)^{\xi(x) - \xi(y)} \frac{c(x)}{c(y)}
      \exp\left\{ \int_1^{n/k} \frac{\alpha_x(t) - \alpha_y(t)}{t} \mathrm{d}t \right\},
    \end{split}
    \end{align}
    and so
    \begin{align}\label{eq:ratio}
    \begin{split}
      |\log&(Q_x(\tau_n)) - \log(Q_y(\tau_n))|\\
        \leq &\ \log(n/k)\left| \xi(y) - \xi(x) \right|
        + |\log(c(x)) - \log(c(y))|
        +  \int_1^{n/k} \frac{|\alpha_x(t) - \alpha_y(t)|}{t} \mathrm{d}t.
    \end{split}
    \end{align}
    Recall from Assumption~3 %
    in the main text that $\xi(x)$, $c(x)$ and $\alpha_x(t)$, for every $t \geq 1$, are Lipschitz.
    Therefore, from~\eqref{eq:ratio} we have that
    \begin{align}
      \left|\log(Q_x(\tau_n)) - \log(Q_y(\tau_n))\right|
        \leq &\ \left( \log(n/k) L_\xi + L_{c} + L_\alpha \log(n/k) \right) \norm{x - y}_2.
    \end{align}

    For the second part, let $x \in \mathcal X$ and $\varepsilon_x > 0$. For every $y\in B_{\varepsilon_x}(x)$, the open ball around $x$ with radius $\varepsilon_x$, note that by the Lipschitz property of the quantile function we have for some small $\delta >0$
    \begin{align*} 
      \log(Q_y(\tau_n)) &\geq  \log(Q_x(\tau_n)) - \varepsilon_x C_n \\
      & \geq  \log(n/k)\left[(\xi(x)-\delta) - \varepsilon_x (L_\xi + L_\alpha + L_{c}/\log(n/k))  \right]\\
      & \geq \log(n/k)\frac{(\xi(x)-\delta)}{2}, 
    \end{align*}
    for $n$ large enough, and where we chose $\varepsilon_x < \{\xi(x)-\delta\} / \{2(L_\xi + L_\alpha)\}$ in the last inequality. Therefore, we have
    \[ \inf_{y \in B_{\varepsilon_x}(x)} Q_y(\tau_n)
    \to \infty, \qquad \text{as } n \to \infty. \] 
    This yields an open cover of the predictor space
    \[ \mathcal X \subseteq \bigcup_{x \in \mathcal X} B_{\varepsilon_x}(x). \] 
    Since $\mathcal X$ is compact, there exists $x_1,\dots, x_K \in \mathcal X$ that form a finite subcover 
    \[ \mathcal X \subseteq \bigcup_{j=1}^K B_{\varepsilon_{x_j}}(x_j). \] 
    Consequently, we obtain a uniform lower bound on the quantile function by
    \[ \inf_{y \in \mathcal X} Q_y(\tau_n)
    \geq \log(n/k)\frac{(\min_{j=1}^K \xi(x_j)-\delta)}{2}, \] 
    which yields the assertion since $\xi(x) > 0$ for all $x\in \mathcal X$.
  \end{proof}
  
  \begin{cor}[Rate of convergence of $C_n$ relative to leaf's diameter]\label{cor:flucts-2}
    Suppose that the Assumptions of Lemma~\ref{lem:flucts-1} and Equation~(3.6) %
    hold.
    Then, the fluctuation constant $C_n$ of the quantile function satisfies
    \begin{align}
      C_n^a\ s^{-b} \to 0,\ \text{as}\ n \to \infty,
    \end{align}
  for any $a, b > 0$.
  \end{cor}

  \begin{proof}
    From~(3.6), %
    we have that $k = n^{\beta_k}$, and $s = n^{\beta_s}$ with $0 < \beta_s < \beta_k < 1$.
    It follows that
    \begin{align}
      C_n^a s^{-b} = \mathcal{O}\left( \log(n/k)^a  n^{-\beta_s b} \right) = \mathcal{O}\left(\log( n)^a n^{-\beta_s b} \right).
    \end{align}
    
  \end{proof}

  \begin{lemma}[Logarithm bound]\label{lem:log-bound}
    Let $b = 1, \dots, B$ be a tree of the forest. Let $\eta > 0$ and consider the event $A_n(\eta)$ defined in~\eqref{eq:event}.
    Then, for all observations satisfying $Z_i > 0$ and $w_{i, b}(x) > 0$, on the event $A_n(\eta)$ it holds
    \begin{align*}
      \Bigg|\log\left( 1 + Z_i/ \hat{Q}_x(\tau_n) \right)-\log\left( \frac{Y_i}{Q_x(\tau_n)} \right) \Bigg|
      \leq &\ 
      C_n \norm{X_i - x}_2
      + \log\left(\frac{1+\eta}{(1 - \eta)^2}  \right).
    \end{align*}
    where $C_n > 0$ is the sequence defined in Lemma~\ref{lem:flucts-1}.
  \end{lemma}
  \begin{proof}
    Fix a tree $b = 1, \dots, B$, fix an observation satisfying $Z_i > 0$ and $w_{i, b}(x) > 0$, and fix $0 < \eta < 1$.
    Notice that
    \begin{align*}\label{eq:ineq}
       \Bigg| & \log\left( 1 + Z_i / \hat{Q}_x(\tau_n) \right) -  \log\left( \frac{Y_i}{{Q}_x(\tau_n)} \right) \Bigg|\\
      \leq&\ 
       \left | \log\left( 1 + Z_i / \hat{Q}_x(\tau_n) \right) -  \log\left( \frac{Y_i}{\hat{Q}_x(\tau_n)} \right)\right |  + \left | \log\left( \frac{{Q}_x(\tau_n)}{\hat{Q}_x(\tau_n)} \right) \right |. \numberthis
    \end{align*}  
    On the event $A_n(\eta)$, recall that $\hat{Q}_{X_i}(\tau_n) > 1 - \eta$ and $\hat{Q}_{x}(\tau_n) > 1 - \eta$.    

    We bound the first term in~\eqref{eq:ineq}. We have that
    \begin{align*}
      \Bigg|\log&\left( 1 + Z_i / \hat{Q}_x(\tau_n) \right) - \log\left( \frac{Y_i}{\hat{Q}_x(\tau_n)} \right)\Bigg|
      = 
      \left| \log\left( 1 + \frac{Y_i - \hat{Q}_{X_i}(\tau_n)}{\hat{Q}_x(\tau_n)} \right) - \log\left( \frac{Y_i}{\hat{Q}_x(\tau_n)} \right)\right| \\
      =&\ 
       \left| \log\left( \frac{\hat Q_x(\tau_n) - \hat{Q}_{X_i}(\tau_n) + Y_i}{Y_i} \right)\right|
      = 
      \left| \log\left( \left[\frac{\hat Q_x(\tau_n)}{\hat{Q}_{X_i}(\tau_n)} - 1\right]\frac{\hat{Q}_{X_i}(\tau_n)}{Y_i} + 1 \right)\right| \\
      \leq&\
      \left|\log\left(\frac{\hat{Q}_{x}(\tau_n)}{\hat{Q}_{X_i}(\tau_n)} \right) \right|,
    \end{align*}
    since $|\log((t-1) x + 1)| \leq |\log(t)|$ for $x \in (0,1)$ and $t>0$, and since ${\hat{Q}_{X_i}(\tau_n)}/{Y_i}\in (0,1)$ for $Z_i > 0$.
    On the event $A_n(\eta)$, it holds that
    \begin{align*}
      \left|\log\left( \frac{\hat{Q}_x(\tau_n)}{\hat{Q}_{X_i}(\tau_n)} \right)\right|
      \leq &\  \left|\log\left(\frac{Q_x(\tau_n)}{Q_{X_i}(\tau_n)}  \right)\right|
      + \log\left( \frac{1+\eta}{1-\eta} \right)\\
      \leq &\  C_n \norm{X_i - x}_2 + \log\left( \frac{1+\eta}{1-\eta} \right).
    \end{align*}
    where in the last inequality we used Lemma~\ref{lem:flucts-1}.

    We now bound the second term in~\eqref{eq:ineq}. On the event $A_n(\eta)$, it holds that
    \begin{align*}
      \left|\log\left( \frac{Q_x(\tau_n)}{\hat{Q}_x(\tau_n)} \right)\right|
      \leq &\ 
      \left|  \log\left( \frac{Q_x(\tau_n)}{Q_x(\tau_n)(1 - \eta)} \right)  \right|
      = 
      \left|  \log\left(1-\eta \right)  \right|.
    \end{align*}

    Putting everything together, we have that
    \begin{align*}
      \Bigg|\log\left( 1 + Z_i/ \hat{Q}_x(\tau_n) \right)-\log\left( \frac{Y_i}{Q_x(\tau_n)} \right) \Bigg|
      \leq &\ 
      C_n \norm{X_i - x}_2
      + \log\left(\frac{1+\eta}{(1 - \eta)^2}  \right).
    \end{align*}
  \end{proof}

  \begin{lemma}[Uniform bound on regular varying tails]\label{lem:reg-var-1}
    Let $\eta, \zeta > 0$, and define $\xi_{+} \coloneqq \max\{\xi(x): x \in \mathcal{X}\}$ and $\xi_{-} \coloneqq \min\{\xi(x): x \in \mathcal{X}\}$. 
    Then, there exists a sample size $n_0$ such that for all $n > n_0$ it holds
    \begin{align*}
      1
      < \sup_{x \in \mathcal{X}}\frac{n}{k}\P&\left( Y > Q_X(\tau_n)(1 - \eta)\mid X = x \right) 
      < (1-\eta)^{-1/\xi_{-} - \zeta},
      \\
      (1+\eta)^{-1/\xi_{-} - \zeta}
      < \sup_{x \in \mathcal{X}} \frac{n}{k}\P&\left( Y > Q_X(\tau_n)(1 + \eta)\mid X = x \right) 
      < 1.
    \end{align*}
  \end{lemma}

\begin{proof}

  Fix $\eta, \zeta > 0$.
  From Assumption~3 in the main text, %
  there exists a sample size $n_0$ such that for all $n > n_0$ it holds
  \begin{align*}
    \sup_{x \in \mathcal{X}}\sup_{t \geq Q_x(\tau_n)(1 - \eta)} |\tilde\alpha_x(t)| < \zeta.
  \end{align*}
    For the first result, using the regular variation of the tail, we observe that for any $x\in\mathcal X$ we have
    \begin{align*}
      \frac{n}{k}\P\left( Y > Q_X(\tau_n)(1 - \eta)\mid X = x \right) &=  \frac{\P\left( Y > Q_X(\tau_n)(1 - \eta)\mid X = x \right)}{\P\left( Y > Q_X(\tau_n)\mid X = x \right)}\\
      & = (1-\eta)^{-1/\xi(x)} \exp\left\{ - \int_{Q_x(\tau_n)(1-\eta)}^{Q_x(\tau_n)} \frac{\tilde\alpha_x(t)}{t}\mathrm{d}t \right\}\\
      & < (1-\eta)^{-1/\xi(x)} \exp\{ -\zeta \log(1-\eta) \}\\
      & \leq (1-\eta)^{-1/\xi_{-} - \zeta}.
    \end{align*}
    The lower bound is trivial.

    Similarly, for the second result, we observe that for any $x\in\mathcal X$ we have
    \begin{align*}
      \frac{n}{k}\P\left( Y > Q_X(\tau_n)(1 + \eta)\mid X = x \right) &=  \frac{\P\left( Y > Q_X(\tau_n)(1 + \eta)\mid X = x \right)}{\P\left( Y > Q_X(\tau_n)\mid X = x \right)}\\
      & = (1+\eta)^{-1/\xi(x)} \exp\left\{ \int_{Q_x(\tau_n)}^{Q_x(\tau_n)(1+\eta)} \frac{\tilde\alpha_x(t)}{t}\mathrm{d}t \right\}\\
      & > (1+\eta)^{-1/\xi(x)} \exp\{ -\zeta \log(1+\eta) \}\\
      & \geq (1+\eta)^{-1/\xi_{-} - \zeta}.
    \end{align*}
    The upper bound is trivial.

  \end{proof}
  
  \subsection[Proof strategy when shape parameter is negative]{Proof strategy when $\xi(x) \in(-1, 0)$}\label{app:neg-shape}
  When $\xi(x) \in(-1, 0)$, the proof strategy of Theorem~1 %
  must be adapted (notice that $\xi(x) > -1$ is necessary to ensure consistency even in the i.i.d. setting).
  The structure of the proof would follow \citep[][Appendix~A, proof of Theorem~2.1]{zhou2009}.
  The first difference, compared to the proof of Theorem~1, %
  is to define the approximate solution 
  \begin{align}\label{eq:approx-sol-2}
    t^{(\delta)}_{x} \coloneqq - \frac{(1 + \delta)}{Q_x(1) - Q_x(\tau_n)}, 
  \end{align}
  for a fixed $\delta \in (-1/2, 0)$, where $Q_x(1)$ denotes the finite upper endpoint. 
  Unlike the approximate solution in~\eqref{eq:approx-sol} for the case $\xi(x) > 0$, here $t^{(\delta)}_{x}$ is not an estimator since it depends on the population quantities $Q_x(1)$ and $Q_x(\tau_n)$.
  Following~\cite{zhou2009}, the second main difference, compared to the proof of Theorem~1, %
  is to show for all $\delta \in (-1/2, 0)$ that
  \begin{align}
    f_n(t^{(\delta)}_{x}) &\stackrel{\P}{\to} 1 + \int_0^1 \log\left( (1 + \delta) u^{-\xi(x)} - \delta \right) \d u \label{eq:fn-def-2}\\
    g_n(t^{(\delta)}_{x}) &\stackrel{\P}{\to} \int_0^1 \frac{1}{(1 + \delta) u^{-\xi(x)} - \delta} \ \d u, \label{eq:gn-def-2}
  \end{align}
  where $f_n$ and $g_n$ are defined in~\eqref{eq:fn-def} and~\eqref{eq:gn-def}, respectively.
  To establish~\eqref{eq:fn-def-2} and~\eqref{eq:gn-def-2} one would need to adapt the bounds from Propositions~\ref{prop:hill} and~\ref{prop:gn} and use the fact that $Q_x(1) - Q_x(U_i)$ is regularly varying at 1 with index $\xi(x)$, for every $x \in \mathcal{X}$ and $i = 1, \dots, n$.

  \section{Weight Function Estimation}\label{app:sim_weights}

  In quantile regression tasks, the weight function $(x, y) \mapsto w_n(x, y)$ estimated by GRF measures the similarity between $x$ and $y$ according to their conditional distribution.

  Figure~\ref{fig:weight} shows the localizing weights $w_n(x, X_{i})$, $x, X_i \in \R^p$, for two test predictors $x$ with $x_1 = -0.2, 0.5$, respectively.
  The data is generated according to Example~1, %
  with $n = 2000$ observations and $p = 40$ predictors.
  In the left panel of Figure~\ref{fig:weight},
  the observations $(X_i, Y_i)$ with $X_{i1} < 0$ are the ones influencing most the test predictor $x$ with $x_1 = -0.2$. This is because they share the same conditional distribution. A similar argument holds for the right panel of Figure~\ref{fig:weight}.
  \begin{figure}[H]
    \centering
    \includegraphics[scale=.8]{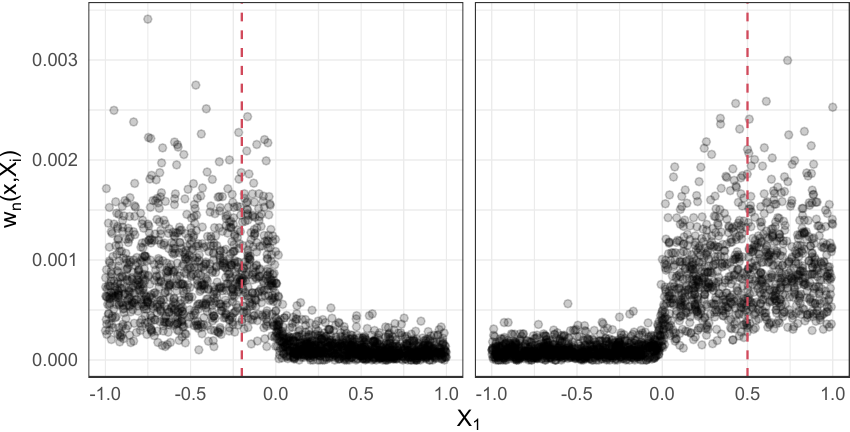}
    \caption{The height of the points represents the localizing weights $w_n(x, X_i)$ between a test predictor $x\in\R^p$ and each training observation $X_i\in\R^p$. The dashed line indicates the first coordinate of the test predictor values.}
    \label{fig:weight}
  \end{figure}

\section{Hyperparameter Tuning}\label{sec:cv}
Generalized random forests have several tuning parameters, such as the number of predictors selected at each split and the minimum node size.
This section presents a cross-validation scheme to tune such hyperparameters within our algorithm.
For large values of $\tau \approx 1$, the quantile loss is not a reliable scoring function since there might be few or no test observations above this level.
In our case, we can instead rely on the tail approximation in~(3.1) %
and use the deviance of the GPD as a reasonable metric for cross-validation.
Let $\NN_1, \dots, \NN_M$ be a random partitioning of $\{1, \dots, n\}$ into $M$ equally sized folds of the training data. %
For a sequence $\alpha_1, \dots, \alpha_J$ of tuning parameters,
we fit an \texttt{erf} object on the training set $(\g X_i, Y_i)$, $i \notin \NN_m$,
for each $\alpha_j$ and each fold $m$ as described in the \textsc{ERF-Fit} function in Algorithm~1. %
Given the fitted \texttt{erf} object, we estimate the GPD parameter vector $\hat\theta_m(X_i; \alpha_j)$ on the validation set $(\g X_i, Y_i)$, $i \in \NN_m$, as in the \textsc{ERF-Predict} function in Algorithm~1, %
and evaluate the cross-validation error by
\begin{equation}\label{eq:cv}
  CV(\alpha_j) = \sum_{m = 1}^{M}  \sum_{i \in \NN_m} \ell_{\hat\theta_m(X_i; \alpha_j)}(Z_i) 1\{Z_i > 0\},
\end{equation}
where $\theta \mapsto \ell_\theta(z)$ is the deviance of the GPD and $Z_i := (Y_i - \hat Q_{X_i}(\tau_n))_+$ are the exceedances.
Finally, we select the optimal tuning parameter $\alpha^*$ as the minimizer of $ CV(\alpha_j)$, $j = 1, \dots, J$.
To make the problem computationally tractable, we first fit the intermediate quantile function $x \mapsto \hat Q_{x}(\tau_n)$ on the entire data set. Then, on each fold, we estimate the similarity weight function $(x, y) \mapsto w_n(x, y)$ with ``small'' forests made of 50 trees.   We repeat the cross-validation scheme several times to reduce the variability of the results.

Even though, in principle, one could perform cross-validation on several tuning parameters, we find that the minimum node size $\kappa \in \mathbb{N}$ plays the most critical role for ERF.
The reason is that $\kappa$ controls the model complexity of the individual trees in the forest and consequently of the similarity weights $w_n(\cdot, \cdot)$.
Small (large) values of $\kappa$ correspond to trees with few (many) observations in each leaf and produce strongly (weakly)
localized weight functions $w_n(\cdot, \cdot)$.
The estimates of the shape parameter $\hat\xi(x)$ in~(3.4) %
may be sensitive to small changes of the localizing weights in the covariate space, leading to unstable quantile predictions through~(2.5). %
To reduce the variance of $\hat\xi(x)$, it is helpful to stabilize the log-likelihood $x \mapsto L_n(\theta; x)$ by estimating the similarity weights $w_n(\cdot, \cdot)$ with a forest made of trees with relatively large leaves. Notice that $w_n(x, y)$ influences the effective number of observations used in the weighted (negative) log-likelihood $L_n(\theta; x)$ equation~(3.3). %

Figure~\ref{fig:cv-works} shows numerical results of cross-validating the minimum node size $\kappa$ for the model described in Example~1. %
Here, we perform 5-fold cross-validation repeated three times by growing forests of 50 trees on each fold.
We measure the performance as the square root of the mean integrated squared error (MISE) between the estimated and the true quantile function over 50 simulations; see Section~4 %
for the definition of the MISE. We observe that the cross-validated performance of ERF (dashed line) is close to the minimum square root MISE, suggesting that the proposed cross-validation scheme works well.
\begin{figure}[!tb]
  \centering
  \includegraphics[scale=1]{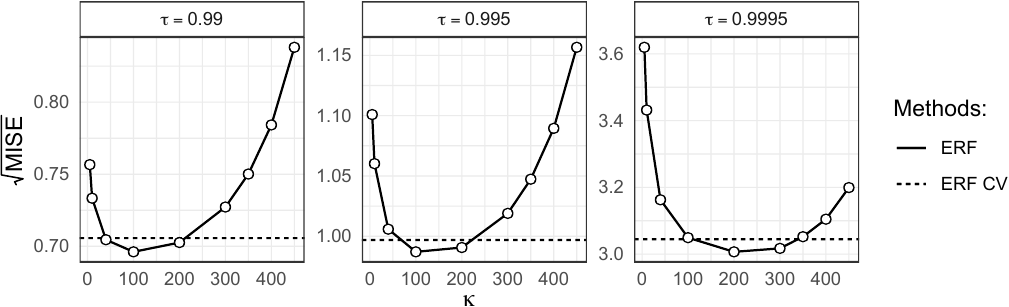}
  \caption{Solid line shows the square root of the MISE of ERF for different minimum node sizes $\kappa$ over 50 simulations. The dashed line shows the square root MISE of the cross-validated ERF. The data is generated according to Example~1. %
  }
  \label{fig:cv-works}
\end{figure}

  \section{Additional Material for Simulation Study}

  \subsection{Experiment 3}\label{app:experiment3}

  In this section, we consider more complex regression functions depending on more signal variables both in the scale and shape parameters. While the predictor variables $\g X$ are uniform distributed on $[-1,1]^p$ with $p=10$, the conditional response follows three different models
  \begin{align*}
    (Y \mid \g X = \g x) \sim s_j(\g x)T_{\nu(x)}, \quad j = 1, 2, 3,
  \end{align*}
  where we allow both degrees of freedom $\nu(x)$ and the scale $s_j(x)$ of the Student's $t$ distribution to depend on the predictors.
  In particular, we model the degrees of freedom as a decreasing function of the first predictor as $\nu(\g x) =  3 [2 + \tanh(-2x_1)]$, and the different scale functions as
  \begin{align*}
    \begin{split}
      & s_1(\g x) = [2 + \tanh(2 x_1)] (1 +  x_2 / 2),\\
      & s_2(\g x) = 4 - (x_1^2 + 2 x_2^2),\\
      & s_3(\g x) = 1 + 2 \pi \varphi(2x_1, 2x_2),
    \end{split}
  \end{align*}
  where $\varphi$ denotes a centered bivariate Gaussian density with unit variance and correlation coefficient equal to $0.75$.
  The first scale function $s_1(x)$ is non-linear with respect to the first predictor and contains an interaction effect between the first two predictors. The function $s_2(x)$ is quadratic and decreasing in the first two dimensions. The third scale function $s_3(x)$ is non-linear in the first two predictors and contains an interaction effect. The sample size is $n=5000$.

  In this experiment we compare ERF, GRF, GBEX, EGP Tail and the unconditional method. We leave out EGAM because we observed it performs poorly in the scenarios considered here.
  Figure~\ref{fig:boxplots-complex} shows the boxplots of $\rise$ over $m = 50$ simulations over different models, methods, and quantile levels. For better visualization, we remove large outliers of GRF, QRF, and EGP Tail.
  We observe that ERF and GBEX generally outperform the other methods over all models and quantile levels, where GBEX has a slight advantage in high quantiles for Models~2 and~3. GRF and QRF seem to deteriorate completely for very large quantiles.
  \begin{figure}[!ht]
    \centering
    \hspace*{-.5in}
    \includegraphics[scale=1]{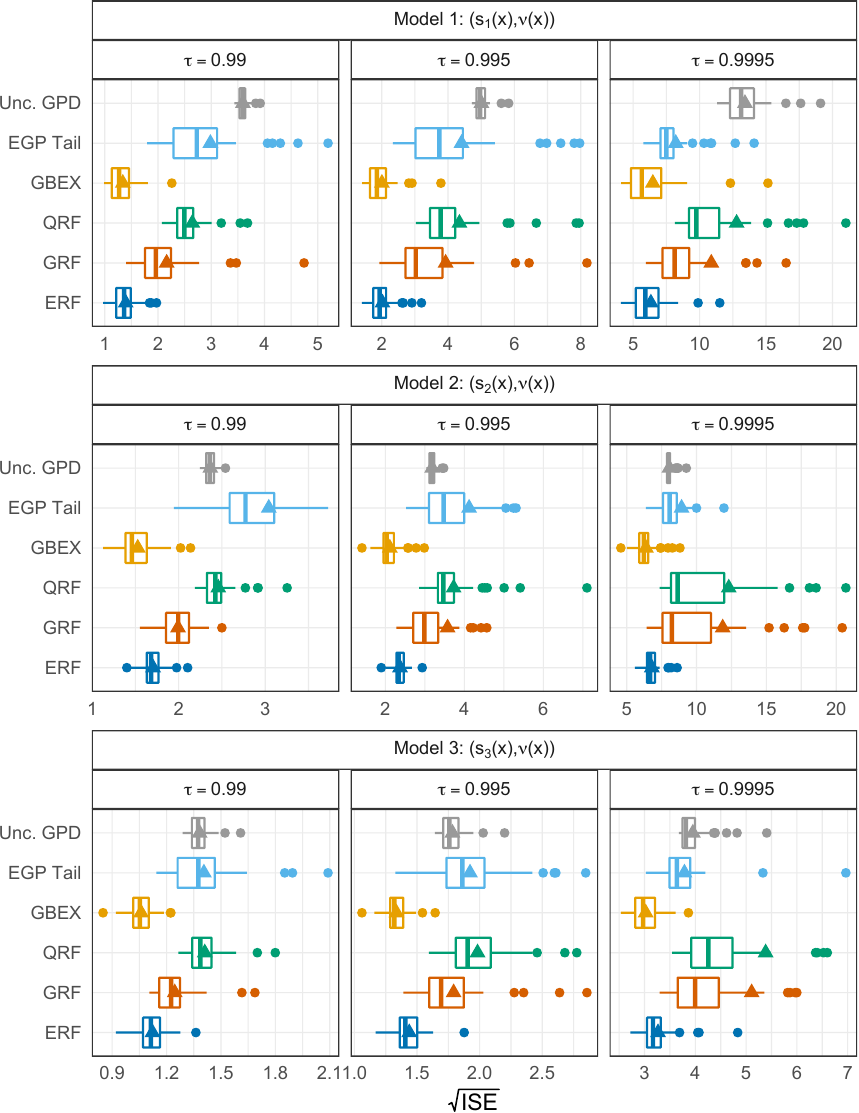}
    \caption{Boxplots of $\rise$ over $m=50$ simulations for different generative models (rows) and quantile levels (columns). The predictor space dimension is set to $p = 10$. Triangles represent the average values.}
    \label{fig:boxplots-complex}
  \end{figure}

  \subsection{Sensitivity of ERF and two Alternative Random Forest Methods to the Intermediate Threshold Level}\label{app:sensitivity}

  In this section, we study the sensitivity of ERF and the two alternative random forest method for the Weissman extrapolation mentioned in Section~(3.2) %
  to different choices of the intermediate quantile level $\tau_n$.

  While ERF relies on the approximation~(2.5) %
  for extreme quantile estimation, the two alternative methods we compare are both based on the Weissman approximation
  \begin{align}\label{eq:extrap_simple}
    Q_x(\tau) \approx Q_{x}(\tau_n)\left(\frac{1 - \tau}{1 - \tau_n}\right) ^{-\xi(x)},
  \end{align}
which only requires estimation of the intermediate quantile and the shape parameter (but only works for heavy-tailed data). The first method, which we refer to as the random forest Hill estimator, uses our new localized Hill estimator introduced in~(3.8). %
The second method, suggested by a referee and referred to as the random forest shape estimator, relies on the fact that the log-transformed exceedances are approximately exponential distributions with mean $\xi(x)$, that is, approximately $\log(Y_i / \hat{Q}_{X_i}(\tau_n))_+ \sim \mathrm{Exp}(1 / \xi(X_i))$ for $n$ large enough and all $i$ with $Y_i > \hat{Q}_{X_i}(\tau_n)$. We therefore can fit a regression random forest to the mean parameter $\xi(x)$ and estimate the target quantiles using~\eqref{eq:extrap_simple}. All methods use the same intermediate quantile estimator, namely a quantile random forest.

  Figure~\ref{fig:sensitivity} shows the prediction error of ERF compared to the two Weissman-type methods as a function of the intermediate quantile level $\tau_n$ for a fixed target quantile $\tau = 0.9995$, and three data-generating processes. 
  We measure the performance as the square root of the median integrated squared error (ISE) between the estimated and the true quantile function over $m = 100$ simulations. We choose the median instead of the mean ISE to remove the effect of large outliers in the Weissman-type methods.

 When the conditional response $Y \mid X = x$ follows a Student's $t$-distribution (left panels of Figure~\ref{fig:sensitivity}), the pre-asymptotic bias of the Weissman-type methods dominates their smaller variance, compared to ERF. As a consequence, we observe that these methods are very sensitive to the choice of the intermediate quantile $\tau_n$, and in particular, it must be chosen very high to decrease the bias. In comparison, ERF does not seem to be very sensitive to the choice of $\tau_n$.
  In the less realistic case where the conditional response $Y \mid X = x$ follows exactly a Pareto distribution (right panels of Figure~\ref{fig:sensitivity}), the pre-asymptotic bias of the Weissman-type methods disappears by construction, and we can observe the effect of the variance. As expected, we see that the Weissman-type methods have a slight advantage over ERF due to their smaller variance (since they estimate one parameter instead of two). In particular, our random forest Hill estimator seems to perform well in this case. In general, we recommend using ERF since in practice, the presence of an (unknown) pre-asymptotic bias can usually not be excluded.

  \begin{figure}[ht]
    \centering
    \hspace*{-.5in}
    \includegraphics[scale=.75]{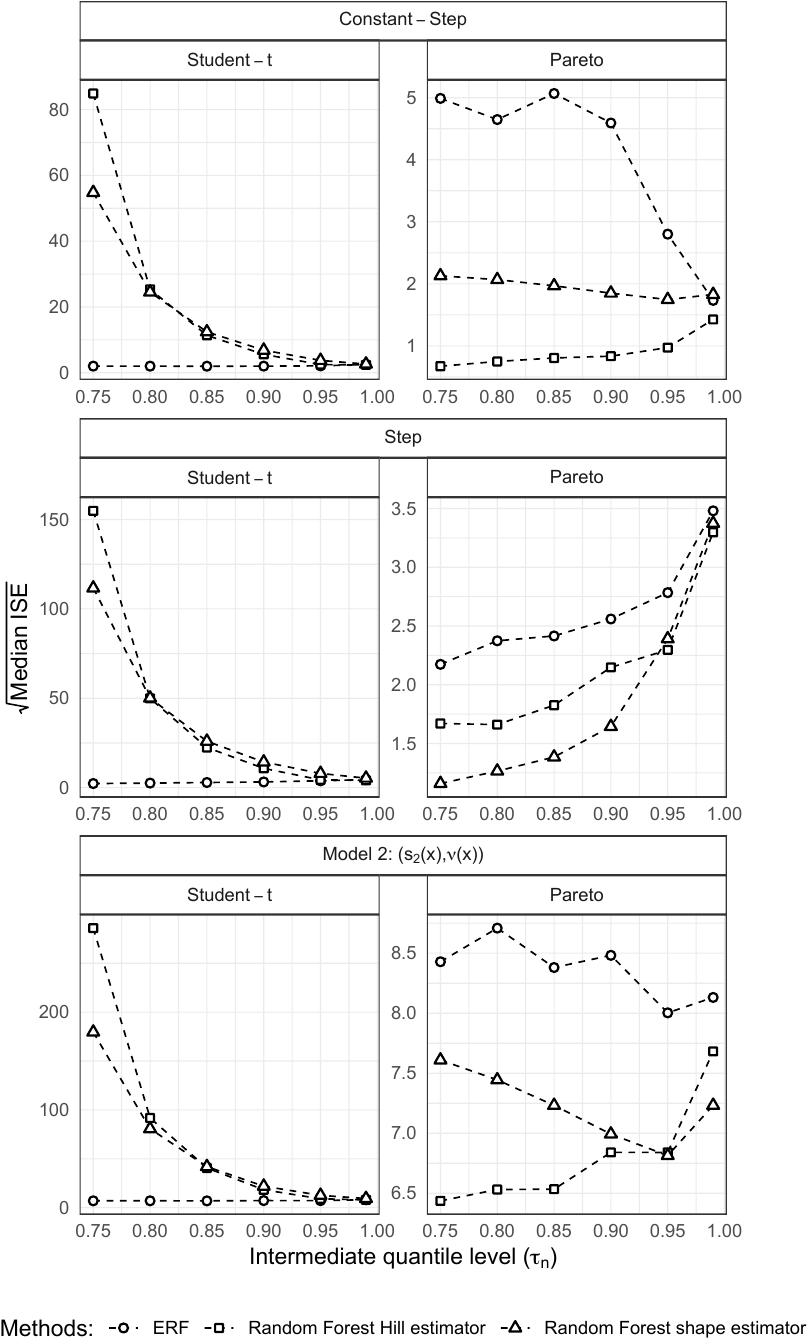}
    \caption{Square root of the median ISE for different intermediate quantile levels $\tau_n$ over $m = 100$ simulations for ERF (circles), random forest Hill estimator (squares) and random forest shape estimator (triangles).
    The target quantile level is set to $\tau = 0.9995$ and the training and test sample sizes are $n = 1000$ and $n'= 100$.
    The response variable $Y \mid X = x$ follows a Student's $t$-distribution (left) and a Pareto distribution (right) 
    with constant scale $s(x) \equiv 1$ and shape parameter $\xi(x) = 1/(4 + 8 \cdot \1\{x_2 > 0\})$ (top),
    with scale $s(x) = 1 + \1\{x_1 > 0\}$ and constant shape parameter $\xi(x) \equiv 0.25$ (middle),
    and with scale $s_2(\g x) = 4 - [x_1^2 + 2 x_2^2]$ and shape parameter $\xi(\g x) =  1 / [6 + 3\tanh(-2x_1)]$ and (bottom).
    The predictor space has $p = 2$ dimensions.}
    \label{fig:sensitivity}
  \end{figure}

  \subsection{Bias--Variance decomposition of the MISE}
  In this section, we consider again the experiments of Section~4.3 %
  where we decompose the MISE into its bias and variance terms (see Figure~\ref{fig:bias-var}).
  In the top three panels of Figure~\ref{fig:bias-var}, we fix the dimension to $p = 10$ and study the performance as the target quantile $\tau$ grows. 
  We observe that the poor performance of classical forest-based methods is mainly driven by a large variance, since there are few or no observations available at very high quantile levels.
  On the other hand, the methods that rely on extrapolation have much lower variance and bias.
  In the bottom three panels of Figure~\ref{fig:bias-var}, we fix the target quantile $\tau = 0.9995$ and study the performance as the dimension of the predictor space $p$ grows.
  We can clearly observe here that EGAM poor performance is mainly driven by its bias since the method is not designed to scale with larger dimensions.

\begin{figure}[!tb]
  \centering
  \includegraphics[scale=1]{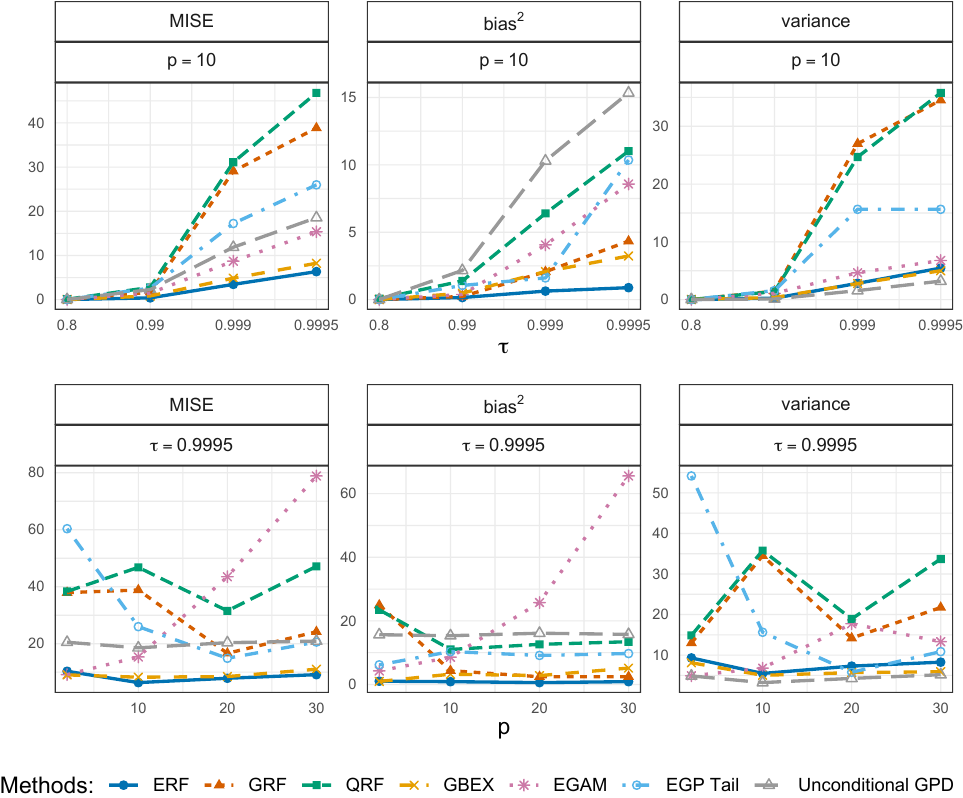}
  \caption{Square root MISE and its bias and variance decomposition for different methods against the quantile level $\tau$ in dimension $p = 10$ (top three panels), and against the model dimension $p$ for quantile levels $\tau=0.9995$ (bottom three panels).
  }
  \label{fig:bias-var}
\end{figure}	

  \section{Additional Material for U.S.~Wage Analysis}\label{app:additional-app}

  \subsection{Additional Figure}
  \label{app:additional-app-1}

  Figure~\ref{fig:erf-params-age} shows that estimated GPD parameters $\hat\theta(x)$ for the original response as a function of age for groups with less or more than 15 years of education.

  \begin{figure}[!ht]
    \centering
    \includegraphics[scale=1]{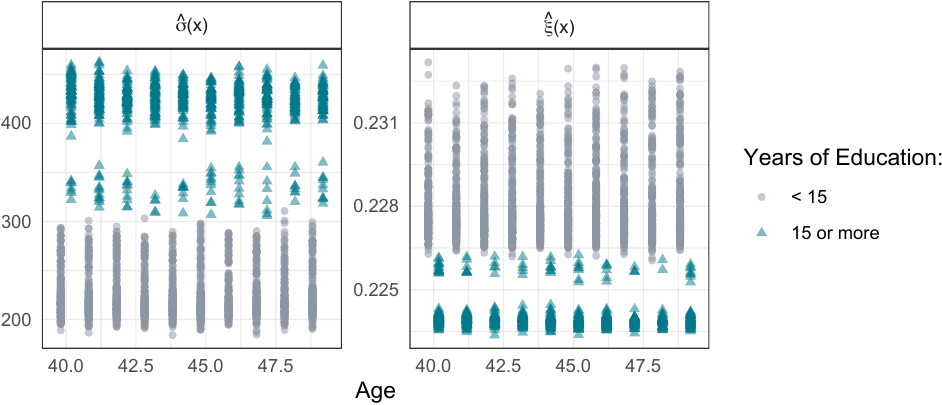}
    \caption{Estimated GPD parameters $\hat\theta(x)$ as a function of age for groups with less (circles) or more (triangles) than 15 years of education.
    }
    \label{fig:erf-params-age}
  \end{figure}

  \subsection{Analysis with Log-Transformed Response}
  \label{app:additional-app-2}

  Following \cite{angrist2009}, we consider here the natural logarithm of the wage as a response variable for quantile regression. We perform the same analysis as in Section~5 %
  again with this log-transformed response since it highlights several interesting properties of the ERF algorithm.
  Figure~\ref{fig:erf-params-log} shows the GPD parameters $\hat\theta^{\log}(x)$ estimated by ERF as a function of years of education when the response is $\log(Y)$.
  We notice that the log transformation makes the response lighter-tailed, with estimated shape parameters $\hat\xi^{\log}(x)$ fairly close to $0$. The scale parameters $\hat\sigma^{\log}(x)$ still show a certain structure, but they vary on a much smaller scale compared to $\hat\sigma(x)$ estimated on the original response; see Figure~5 %
  in the main text. These observations are consistent with theory since it is well-known that the log-transformation renders heavy-tailed data into light-tailed \citep[][Example 3.3.33]{embr2012}. Moreover, the shape parameter on the original data then essentially acts as a scale parameter in the GPD approximation of the log-transformed data, explaining the smaller variation of $\hat\sigma^{\log}(x)$.

  \begin{figure}[!t]
    \centering
    \includegraphics[scale=1]{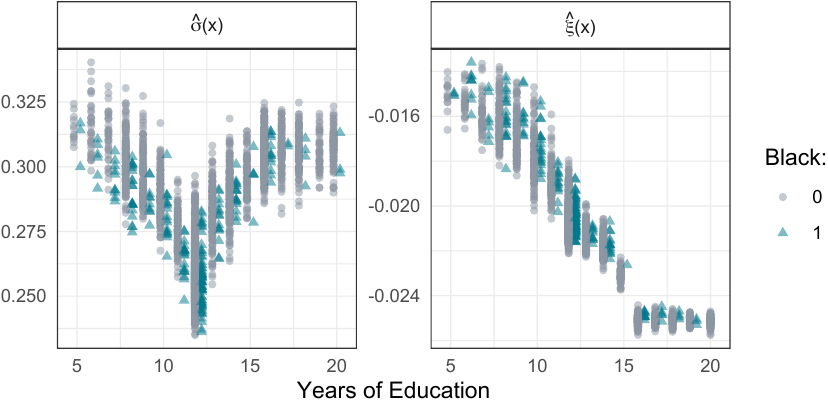}
    \caption{Estimated GPD parameters $\hat\theta(x)$ for the log-response as a function of the years of education for the black (triangles) and white (circles) subgroups.
    }
    \label{fig:erf-params-log}
  \end{figure}

  Figure~\ref{fig:erf-predicts-log} shows the (exponentiated) predicted quantiles $\exp\{\hat Q_x^{\log}(\tau)\}$ of the different methods as a function of years of education when the response is $\log(Y)$; we removed again all quantiles above 6,000 predicted by GRF. By construction, GRF is invariant to the log-transformation, while the methods based on extrapolation may produce predictions that differ from $\hat Q_x(\tau)$ in Figure~6 %
  fitted on the original data. The reason is that the approximation by the GPD is done on heavy-tailed data on the original scale and on much lighter-tailed data on the log scale. We observe in Figure~\ref{fig:erf-predicts-log} that the flexible methods ERF and GBEX have the desirable property that the predictions do not change much under marginal transformations.
  The unconditional method on the other hand seems to be sensitive to marginal transformation and works better on the log-transformed data as it captures a larger variability of the conditional quantiles even for high $\tau$.
  This is confirmed by Figure~\ref{fig:wage-cv-log} where we observe that the unconditional method has a smaller loss, especially for higher quantiles, while all other methods have a similar performance as on the original data.
  To better understand this behavior, we recall the GPD approximation~(2.5) %
  for large quantiles estimated on the original response as   
  \begin{align}\label{eq:quantiles-y}
    \hat Q_x(\tau) \approx \hat Q_x(\tau_n) + G^{-1}\left(\frac{\tau-\tau_n}{1-\tau_n}; \hat\theta(x)\right),
  \end{align}
  where $G^{-1}$ is the inverse of the distribution function~(2.2) %
  of the GPD; see Figure~6 %
  in the main text. On the other hand, first estimating the quantiles of the log-transformed data with a similar approximation and then exponentiating these estimates results in
  \begin{align}\label{eq:quantiles-logy}
    \exp\{\hat Q_x^{\log}(\tau) \} \approx \hat Q_x(\tau_n) \exp \left\{G^{-1}\left(\frac{\tau-\tau_n}{1-\tau_n}; \hat\theta^{\log}(x)\right)\right\},
  \end{align}
  where $\hat\theta^{\log}(x)$ is the parameter vector of the GPD fitted for the response $\log(Y)$; see Figure~\ref{fig:erf-predicts-log}. We note that $\hat Q_x(\tau_n)$ is the same in both approximations since it is fitted using quantile GRF, which is invariant under marginal transformations. Comparing \eqref{eq:quantiles-y} and \eqref{eq:quantiles-logy} shows that the intermediate quantiles have an additive and multiplicative influence on the extreme quantiles, respectively. This explains why using the unconditional method for the GPD with $\hat\theta^{\log}(x) \equiv \hat\theta^{\log}$ seems to work better on the log-transformed data. Indeed, the different multiplicative scalings observed for ERF and GBEX in Figure~6 %
  in the main text cannot be represented by \eqref{eq:quantiles-y} with unconditional GPD, but they can be represented by \eqref{eq:quantiles-logy} if the intermediate quantile already carries the structure.

  \begin{figure}[!ht]
    \centering
    \hspace*{-.5in}
    \includegraphics[scale=1]{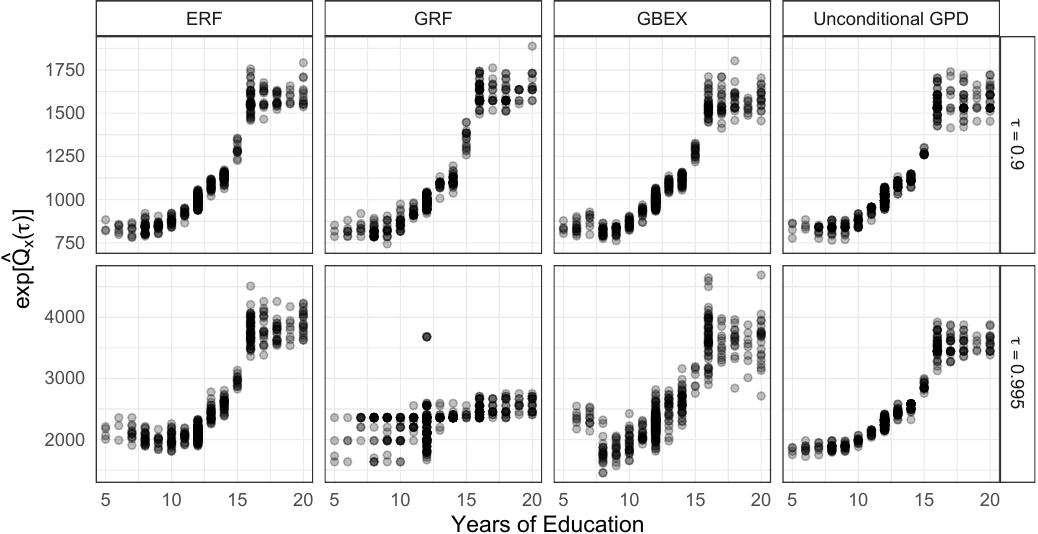}
    \caption{Predicted quantiles at levels $\tau = 0.9, 0.995$ for ERF, GRF, GBEX, and the unconditional method fitted on the log-response.
    }
    \label{fig:erf-predicts-log}
  \end{figure}

  \begin{figure}[!ht]
    \centering
    \hspace*{-0.75in}
    \includegraphics[scale=1]{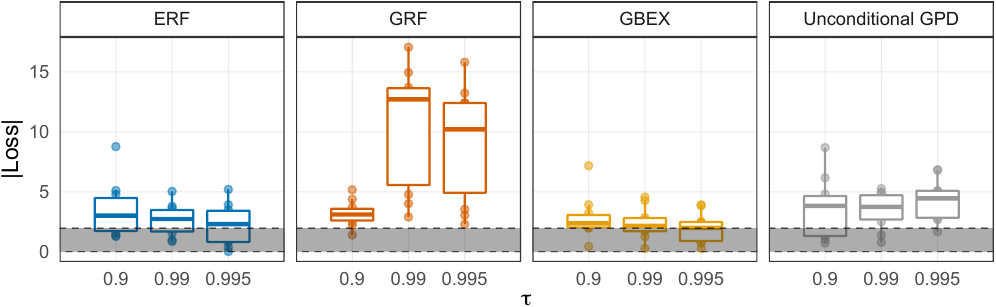}
    \caption{Absolute value of the loss~(5.1) %
    for the different methods fitted on the log-response of the U.S.~wage data. The shaded area represents the 95\% interval of the absolute value of a standard normal distribution.}
    \label{fig:wage-cv-log}
  \end{figure}

\end{appendix}

\clearpage
\bibliography{reference}
\bibliographystyle{abbrvnat}
\end{document}